\documentclass[sigconf]{acmart}

\AtBeginDocument{%
  }

\acmConference[]{SIGMOD Industrial Track}{Bengaluru, India}{2026}

\usepackage{xcolor}
\usepackage{balance}
\usepackage{colortbl}
\usepackage{adjustbox}
\usepackage{subfigure}
\usepackage{multirow}
\usepackage{pifont}
\usepackage{xspace}
\usepackage{enumitem}
\usepackage{bm}
\usepackage{tcolorbox}

\newcommand{\sys}{\textsf{LOAM}\xspace}
\newcommand{\sysprefix}{\textsf{LOAM-}}
\newcommand{\sysnative}{\textsf{LOAM-NA}\xspace}
\newcommand{\MC}{\textsf{MaxCompute}\xspace}
\newcommand{\citylife}{\textsf{Project~1}\xspace}
\newcommand{\gd}{\textsf{Project~2}\xspace}
\newcommand{\fin}{\textsf{Project~3}\xspace}
\newcommand{\ads}{\textsf{Project~4}\xspace}
\newcommand{\mkt}{\textsf{Project~5}\xspace}

\newcommand{\stitle}[1]{\noindent\underline{\textbf{#1}}}
\newcommand{\sstitle}[1]{\vspace{0.6ex}\noindent\underline{\textbf{#1}}}

\newcommand{\ie}{\emph{i.e.}\xspace}
\newcommand{\eg}{\emph{e.g.}\xspace}
\newcommand{\wrt}{\emph{w.r.t.}\xspace}

\newcommand{\squishlist}{
	\begin{list}{$\bullet$}{
		\setlength{\itemsep}{0pt}
		\setlength{\parsep}{3pt}
		\setlength{\topsep}{3pt}
		\setlength{\partopsep}{0pt}
		\setlength{\leftmargin}{1.0em}
		\setlength{\labelwidth}{1em}
		\setlength{\labelsep}{0.5em}
   }
}

\newcommand{\squishenum}{
	
	\begin{list}{\usecounter{scount}}{
		\setlength{\itemsep}{0pt}
		\setlength{\parsep}{3pt}
		\setlength{\topsep}{3pt}
		\setlength{\partopsep}{0pt}
		\setlength{\leftmargin}{1.2em}
		\setlength{\labelwidth}{1em}
		\setlength{\labelsep}{0.5em}
	}
}

\usepackage{xcolor}

\newcommand{\squishend}{
	\end{list}
}

\newcommand{\eat}[1]{}

\newtheorem{lemma}{Lemma}
\newtheorem{theorem}{Theorem}

\newtheorem{definition}{Definition}

\newtheorem{challenge}{Challenge}

\definecolor{mygrey}{RGB}{230,230,240}

\begin{document}

\title{Learned Query Optimizer in Alibaba \MC: \break Challenges, Analysis, and Solutions}

\author{
Lianggui Weng$^{\#}$, Dandan Liu$^{\#}$, Wenzhuang Zhu$^{\#}$, Rong Zhu$^{*}$, Junzheng Zheng$^{*}$, \break Bolin Ding, Zhiguo Zhang, Jingren Zhou}
\affiliation{\LARGE \institution{\textsf{Alibaba Group}\country{}}}

\renewcommand{\shortauthors}{Lianggui Weng, et al.}

\begin{abstract}
{\renewcommand\thefootnote{}\footnote{\noindent
\#: Equal Contribution; 
*: Corresponding Authors \textsf{\{red.zr, junzheng.zjz\}@alibaba-inc.com}
}}

\setcounter{footnote}{0}
\vspace{-1em}

\noindent Existing learned query optimizers remain ill-suited to modern distributed, multi-tenant data warehouses due to idealized modeling assumptions and design choices. Using Alibaba's \MC as a representative, we surface four fundamental, system-agnostic challenges for any deployable learned query optimizer: 1) highly dynamic execution environments that induce large variance in plan costs; 2) potential absence of input statistics needed for cost estimation; 3) infeasibility of conventional model refinement; and 4) uncertain benefits across different workloads. These challenges expose a deep mismatch between theoretical advances and production realities and demand a principled, deployment-first redesign of learned optimizers.

To bridge this gap, we present \sys, a one-stop learned query optimization framework for \MC. Its design principles and techniques generalize and are readily adaptable to similar systems. Architecturally, \sys introduces a statistics-free plan encoding that leverages operator semantics and historical executions to infer details about data distributions and explicitly encodes the execution environments of training queries to learn their impacts on plan costs. For online queries with unknown environments at prediction time, \sys provides a theoretical bound on the achievable performance and a practical strategy to smooth the environmental impacts on cost estimations. For system operating, \sys integrates domain adaptation techniques into training to generalize effectively to online query plans without requiring conventional refinement. Additionally, \sys includes a lightweight project selector to prioritize high-benefit deployment projects. \sys has seen up to $30\%$ CPU cost savings over \MC's native query optimizer on production workloads, which could translate to substantial real-world resource savings. \sys and example workloads are open-sourced at \url{https://drive.google.com/drive/folders/1B0tqOLo2Jt6aBFxaPfjbV0B7sagFPT5u}. 

\end{abstract}





\maketitle

\section{Introduction}\label{sec:intro}


Query optimization serves as a cornerstone of database systems and has been a long-standing research topic. Over the past decade, learned query optimizers leveraging machine learning techniques have emerged as a promising approach to either replace~\cite{neo, balsa, loger, base, glo} or refine traditional query optimizers~\cite{bao, autosteer, lero, hybrid-QO, foss} for their capabilities of better capturing the complex data distributions and providing more accurate cost estimations. 


Despite these advances, applying learned optimizers to distributed multi-tenant data warehouses, such as \textsf{Microsoft Azure Synapse Analytics}~\cite{azure_synapse_analytics}, \textsf{Google BigQuery}~\cite{bigquery}, \textsf{Amazon Redshift Serverless}~\cite{redshift}, \textsf{Alibaba MaxCompute}~\cite{mc}, and \textsf{Snowflake Serverless}~\cite{snowflake}, which serve as today's main cloud data infrastructures, remains highly problematic. This paper focuses on Alibaba's \MC, a production-grade data warehouse that handles storage and compute workloads ranging from 100 GB to the exabyte level and has been battle-tested at scale within Alibaba Group (routinely serves over 40 million queries per day). \MC embodies key architectural and operational principles shared by state-of-the-art cloud-native data warehouses, which give rise to unique challenges in deploying learned query optimizers. By studying \MC, we aim to identify systemic barriers and derive generalizable solutions for bridging the theoretical advances in learned optimizations to production-grade data warehouses.

\stitle{Challenges in \MC.} 
Existing learned query optimizers typically employ a plan explorer to generate a set of candidate execution plans for a given query, combined with a learned cost model that predicts the cost of each plan. The plan with the lowest predicted cost is then selected for execution. In \MC, however, applying this paradigm faces several key challenges. 


To begin with, cost estimation in \MC is fundamentally challenged by two system realities, which also break the core assumptions of current approaches. Compared to on-premises setups or low-capacity cloud containers typically assumed in prior work, query execution in \MC involves highly dynamic resource allocation from cluster-wide pools averaging over 5,000 machines with varying loads, leading to significant variation in plan costs~\cite{bao, neo} \textbf{(C1.~Environment variation)}. Worse still, given the sheer scale of production data and frequent data modifications, statistics (\eg, histograms of attributes) that are helpful to understand data distributions are not automatically and timely maintained in \MC by default, and are therefore often remain stale or missing (\textbf{C2. Missing helpful statistics}). Together, these factors render current learned cost models, in which cost is deterministically mapped from static plan features and input statistics, inapplicable in \MC.

In terms of system operating, to generalize the learned cost model to candidate plans produced by the plan explorer for online queries (real-time production queries), existing work refines the model using plan-cost pairs collected either offline by executing a large set of candidate plans~\cite{neo, balsa, glo} or online by deploying a partially trained model to serve user queries and periodically updating it with observed plan-cost pairs~\cite{neo, bao, lero, hybrid-QO, foss}. 
However, in \MC, production OLAP queries routinely process multiple data partitions and are often join-heavy, with CPU costs reaching up to $10^7$ or more, so executing additional candidate plans solely for refinement is prohibitively costly.
Moreover, online refinement risks generating disastrously inefficient plans that consume excessive computing resources, which is unacceptable in \MC's multi-tenant environment. Together, these factors make conventional refinement strategies for learned optimizers impractical in \MC \textbf{(C3. Infeasibility of conventional refinement)}.

Finally, beyond query optimization, \MC hosts over 100,000 user-created database instances, called \emph{projects}, each exhibiting distinct workload characteristics (\eg, join topology) and data properties (\eg, table sizes and update frequencies). This heterogeneity prevents learned optimizers from performing consistently well across projects. As a result, universally deploying learned optimizers is neither efficient nor sustainable, and it is essential to automatically identify projects that are most likely to achieve substantial performance gains (\textbf{C4. Project selection}).

Because \MC shares core design principles with other modern cloud data warehouses (see Appendix~\ref{app:mc-vs-similar-systems}~\cite{full-paper}), these challenges (C1--C4)
transfer to similar systems. This highlights a broader implication: \emph{The learned optimizers in distributed, multi-tenant data warehouses demand a fundamental rethinking of the architecture design and system operating paradigms}. 

\stitle{Our Contributions.}
This paper bridges the gap between theoretical advances in learned query optimizations and production-grade data warehouses. Guided by challenges C1--C4, we present \sys, a one-stop \underline{l}earned query \underline{o}ptimizer in \underline{A}libaba’s \textsf{\underline{M}axCompute}. \emph{It addresses the core challenges in \MC while offering a transferable design paradigm for learned optimizers in large-scale, multi-tenant data warehouses, and paves the way for more practical AI-enhanced query optimization techniques in production systems.}


\sys operates in a steering style that guides \MC's native optimizer with knobs to produce diverse candidate plans and selects the one with the lowest estimated cost (see Section~\ref{sec:framework}). Beyond this, we distill four core principles that reshape learned optimizer design for \MC-like systems and drive \sys's key innovations, setting it apart from existing approaches:

$\bullet$ \textbf{Environment-Aware Plan Cost Modeling.} 
\emph{Impacts from execution environments
must be explicitly considered when estimating a plan's cost (as per C1)}. 
Accordingly, \sys explicitly encodes runtime environmental features for each training plan (see Section~\ref{sec:model}). This enables the cost model to learn how different execution environments affect the plan cost. 
Critically, environmental information is unknown for online queries at prediction (query optimization) time. This breaks a core assumption in current learned optimizers that all feature values are observed before prediction, and introduces intrinsic errors to any query optimizers. To quantify this impact, we derive an upper bound on the best-achievable performance of query optimizers under the unobserved environment variation. In practice, \sys mitigates this error by making cost predictions under representative average-case conditions (see Section~\ref{sec:per-bound}).

$\bullet$  \textbf{Statistics-Free Plan Encoding.} \emph{Plan cost estimation must operate without 
input statistics (as per C2)}. In line with this, \sys encodes attributes of key operators to enable coarse inference of details \wrt data distributions from historical queries (see Section~\ref{sec:model}). 

$\bullet$ \textbf{Preemptive Generalization.} \emph{Generalizing the optimizer to unseen candidate plans must be preemptively handled before deployment and in a computationally efficient manner (as per C3)}. As a key feature of data warehouses, \MC preserves extensive historical data for long-term analysis. Leveraging this, \sys is trained offline on historically executed query plans and their costs. Because the feature distributions of these training plans can differ markedly from candidate plans tuned for online queries, we integrate domain adaptation techniques~\cite{domain-adaption} to learn domain-invariant intermediate representations for both training and candidate plans. This enables \sys optimized for training plans to generalize well to estimate costs for candidate plans, eliminating both the prohibitive cost of executing additional candidate plans and possible risks from online refinement (see Section~\ref{sec:model}). 

$\bullet$ \textbf{Automatic Project Selection.} Beyond query optimization, \emph{an efficient mechanism is required to automatically identify projects with potentially large deployment benefits (as per C4)}. To this end, \sys automates project selection (see Section~\ref{sec:project-selection}) through a two-stage process that first filters out projects unsuitable for training learned optimizers with rule-based heuristics and then employs a learned ranker to prioritize the remainder with high potential benefits.

We evaluate \sys on production workloads from \MC and observe up to $30\%$ performance improvement over the native query optimizer. \sys's core techniques also prove highly effective in boosting practical performance and enabling broader deployment benefits. 
Crucially, our analysis shows that about $4\%$ of \MC's projects can expect performance gains of $\ge10\%$ with \sys, 
which represents a non-trivial portion of \MC's currently vast serving landscape (more than 100,000 projects)
and can translate to substantial real-world resource savings. Notably, this estimation is conservative and can be greatly improved by incorporating more strategies to produce candidate plans.

\section{Background}
\label{sec:lqo-review-challenge}

Query optimization in \MC is challenged by both the sheer scale of its production workloads and its distributed, multi-tenant architecture. In this section, we first describe the query execution workflow in \MC (Section~\ref{subsec:lqo-challenges}), then briefly review existing learned query optimizers and explain why their assumptions and operating paradigms break down in \MC (Section~\ref{subsec:reviewlqo}).


\subsection{Query Execution in \MC}\label{subsec:lqo-challenges}

\emph{Projects}, akin to user-created database instances, are the primary organizational units in \MC. They provide logical isolation of data and basic access control. Users can query and manipulate data within a project via SQL queries. Upon submission, each query undergoes a four-phase execution process, as illustrated in Figure~\ref{fig:workflow}.

\textit{1. Query compilation and optimization.} 
Initially, a SQL statement is compiled and optimized by \MC's native query optimizer into a physical plan, represented as a tree structure where each node corresponds to a data operation such as table scanning, joining, or aggregation. \MC employs a cost-based optimizer that explores equivalent query plans using a rich set of transformation rules (\eg, join reordering, and partition pruning) to generate more efficient alternatives, with the plan of the lowest cost estimated by a cost model selected for execution. The cost model relies heavily on up-to-date input statistics, such as top‑$N$ values, the number of distinct values (NDVs), and frequency histograms of attributes. However, due to the sheer data scale (terabytes, even to exabytes) and frequent updates, \MC does not automatically maintain these statistics by default to prioritize system efficiency, unless users explicitly enable them. In their absence, the effectiveness of the optimizer is highly constrained. Cost estimation must fall back to coarse, metadata-driven approximations such as based on historical table row counts, which often lead to unreliable plan selection. Furthermore, many transformation rules, such as join reordering, that rely on accurate input statistics become disabled, increasing the likelihood of missing better plans.

\begin{figure}[t]
\centering
\includegraphics[width=0.8\linewidth]{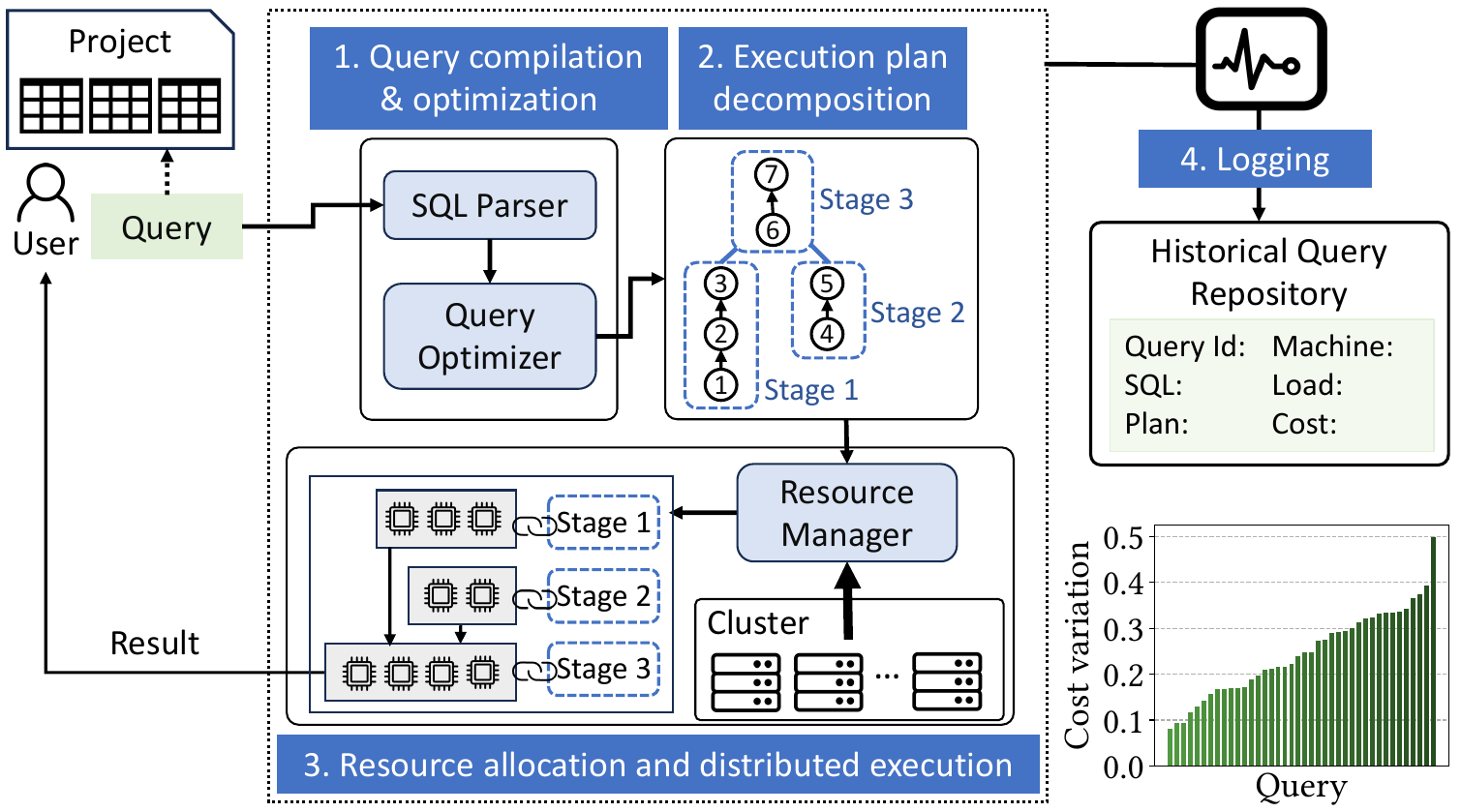}
\vspace{-0.5em}
\caption{Query execution workflow in \MC.}
\label{fig:workflow}
\vspace{-0.8em}
\end{figure}





\textit{2. Execution plan decomposition.} 
The physical query plan is then decomposed into a tree of \emph{stages} at operators requiring data reshuffling across machines, such as hash joins and grouping aggregations that demand co-partitioned inputs on the join or group-by key. Each stage consists of a sequence of connected operators that are executed as an intra-machine pipeline, and edges in the tree correspond to the data dependencies between stages.

\textit{3. Resource allocation and distributed execution.} \MC manages query execution at the stage level, treating each stage as the atomic unit for resource allocation and scheduling. Once all parent stages are complete, a stage becomes eligible for execution. The resource manager, powered by Alibaba's \textsf{Fuxi}~\cite{fuxi}, dynamically allocates resources from a shared, cluster-wide pool based on the stage’s computational demands and current system load. Each stage is executed by $1$ to over 100,000 parallel \emph{instances}, each processing a data partition and producing partial results. The final output is returned after all stages are complete. Despite its efficiency, the per-stage resource allocation and varying system loads across machines caused by multi-tenant interference introduce non-negligible variation in query execution cost. 
The bar plot in Figure~\ref{fig:workflow} shows the relative standard deviation of CPU costs for recurring queries from a production workload in \MC observed over one month, where an identical query can exhibit up to $50\%$ cost fluctuation. This makes cost prediction highly challenging for query optimization.




\textit{4. Logging.} 
Upon query completion, execution details, including the SQL statement, physical plan, execution environment (\eg, allocated resources and system load), end-to-end cost, and latency, are logged into a historical query repository maintained by the project for future inspection or downstream use. This feature distinguishes \MC from traditional database systems by offering a much richer data foundation for optimization.

\subsection{Existing Works: Review and Analysis}\label{subsec:reviewlqo}







Nearly all existing learned query optimizers
share a common architectural design and the paradigm for system operating, while incorporating specialized techniques to enhance their performance.
 
\stitle{Architecture Design Framework.} 
The learned query optimizers~\cite{neo, balsa, loger, base, bao, autosteer, lero, perfguard, fastgres, glo, foss, QO-advisor} all employ a \emph{plan explorer} to generate candidate query plans, alongside a learned \emph{cost model} that predicts the cost of each plan, but differ in the plan exploration technique and cost model design. For plan exploration, \textsf{Neo}~\cite{neo}, \textsf{Balsa}~\cite{balsa}, and subsequent works~\cite{loger, base, glo} do by themselves, using the best-first or ($\epsilon$-)beam search strategy to construct candidate plans in a bottom-up manner. Alternatively, \textsf{Bao}~\cite{bao} (along with its follow-ups~\cite{autosteer, bao-on-scope, QO-advisor}), \textsf{HyperQO}~\cite{hybrid-QO}, and \textsf{Lero}~\cite{lero}, steer traditional query optimizers with different knobs (\ie, hints on the operators, leading hints in the join orders, or scaled cardinality estimations) to produce a set of diverse plans. \textsf{FOSS}~\cite{foss} advances this line of work by iteratively optimizing query plans produced by traditional query optimizers using join-related hints. 

In cost model design, the cost metric can take various forms, including the query latency~\cite{neo, bao, loger, glo, foss, hybrid-QO}, CPU cost~\cite{balsa}, or even a relative score~\cite{lero}, which allows for comparing the relative quality of different query plans. However, they commonly assume a stable execution environment where plan cost is treated as a function of the plan and data alone. Accordingly, the cost model takes well-crafted feature vectors that encode plan structure (\eg, operator types) and predicates with statistics \wrt the distributions of input columns and tables (\eg, histograms of attributes, estimated cardinality, or cost).
Estimated costs are then derived through a deep network built upon Tree-LSTMs~\cite{hybrid-QO}, Graph or Tree Convolutional Networks ~\cite{neo, balsa, bao, lero, autosteer, glo}, or Transformers~\cite{glo, loger, foss}, that are capable of capturing the structural information of complex tree- or DAG-shaped query plans. 


In \MC, these practices become problematic. First, as introduced in Section~\ref{subsec:lqo-challenges}, \emph{the per-stage resource allocation and varying system loads across machines produce non-negligible cost fluctuations (Challenge~1)}, making cost prediction based purely on static plan structure and data statistics rather unreliable. Worse still, by our analysis in Section~\ref{sec:per-bound}, since the execution environment for online queries is invisible during the query optimization phase, environment variations introduce intrinsic errors to any query
optimizer, no matter how we improve the cost model. Second, \emph{\MC does not guarantee up-to-date or accurate statistics \wrt the underlying data that these learned cost models highly rely on (Challenge~2)}, which further degrades the reliability of cost estimations.

\stitle{System Operating Paradigm.} 
Almost all learned query optimizers perform \emph{instance-level} optimization over each database. Given a specific database $D$, the learned cost model is trained on pairs $(P, c(P))$ \wrt plans $P$ and cost $c(P)$. The training process consists of two phases: \emph{preparation} and \emph{refinement}. In the preparation phase, an initial version of the cost model is built in different ways: 1) initialized randomly~\cite{bao, hybrid-QO, foss}; or 2) trained on a collection of plans generated by the naive query optimizer with actual~\cite{neo, loger} or estimated~\cite{balsa, lero, QO-advisor} cost. 
The model produced in this phase is very coarse and crucially relies on the refinement phase to learn and generalize to the diverse candidate plan patterns. 

Specifically, given a workload containing numerous queries $Q$, the plan explorer generates one or more candidate plans $P$ for query $Q$. After executing $P$, $P$ and $c(P)$ are stored as experience to update the cost model. This phase can be carried out either offline~\cite{neo, balsa, glo} or online by deploying the partially trained model to serve real queries and periodically updating it with observed plan-cost pairs~\cite{neo, bao, lero, hybrid-QO, foss}.


However, \MC primarily serves OLAP queries that operate across multiple data partitions and involve heavy joins. Routine statistics record more than 7 million join-intensive queries per day, with an average of 3.8 tables joined. As exemplified in our evaluation projects in Section~\ref{subsec:qo-setup}, the CPU cost per query is typically on the order of $10^3$–$10^7$.
At this scale, executing extra candidate plans for offline refinement becomes prohibitively expensive. Online refinement is also problematic because deploying a partially trained optimizer can produce severely suboptimal plans that consume excessive resources, which is particularly unacceptable in a shared multi-tenant environment. Consequently, \emph{both of the conventional refinement approaches are infeasible in \MC (Challenge~3)}.

Despite addressing all the aforementioned challenges in query optimization, one critical barrier remains: Even when optimized, the benefits of learned optimizers are bounded by workload patterns (\eg, join topology) and data properties (\eg, table sizes and update frequency), and thus vary across projects. An optimizer that improves one project may yield little or no gain, or even regress, on another (see Section~\ref{subsec:e2e-evaluation}). With over 100,000 projects in \MC, it is neither efficient nor sustainable to train and deploy a separate learned optimizer per project. Instead, \emph{we must automate the selection of projects that will benefit meaningfully from deploying learned optimizers (Challenge~4)}.


We note that these challenges are not unique to \MC. Its workload patterns and operational characteristics are shared by other leading cloud‑native data warehouses~\cite{azure_synapse_analytics, bigquery, redshift, snowflake} (see Appendix~\ref{app:mc-vs-similar-systems}~\cite{full-paper}), and those systems therefore face similar obstacles when deploying learned optimizers. Motivated by this,
we formulate a set of foundational principles as introduced in Section~\ref{sec:intro} and develop solutions to address these issues. 
Despite being implemented and evaluated in \MC in this paper, they are broadly applicable and transferable to other production-grade data warehouses.

\vspace{-0.5em}
\section{The \sys Framework}\label{sec:framework}


In this section, we overview \sys, our one-stop solution for deploying \underline{l}earned query \underline{o}ptimizers on \underline{A}libaba's \textsf{\underline{M}axCompute}. 
It addresses all challenges described above and has seen notable performance improvements on \MC's production workloads.

The overall architecture of \sys is presented in Figure~\ref{fig:framework}. Rather than universally deploying \sys across all \MC's projects, we employ a \textbf{\emph{project selector}} to evaluate each project based on its historical queries and automatically select those most likely to benefit from \sys (as per Challenge~4). 


For a selected project, an \textbf{\emph{adaptive cost predictor}} is trained using data from its historical query repository to estimate the cost of a given query plan. By our observation, end-to-end latency in distributed systems like \MC is highly sensitive to transient system conditions such as queuing delays and network congestion, and thus often noisy. Accordingly, \sys predicts CPU cost as a more stable proxy for a plan’s total computational effort and long-term efficiency. To be applicable to \MC, the predictor explicitly models impacts from execution environments to make reliable cost estimations and intentionally avoids using statistics \wrt the underlying data (as per Challenges~1 and~2). Instead of the conventional refinement process, we also adopt an adaptive training mechanism that is both lightweight and effective to generalize the cost predictor to candidate plans for online queries. Before deployment, the predictor is evaluated on a sampled set of test queries (not seen in training) from the historical query repository. To obtain their actual cost as ground truth, they are executed in \MC's \emph{flighting} environment, which can replay user query plans without compromising privacy or disrupting the normal service of the user’s project. The results are then used to decide whether the predictor is suitable for production use.

Once \sys is deployed, each incoming user query $Q$ is optimized in the \emph{steering} style, as in \textsf{Bao}~\cite{bao} and \textsf{Lero}~\cite{lero}. 
Specifically, the \textbf{\emph{plan explorer}}, working with \MC's native query optimizer, generates a set of candidate plans $\{P_1, P_2, \cdots, P_k\}$. The adaptive cost predictor then estimates the costs of these plans, and the optimal plan $P^*$ with the minimum estimated cost is selected for actual execution. After executing $P^*$, its execution details are stored in the historical query repository in the background to refine the cost predictor in the future.

Next, we briefly introduce the three key components of \sys.

\begin{figure}
  \centering
  \includegraphics[width=\linewidth]{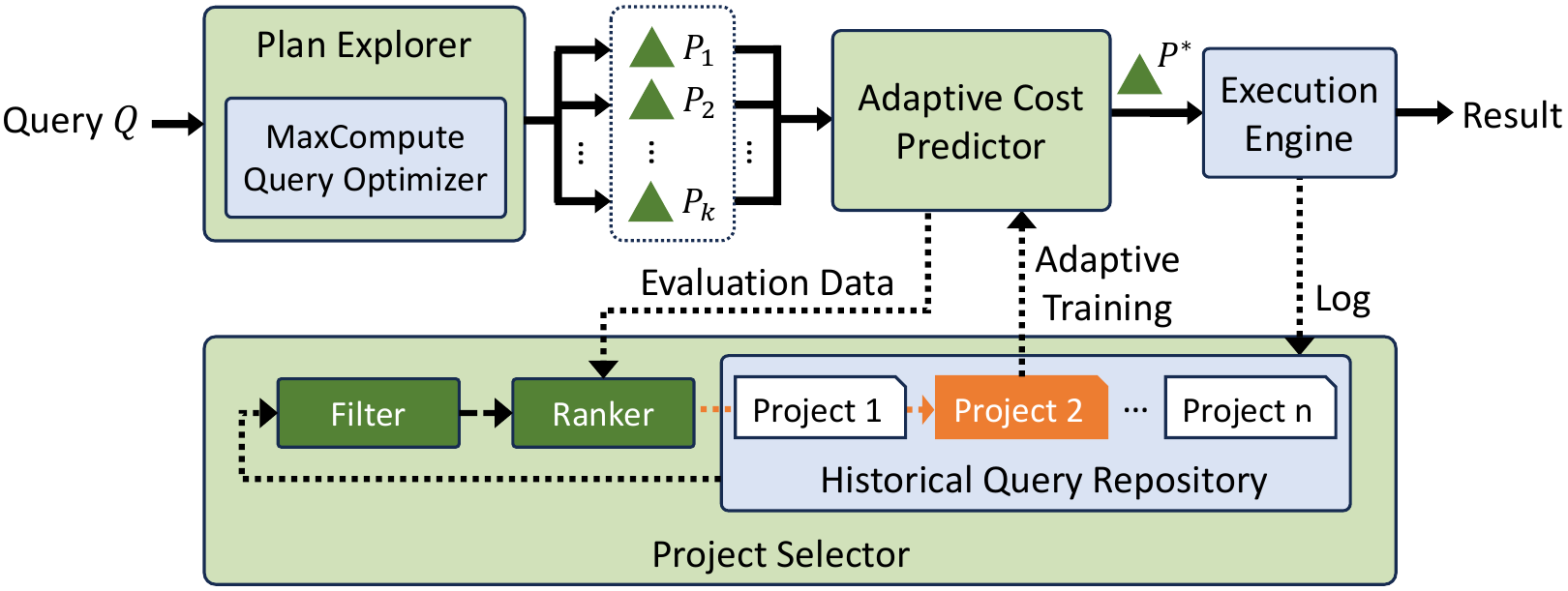}
  \vspace{-2em}
  \caption{The framework of \sys.}
  \label{fig:framework}
  \vspace{-2.1em}
\end{figure}


\stitle{Plan Explorer.} 
For a given query, the plan explorer aims to produce a diverse set of candidate plans that is likely to include plans with lower cost than the one produced by \MC's native query optimizer, which we refer to as \emph{default plan}. 
It generates candidate plans using ideas from prior work: i) following \textsf{Bao}~\cite{bao}, it toggles flags controlling \MC-specific optimizations. \MC exposes 75 tunable flags across six categories (\ie, execution mode, physical implementation optimizations, data-flow optimizations, degree of parallelism, storage-related settings, and other optimizations), providing a rich search space for discovering better candidate plans. In this paper, we restrict ourselves to six flags spanning join, shuffling, spool, and filter-related optimizations. They were selected by \MC's domain experts because they are more likely to 
yield diverse candidate plans, while remaining safe enough to avoid drastically bad plans and to keep evaluation overhead manageable in the experimental phase; 
and ii) following \textsf{Lero}~\cite{lero}, it scales estimated cardinalities for subqueries with at least three inputs to influence \MC's decisions for plan structures.
Notably, designing plan exploration strategies itself is not the focus of this paper. \sys is agnostic to the specific strategy employed and supports tuning any type of knob for plan generation. 



\stitle{Adaptive Cost Predictor.}
The cost predictor is constructed based on historically executed default plans produced by native \MC. Its architecture incorporates two key design aspects tailored to \MC. First, to compensate for the absence of statistical information such as histograms of attributes (Challenge~2), it encodes attributes of key operators to enable coarse inference of data distribution details from historical query plans. Second, to address the impact of dynamic execution environments (Challenge~1), we train the predictor to estimate a plan’s cost from both its structure and environmental features observed at stage granularity during execution, using historical query plans and their logged execution-time environments. This enables the model to learn how environmental variability influences actual execution costs. 
Because the feature distributions of these default plans can differ markedly from those of candidate plans tuned via various knobs for online queries, we also integrate \emph{domain adaptation}~\cite{domain-adaption} with adversarial techniques into the training process. This enables the predictor to generalize to estimate costs for candidate plans without conventional refinement processes (as per Challenge~3). 
The design details of the cost predictor and the training paradigm are presented in Section~\ref{sec:model}. 

Upon deployment, however, the execution environment for an online query is unavailable at prediction time (as the query has not yet started), which contrasts with the assumptions of current learned query optimizers (Section~\ref{subsec:reviewlqo}), where all feature values are available for cost estimation. To address this, in Section 5, we develop a theoretical model to quantify the impact of such unobserved environmental features and propose a simple yet effective method to robustly estimate plan costs against environment variation.

\stitle{Project Selector.} We have identified two primary reasons why a project may be unsuitable for deploying a learned query optimizer: 
1) \emph{Training Challenge}: Issues such as insufficient training queries or frequent data modifications hinder the development of an effective learned optimizer;
and 2) \emph{Low Deployment Benefit}: \MC's native query optimizer already produces near-optimal plans compared to those tuned by upfront knobs, making the benefits marginal. To filter out these projects, we adopt a two-stage, highly automated, and lightweight project selection pipeline (as per Challenge~4). First, a \emph{rule-based filter} coarsely excludes projects that exhibit training challenges using a set of well-defined rules. After that, a learned project \emph{ranker} orders the remaining projects by their potential deployment benefit. Highly ranked (\eg, top-$N$) projects are then selected to deploy \sys. The ranker is trained across multiple projects using the evaluation data for the cost predictor periodically. 
We provide the details in Section~\ref{sec:project-selection}.

\section{Building an Adaptive Cost Predictor}
\label{sec:model}


\begin{figure}
  \centering
  \includegraphics[width=\linewidth]{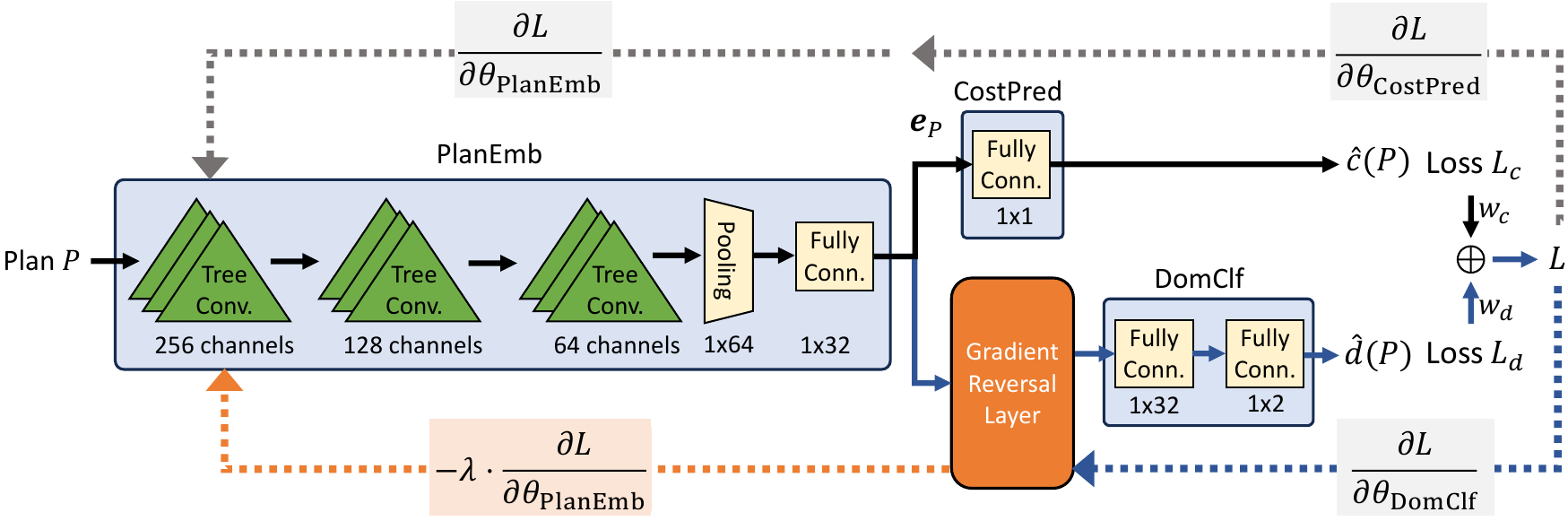}
  \vspace{-1.8em}
   \caption{Structure of the adaptive cost predictor.}
  \label{fig:model}
  \vspace{-1.8em}
\end{figure}

The adaptive cost predictor is designed ambitiously to learn from the historical CPU costs of default plans and to predict the costs of candidate plans generated by the plan explorer. Its overall architecture is illustrated in Figure~\ref{fig:model}. As in prior work (\eg, \cite{neo, balsa, bao}), the predictive module comprises a \emph{plan embedding layer} (\textsf{PlanEmb}) that transforms a carefully designed vectorized plan $P$ into a multidimensional intermediate representation, which is then fed into a \emph{cost prediction layer} (\textsf{CostPred}) to estimate the cost of $P$. Notably, although candidate plans follow a different distribution \wrt default plans, our method does not rely on conventional refinement to generalize. Instead, we employ an \emph{adaptive training} process with adversarial techniques to encourage \textsf{PlanEmb} to produce representations of similar distributions for both default and candidate plans, thereby making \textsf{CostPred} applicable to estimating the costs of candidate plans~\cite{domain-shift}. In this section, we first introduce a tailored plan vectorization approach for \MC, followed by the overall design of the predictive module and our adaptive training paradigm.





\stitle{Plan Vectorization.} 
Recall that each plan $P$ in \MC is a tree in which each node corresponds to an operator (\eg, table scanning and joining). We associate each node of the plan tree\footnote{We assume a canonical binary tree; otherwise, the plan is converted accordingly following~\cite{bao}.} with a vector encoding its key information that impacts the execution cost from two aspects:
 

\textbf{\emph{1. Operator Information.}} For each node in the plan tree, we adopt a one-hot encoding for the operator type. As shown in Figure~\ref{fig:plan}, the first two positions in the encoding indicate whether the operator is a \textsf{TableScan} or \textsf{MergeJoin}. We note that, although statistical information (\eg, histograms) is often absent in \MC, we could encode key attributes of certain operators to enable the model to \emph{infer information \wrt data distributions} from historical queries. The rationale is twofold:


i) Operator attributes reveal details about accessed tables and columns. Consequently, data-level insights such as volume, cardinality, and even distribution can be inferred through these attributes. For instance, for \textsf{TableScan} operators with identical predicates over different input tables, we could deduce which table is larger and coarsely infer table sizes from observed execution costs. Similarly, for join operations under the same join condition, we could predict whether joining two given tables yields a large or small result set.


ii) Production workloads are pervasively driven by parameterized, template-based queries whose parameters vary across runs. This stable, repetitive pattern provides sufficient signals for the model to observe how parameter changes influence execution costs. For example, by comparing queries that differ only in filtering predicates (\eg, $A_1 = a$ where $a$ is a distinct value in the domain of column $A_1$), the model could learn the selectivity of each predicate (and even infer the distribution of column $A_1$) from cost differences.


Guided by expert knowledge, we encode several classes of operators (among the 30 operator types supported by \MC) that are most frequently used and cost-impacting:

$\blacklozenge$ \emph{TableScan Operator.} 
In \MC, tables containing terabytes (or even exabytes) of data are logically partitioned by attributes for efficient access. For this operator, we encode the table identifier and the number of partitions and columns accessed, which can reflect the amount of processed data. Due to frequent creation and deletion of temporal tables in \MC, we use hash encoding (as detailed in Appendix~\ref{app:hash-encode}~\cite{full-paper}) for the table identifier to prevent explosion of the feature dimensionality. The numbers of partitions and columns are numerical variables, which are log-normalized using min–max normalization on their logarithms.


\begin{figure}[t]
  \centering
  \includegraphics[width=0.9\linewidth]{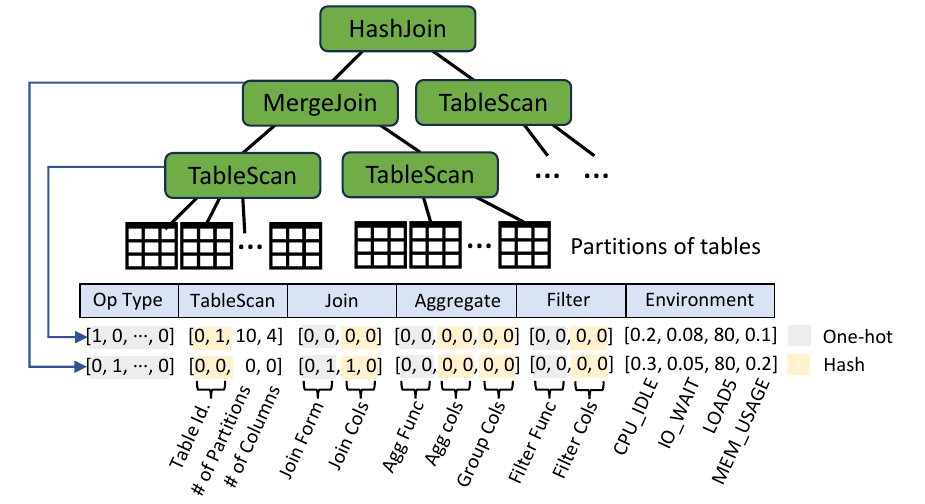}
  \vspace{-1.2em}
  \caption{Plan encoding in \sys.}
  \label{fig:plan}
  \vspace{-2.2em}
\end{figure}

$\blacklozenge$ \emph{All Joining and Aggregation Operators.} 
We represent each type of join operators (\eg, \textsf{HashJoin} and \textsf{MergeJoin}) with two attributes:
1) a one-hot vector for the join form (\eg, inner, outer, left, or right join); and 2) a hash encoding of the joining column identifiers. Similarly,
for each type of aggregation operators (\eg, \textsf{HashAggregate}), we use a one-hot vector for the aggregation function (such as \texttt{SUM} and \texttt{COUNT}), together with hash encodings for both aggregate and group‑by column identifiers. These encodings enable the model to infer the data volume joined or aggregated \wrt different columns.


$\blacklozenge$ \emph{Filtering and Related Operators.} Filtering predicates in \MC are structured as expression trees, where internal nodes denote functions (\eg $>, <, =$), and leaf nodes correspond to columns and constants. These trees can grow to tens or even hundreds of levels, so encoding the entire tree would introduce great complexity without materially improving the model’s understanding of the predicate. Thus, we adopt a simplified representation for the \textsf{Filter} operator with two key attributes: 1) a multi-hot encoding of all functions involved; and 2) a hash encoding for the identifiers of all involved columns. Other operators that implement filtering logic, \eg, the \textsf{Calc} operator (which incorporates both filtering and projection operations), are encoded similarly. This representation enables the model to learn coarse-grained selectivity \wrt the predicate.

\textbf{\emph{2. Execution Environment.}}
Both the hardware configuration of allocated machines and their load can affect plan costs. In \sys, we focus on encoding load information and omit explicit machine-level hardware features for two reasons: 1) in production, a stage may execute across hundreds of machines, making direct hardware encoding impractical; and 2) machines within the same cluster are intentionally homogeneous to facilitate scheduling and load balancing. We therefore reasonably assume identical computational power across machines, \ie, plan cost is independent of per-machine hardware configurations.

We model the load on allocated machines using four standard metrics, namely \texttt{CPU\_IDLE}, \texttt{IO\_WAIT}, \texttt{LOAD5}, and \texttt{MEM\_USAGE}, which capture the status of CPU, I/O, system load, and memory (detailed in Appendix~\ref{app:load-metric}~\cite{full-paper}). On each machine, these metrics are highly dynamic and sampled every 20 seconds. Therefore, to derive a representative environmental feature for each plan node, we average each metric over the stage's execution window and, because execution typically spans multiple machines, also across all allocated machines. 
The metric \texttt{LOAD5} is log-normalized, while other metrics are already bounded and used directly, so all features fall within $[0, 1]$. Notably, all plan nodes within the same stage run on the same set of allocated machines (see Section~\ref{subsec:lqo-challenges}), and thus share the same encoding vector \wrt execution environments.




These features are crucial for offline training on historical queries that were executed under diverse production environments. Incorporating them enables the model to learn hidden patterns of how execution environments
influence execution cost and to disentangle environmental impacts from plan‑intrinsic features. Without these features, cost labels are noisy to the model, making it hard to learn a stable mapping from plans to costs. At serving time, however, environmental features become unobservable for online query plans. We defer details of our approach for online inference with invisible environmental feature values to Section~\ref{sec:per-bound}.


\stitle{Predictive Module Design.} 
After a plan $P$ is vectorized, \textsf{PlanEmb} maps $P$ from its input feature space into a compact $n$-dimensional embedding, denoted as $\bm{e}_P \in \mathbb{R}^n$. We build \textsf{PlanEmb} using the Tree Convolutional Networks (TCNs), similar to~\cite{bao, neo, balsa, lero}. Tree convolution applies learnable filters over each tree node and its children, aggregating information upward from child to parent. By stacking more TCN layers, each node progressively integrates hierarchical information from deeper subtrees. The resulting node representations are pooled
and then passed through a fully connected layer to produce $\bm{e}_P$. \textsf{CostPred}, also implemented as a simple fully connected layer, maps $\bm{e}_P$ to the estimated cost $\hat{c}(P)$. Notice that designing the predictive module is orthogonal to the topic of this paper. 
We also evaluate alternative architectures, including Transformer~\cite{transformer}, Graph Convolutional Networks (GCNs)~\cite{GCN}, and XGBoost~\cite{xgboost}. Section~\ref{subsec:e2e-evaluation} reports the comparative results.

\stitle{Adaptive Training Paradigm}
To generalize the predictive module trained on historical default plans to candidate plans without conventional refinement,
we incorporate unsupervised domain adaptation techniques~\cite{domain-adaption} into the training process. The core idea is to ensure that embeddings produced by \textsf{PlanEmb} are both discriminative for cost prediction and \emph{invariant to domain shifts}. Concretely, we seek the distributions of $\textsf{PlanEmb}(P)$ for default plans $P$ and $\textsf{PlanEmb}(P')$ for candidate plans $P'$ to be as similar as possible, so that \textsf{CostPred} can generalize to estimate cost for candidate plans~\cite{domain-shift}. 
To achieve this, we augment the predictive module with a \emph{domain classifier} \textsf{DomClf} that distinguishes between embeddings produced for default plans and candidate plans (see Figure~\ref{fig:model}). We then conduct adversarial training between \textsf{PlanEmb} and \textsf{DomClf}: \textsf{PlanEmb} aims to produce similar representations across domains, while \textsf{DomClf} strives to differentiate them. This interplay encourages \textsf{PlanEmb} to produce domain-invariant embeddings. In our implementation, \textsf{DomClf} consists of two fully connected layers that map the embedding $\bm{e}_P$ from \textsf{PlanEmb} to a $2$-dimensional probability vector $\hat{d} (P)$ indicating the likelihood that the plan $P$ is a default or candidate plan.

During training, we jointly pursue two objectives: 1) optimize \textsf{PlanEmb} and \textsf{CostPred} to accurately estimate execution costs for default plans sampled from the historical query repository, denoted as $\mathcal{P}_\textsf{def}$; and 2) adversarially train \textsf{PlanEmb} and \textsf{DomClf} so that \textsf{PlanEmb} yields domain-invariant representations for default and candidate plans. The second objective requires the model to see a small set of candidate plans generated by the plan explorer, denoted as $\mathcal{P}_\textsf{cand}$. However, these plans need not be executed, as the training process focuses on distinguishing them from default plans rather than estimating their actual costs. As shown in Section~\ref{subsec:e2e-evaluation}, the overhead of purely generating them is negligible.

To simplify the adversarial training by backpropagation, we follow~\cite{domain-adaption} by inserting a \emph{gradient reversal layer (GRL)} $G$ between \textsf{PlanEmb} and \textsf{DomClf}, which acts as $G(\bm{e}_P) = \bm{e}_P$ and $\partial G(\bm{e}_P) / \partial{\bm{e}_P} = -\lambda \bm{I}$, where $\bm{I}$ is the identity matrix. That is, during the forward pass, the GRL passes features unchanged, but during backpropagation, it reverses the gradient by a factor of $-\lambda$ (whose value is set automatically
following~\cite{domain-adaption}). By this mechanism, \textsf{DomClf} is still optimized to distinguish the domain of plans as the gradient from the final loss to it does not change (the dashed blue arrow in Figure~\ref{fig:model}). Whereas, the gradient from the final loss
to \textsf{PlanEmb} via \textsf{DomClf} is reversed (the dashed orange arrow in Figure~\ref{fig:model}), pushing \textsf{PlanEmb} to produce domain-invariant features.

Let $L_c(\cdot)$ and $L_d(\cdot)$ denote the loss functions for cost prediction and domain classification, respectively. The overall loss is: 
\begin{equation}~\label{eq:loss}
    L = w_c\sum_{P \in \mathcal{P}_\mathsf{def}} L_c\big(\hat{c}(P), c(P)\big) + w_d\sum_{P \in \mathcal{P}_\mathsf{def} \cup \mathcal{P}_\mathsf{cand}} L_d\big(\hat{d}(P), d(P)\big),
\end{equation}
where $d(P) \in \{(0, 1), (1, 0)\}$ indicates whether $P$ is a default or candidate plan. The weights $w_c$ and $w_d$ balance the magnitudes of losses for cost prediction and domain classification, and are adjusted automatically during training based on their scales. We use mean squared error for $L_c(\cdot)$ and cross-entropy loss for $L_d(\cdot)$.

\section{Plan Cost Inference}
\label{sec:per-bound}

When the adaptive cost predictor is deployed online, it estimates the execution costs of candidate plans generated for each query. However, the execution-time environment is unavailable upon query optimization, so we do not know which environmental feature values to use for inferring plan costs. In this section, we first develop a probabilistic framework to formalize and quantify the impact of such invisible environment variation. We show that this variation fundamentally limits a query optimizer's ability to consistently pick the truly optimal plan under the realized environment and establishes a non-trivial upper bound on best-achievable performance in real-world deployment. We then introduce a practical strategy for setting environmental features at inference time that can effectively mitigate the adverse impacts of this variation and improve the efficiency of plan selection.

\stitle{Theoretical Analysis.}
Let $E$ be a multi-dimensional random variable representing the system environments. For any plan $P$, let $C_{E}(P)$ denote its execution cost under environment $E$, which is itself a random variable determined by $E$. The cost $c(P)$ is the value of $C_{E = e}(P)$, or denoted as $C_{e}(P)$, for an observed instance $e$ of $E$. Given a query $Q$ with candidate plans $\mathcal{P}_Q = \{P_1, P_2, \cdots, P_n\}$ and a (learned) cost model $M$ for cost prediction, let $P_{M}$ be the plan selected by $M$, \ie, with the smallest estimated cost. In the ideal case, an \emph{oracle model} $M_{o}$ would always select the plan with the minimum cost under any instance $e$ of $E$, \ie, $P_{M_{o}} = \min_{P_i} C_{e}(P_i)$. 


In practice, $M_{o}$ is unattainable as we cannot foresee the value $e$ before plan execution. Instead, we consider expected cost \wrt $E$ and pursue a model $M_{b}$ that could always select the plan with the minimum expected cost, \ie $P_{M_{b}} = \min_{P_i} \mathbb{E}[C_{E}(P_i)]$. While $P_{M_{b}}$ may not minimize the cost for every instance $e$, it minimizes cost on average over all possible environments.

To evaluate the quality of plan selection for different models, we measure the cost deviance of a selected plan from the oracle choice. For a model $M$ that selects plan $P_{M}$ under environment instance $e$, we define $D_{e}(M) = C_{e}(P_{M}) - C_{e}(P_{M_{o}})$. This induces a random variable $D_{E}(M)$ such that $D_{E=e}(M) = D_{e}(M)$. The expected deviance, \ie, $\mathbb{E}[D_{E}(M)]$, naturally quantifies the performance gap between $M$ and the oracle model $M_o$. Because the environment instance $e$ cannot be known before plan execution, no matter the efforts we devoted to improving the model, the following Theorem~\ref{Thm:perfgap} states that: 
1) the performance gap brought by the invisibility of environments is intrinsic, which is fundamentally different from the assumptions of existing learned query optimizers reviewed in Section~\ref{subsec:reviewlqo};
and 2) the model $M_b$ minimizing the expected cost is \textbf{\emph{best-achievable}}. 
We put its proof in Appendix~\ref{app:proof-for-perf-bound}\cite{full-paper}. 
By testing typical queries in \MC, we find that $\mathbb{E}[D_{E}(M_b)]$ is often around 10\% of the oracle cost (see Section~\ref{subsec:load-effect}). 
\vspace{-0.5em}
\begin{theorem}
\label{Thm:perfgap}
For any model $M$ that can not foresee the environment instance $e$ before plan execution, we always have $\mathbb{E}[D_{E}(M)] \geq \mathbb{E}[D_{E}(M_b)] \geq \mathbb{E}[D_{E}(M_o)] = 0$.
\end{theorem}

\stitle{Practical Method for Model Inference.}
While by Theorem~\ref{Thm:perfgap}, selecting the plan $P$ minimizing the expected cost attains the best-achievable performance, computing $\mathbb{E}[C_E(P)]$ is practically infeasible. In production, it is very difficult to express the environment distribution and its impact on plan costs in canonical form, so we could only approximate the expected cost of plan $P$ using Monte Carlo sampling, \ie, executing $P$ multiple times under different environmental settings and averaging the costs. However, this approach faces an inherent trade-off between stability and efficiency. Achieving a stable estimate requires a large number of samples, which demands repeated plan executions and substantially increases query optimization time. Using fewer samples lowers overhead but raises high variance and can lead to unstable or suboptimal plan choices. 

\begin{figure}[t]
    \centering
    \includegraphics[width=0.43\linewidth]{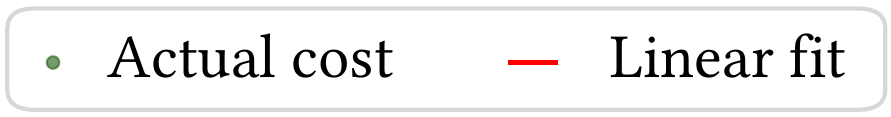}
    
    \includegraphics[width=0.48\linewidth]{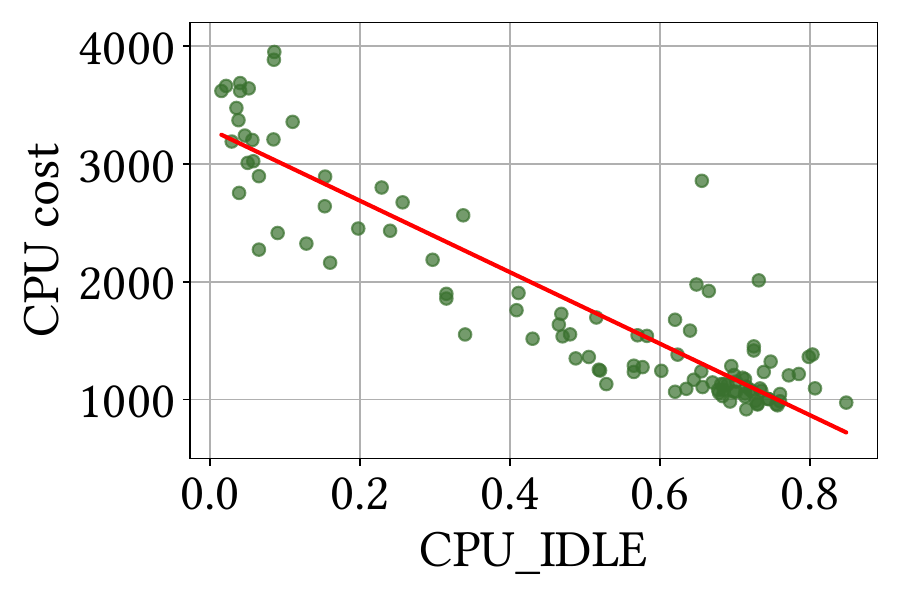}
    \includegraphics[width=0.48\linewidth]{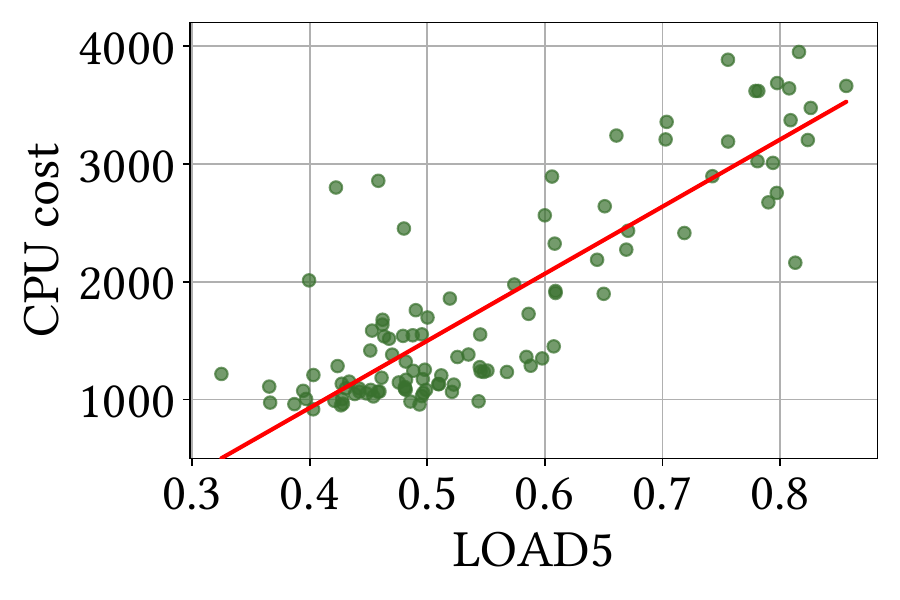}
    \vspace{-0.5em}
    \caption{CPU cost of a recurring query \wrt machine load.}
    \label{fig:cost}
    \vspace{-1em}
\end{figure}

To address this, we adopt a more practical approximation to the expectation of execution costs. Motivated by~\cite{representative}, we select a representative environment instance $e_{r}$ such that, for any candidate plan $P_i$, $C_{e_{r}}(P_i)$ could be close to the expected cost $\mathbb{E}[C_{E}(P_i)]$. Under this setting, the model $M_{r}$ selects the plan $P_{M_{r}} = \arg\min_{P_i} C_{e_{r}}(P_i)$ among all candidate plans. Our empirical analysis of \MC's historical queries indicates that environmental features have a discernible, roughly monotonic influence on plan costs that can be coarsely approximated as linear for cost estimation. Figure~\ref{fig:cost} illustrates this by presenting the CPU cost of a simple production query from \MC's workload against \texttt{CPU\_IDLE} and \texttt{LOAD5} averaged across all plan nodes. 
This motivates us to instantiate $e_{r}$ by setting each environmental feature to its empirical mean, specifically, near $0.5$ within the normalized range $[0, 1]$ (except for the expected \texttt{IO\_WAIT} value approaching $0.05$) on average over queries, for cost inference. Intuitively, this evaluates all plans under typical, average-case system load conditions, which is fair and reasonable. It is also computationally efficient as it requires model inference only once, with no additional plan execution. Our experiments in Section~\ref{subsec:load-effect} also demonstrate that this strategy outperforms alternative choices, and its expected deviance $\mathbb{E}[D_{E}(M_r)]$ is close to the optimal $\mathbb{E}[D_{E}(M_b)]$.

\section{Project Selection}
\label{sec:project-selection}


Projects in \MC may be unsuitable for deploying learned query optimizers due to \emph{training challenges} and \emph{low deployment benefit}. To filter out unsuitable projects, we first sample a small, representative workload $\mathcal{Q}$ from each project's historical queries to capture its overall query patterns and workload characteristics. Then, we apply a rule-based filter (\textsf{Filter}) to exclude projects that are likely to pose training challenges and rank the remaining ones by their estimated benefits using a lightweight learned model (\textsf{Ranker}). We next detail \textsf{Filter} and \textsf{Ranker} below.






\stitle{Rule-based Project Filter.} 
\textsf{Filter} comprises a set of rules that assess whether a project's historical workload is suitable for training an effective cost predictor across multiple aspects, such as the volume of historical queries, properties of the involved tables, and the diversity of plan structures. We quantify these aspects with dedicated metrics and require each to meet a threshold that supports effective model training. The current \textsf{Filter} adopts three rules: i) $\mathsf{R1}$ and $\mathsf{R2}$ jointly require a sufficient number of historical queries for training, which has been shown to strongly affect the performance of the cost predictor (see our evaluation results in Section~\ref{subsec:training-size-effect}); and ii) $\mathsf{R3}$ favors projects with long-lived tables so that data distributions \sys learned from historical queries can be reused to predict costs for future queries over the same tables. Details of these rules and their metrics are provided in Appendix~\ref{app:more-metrics}~\cite{full-paper}. Our statistical analysis shows that 59.5\% of projects in \MC are filtered out by these rules, indicating that many projects are not yet ready to deploy learned optimizers and underscoring the importance of preemptively filtering out such projects.

\stitle{Learned Project Ranker.} Given a query $Q$, let $P_{d}$ denote the default plan produced by \MC's native query optimizer $M_{d}$. As shown in Section~\ref{sec:per-bound}, the expected deviance $\mathbb{E}[D_{E}(M_{d})]$, abbreviated as $D(M_{d})$ below, measures the performance gap between $M_{d}$ and the oracle query optimizer. Because a larger $D(M_{d})$ indicates poorer performance of $M_{d}$ and thus potentially more benefits from replacing $M_{d}$ with a learned optimizer, we naturally use $D(M_{d})$ as a proxy for the improvement space of deploying a learned optimizer. For a project with sampled workload $\mathcal{Q}$, we deploy \sys only if $D(M_{d})$ is large for a sufficient fraction of queries in $\mathcal{Q}$. 
Since directly computing $D(M_{d})$ is time-consuming and does not scale to the sampled workload for each \MC's project, we train a learned model \textsf{Ranker} to estimate $D(M_{d})$ for a given query $Q$.

The key insight behind \textsf{Ranker} is that, \emph{for a query $Q$, the observable properties of its default plan $P_{d}$ reveal optimization opportunities from \sys's plan exploration strategies and can thereby reflect the improvement space (\ie, $D(M_{d})$) of $Q$}. For example, suppose $P_d$ contains multiple nested joining operators while exhibiting an unusually high execution cost. In that case, it often indicates a poor join order and suggests that steering $M_d$ to produce plans with alternative join orders could yield large performance gains. Building on this idea, \textsf{Ranker} learns to estimate $D(M_{d})$ for query $Q$ by the characteristics of $P_d$. These characteristics are generic to plans and require no project-specific identifiers to encode tables or columns, so \textsf{Ranker} can be trained on pairs $(P_d, D(M_{d}))$ derived from multiple projects and then applied to new projects. For greater scalability, \textsf{Ranker} employs an efficient plan vectorization approach that encodes the structure, input features, and execution cost of $P_d$ and a lightweight XGBoost model~\cite{xgboost} to make predictions. Due to space limitations, we defer the vectorization details in Appendix~\ref{app:feat-for-ranker}~\cite{full-paper}.

Using \textsf{Ranker}, we estimate the improvement space $D(M_{d})$ for each query in the sampled workload $\mathcal{Q}$ of every project that passes \textsf{Filter}. Projects are then ranked by their average estimated $D(M_{d})$ across $\mathcal{Q}$. Determining a reliable threshold to select projects with high estimated benefits is non-trivial due to the varying scales of $D(M_d)$ across projects and the inherent noise in \textsf{Ranker}’s estimates. Instead, \sys is deployed on the top-$N$ projects (where $N$ is a parameter specified by \MC's experts) for better robustness. Notably, as more projects are deployed with \sys and evaluated, new $(P_d, D(M_d))$ pairs can be derived from the evaluation data and periodically used to update \textsf{Ranker}, which continuously improves its prediction accuracy and effectiveness over time.

\section{Experimental Evaluation}
\label{sec:exp}


This section presents a thorough empirical evaluation of \sys. We begin in Section~\ref{subsec:qo-setup} by describing the experimental setup. Section~\ref{subsec:perf-of-sys-qo} then reports key results that showcase the effectiveness of \sys and the benefits of its design choices. Finally, Section~\ref{subsec:mc-benefit} addresses the most crucial question concerning \sys's practical contribution: \emph{How broadly can \sys be applied across \MC's projects, and what performance gains can be expected in practice?}

\vspace{-0.5em}
\subsection{Experimental Setup}
\label{subsec:qo-setup}


We evaluate \sys's query optimizer using production workloads from Alibaba Group’s internal projects in \MC. The experimental settings are detailed below. 



\stitle{Evaluation Projects.} 
We construct our experimental dataset by combining the project selection pipeline (Section~\ref{sec:project-selection}) with random sampling. First, we apply our rule-based filter to workloads consisting of queries collected over three consecutive days for each \MC's project, yielding a filtered set that covers approximately $40.5\%$ of all projects. To make the evaluation tractable while remaining statistically representative, we randomly select $30$ projects from this filtered set for our main experiments. Since these sampled projects can show varying performance gains and to understand what \sys can achieve, we calculate the exact improvement space $D(M_d)$ for each sampled project (defined in Section~\ref{sec:project-selection}, with computational details provided in Appendix~\ref{app:proof}~\cite{full-paper}) and select the five with the highest average $D(M_d)$. These projects, anonymized as \textsf{Projects 1--5} due to privacy constraints, are used in our experiments in Sections~\ref{subsec:e2e-evaluation} to~\ref{subsec:load-effect} to evaluate the effectiveness and efficiency of \sys's query optimization techniques. To be conservative, we regard results on these five projects as an upper bound for the sampled set and treat the remaining 25 as low‑benefit cases. Section~\ref{subsec:mc-benefit} uses this convention to estimate the overall benefit of \sys across \MC's projects. In practice, computing $D(M_d)$ exactly is infeasible because it requires enumerating and executing all candidate plans. We therefore use \textsf{Ranker} as a practical surrogate to estimate the improvement space. Its accuracy and robustness are validated in Section~\ref{subsec:eval-ranker}. Among the 5 evaluation projects, \textsf{Ranker} identifies \textsf{Projects~1, 2,} and \textsf{5}.

\stitle{Baselines.} 
We compare \sys with \MC's native optimizer and learned optimizer variants that replace \sys's TCN-based cost predictor with other representative cost models from prior work. Specifically, we consider~\cite{queryformer}, ~\cite{zeroshot}, and \cite{perfguard} that are built on Transformer~\cite{transformer}, Graph Convolutional Network (GCN)~\cite{GCN}, and XGBoost~\cite{xgboost}, respectively. Transformer and GCN excel at capturing complex plan structure and have shown strong performance in plan cost prediction~\cite{queryformer,glo,loger,foss,stage,perfguard,zeroshot}. XGBoost~\cite{xgboost} has recently been applied to estimate plan costs for its simplicity, efficiency, and performance comparable to deep networks~\cite{perfguard,stage}. In the remainder of this section, we refer to \MC's native optimizer as \MC, the default \sys with our TCN-based cost predictor as \sys, and other learned optimizer baselines by their underlying models, namely \textsf{Transformer}, \textsf{GCN}, and \textsf{XGBoost}.

For all learned optimizer baselines, we reuse \sys's plan explorer and implement the cost predictors using public source code for \textsf{Transformer} and \textsf{GCN} from~\cite{queryformer} and~\cite{zeroshot}, and the standard library~\cite{xgboost-lib} for \textsf{XGBoost} following~\cite{perfguard}. To adapt these baselines to \MC and ensure a fair comparison, we adjust their plan vectorization by: i) removing features that rely on statistics unavailable in \MC; and ii) adding \sys's additional feature set using their native feature encoding approaches. 


\begin{table}[t]
\renewcommand{\arraystretch}{1}
    \centering
    \caption{{Statistics of projects used in the experiments.}}
    \vspace{-0.15in}
    \scalebox{0.76}
    {{
    \begin{tabular}{|c|c|c|c|c|c|} 
    \hline
    \rowcolor{mygrey}
    \sf \bf Datasets & \sf \citylife & \sf \gd & \sf \fin & \sf \ads  & \sf \mkt \\ 
    \hline
    \sf \bf \# of tables & 253 & 125& 348&209 &229 \\ 
    \sf \bf \# of columns & 3,782& 714& 7,382 &3,794 & 3,661 \\ 
    \sf \bf \# of training queries &10,000&10,000&10,000&4,187 & 8,701 \\
    \sf \bf \# of test queries &184&101&177&573 & 126 \\ 
    \sf \bf Average CPU cost &11,501&1,824,978&3,265&1,354 & 103,040 \\
    \hline
    \end{tabular}
    }}
    \label{tbl:dataset}
    \vspace{-0.8em}
\end{table}

\stitle{Model Training and Evaluation.} 
\sys and all learned optimizer baselines are trained separately for each evaluation project. Specifically, given a project, we collect deduplicated queries over 30 consecutive days from its historical query repository and use those from the first 25 days for training, while the remaining 5 days for testing. For fair comparisons, we limit the maximum number of training queries to 10,000. For \textsf{Projects~4} and \textsf{5}, which have fewer than 10,000 queries for training, we use all their available queries. 
Table~\ref{tbl:dataset} summarizes the statistics of these workloads. 


Importantly, without an adaptive training process to align feature distributions of training (default plans) and test query plans (candidate plans), all learned optimizer baselines suffer from distribution shifts. This creates a critical issue where conventional hyperparameter tuning on validation splits of the training data becomes ineffective, because validation performance no longer reliably reflects test performance. To maintain methodological parity, we avoid explicit hyperparameter tuning across all models. For \textsf{Transformer} and \textsf{GCN}, we adopt the configurations from their original papers~\cite{queryformer,zeroshot}, and for \textsf{XGBoost}, we use the standard library defaults~\cite{xgboost-lib}. For our proposed \sys built on TCNs following \textsf{Bao}\cite{bao} and \textsf{Lero}\cite{lero}, we inherit their parameter settings and additionally adopt a standard optimization setup with an initial learning rate of 0.01 and an exponential decay factor of 0.99 per epoch.


At evaluation time, the plan explorer generates multiple candidate plans for each test query and, without loss of generality, always includes the default plan (produced by \MC's native optimizer without using any flags or scaled estimated cardinality). Among these, learned optimizers select the candidate with the lowest estimated cost. This setup mirrors real-world deployment scenarios. To control evaluation overhead, we retain only the top-5 candidates for each test query based on \MC's rough cost estimates. Learned optimizers are evaluated based on the average end-to-end (E2E) CPU cost of their selected plans, while the performance of \MC's native optimizer is measured by the cost of default plans. To obtain reliable measurements, each candidate plan is executed multiple times, and the average cost is used.

\stitle{Environments.} 
All queries were executed in \MC's flighting environment (see Section~\ref{sec:framework}) on an Alibaba's internal distributed cluster spanning over 10,000 machines. Model training and inference ran on a Linux server with an Intel Xeon Platinum 8163 CPU (96 cores, 2.5 GHz), 768 GB DDR4 RAM, a 2 TB SSD, and 8 NVIDIA Tesla V100-SXM2 GPUs.

\eat{
\stitle{Evaluation Metrics}
We evaluate \emph{project ranker} and \emph{plan selection} components using multiple metrics.
For \emph{project ranker}, we employ the Recall and NDCG metric, while for \emph{plan selection}, we use average end-to-end (E2E) execution cost and average expected error. These metrics are defined as follows:

\sstitle{Average E2E execution cost}: This metric evaluates the effectiveness of the plan
selection component. For a query set $\mathcal{Q}$, let $\hat{c_i}$ denote the average cost of multiple executions of the selected plan for the i-th query. The average E2E execution cost across all queries is defined as $(1/{|\mathcal{Q}|})\sum_{i}^{|\mathcal{Q}|} \hat{c_i}$.
}

\vspace{-0.5em}
\subsection{Experimental Results}
\label{subsec:perf-of-sys-qo} 


Experiments in this section aim at answering the most central questions concerning \sys's query optimizer and project selector: 




$\bullet$ \textbf{\emph{Overall Performance and Overhead:}} How much improvement on query execution efficiency can \sys achieve over \MC's native optimizer and other state-of-the-art learned optimizers, and at what additional deployment and operational overhead? (Sections~\ref{subsec:e2e-evaluation} and~\ref{subsec:per-query-evaluation})

$\bullet$ \textbf{\emph{Cost Predictor Training:}} 
How much does adaptive training improve the cost predictor (Section~\ref{subsec:grl-evaluation}), and how does the volume of training data affect the end-to-end performance (Section~\ref{subsec:training-size-effect})? 

$\bullet$ \textbf{\emph{Plan Cost Inference:}} How effective is our plan cost inference method in comparison to alternative strategies? (Section~\ref{subsec:load-effect}) 

$\bullet$ \textbf{\emph{Project Selection:}} How well does \textsf{Ranker} perform in prioritizing projects with large improvement space? (Section~\ref{subsec:eval-ranker})


\subsubsection{\textbf{End-to-end Performance in \MC}}
\label{subsec:e2e-evaluation}



\begin{figure*}
\begin{minipage}[c]{0.65\textwidth}
    \centering
    \includegraphics[width=0.8\textwidth]{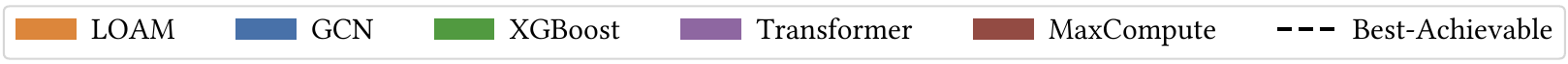} 
    
    \includegraphics[width=0.19\linewidth]{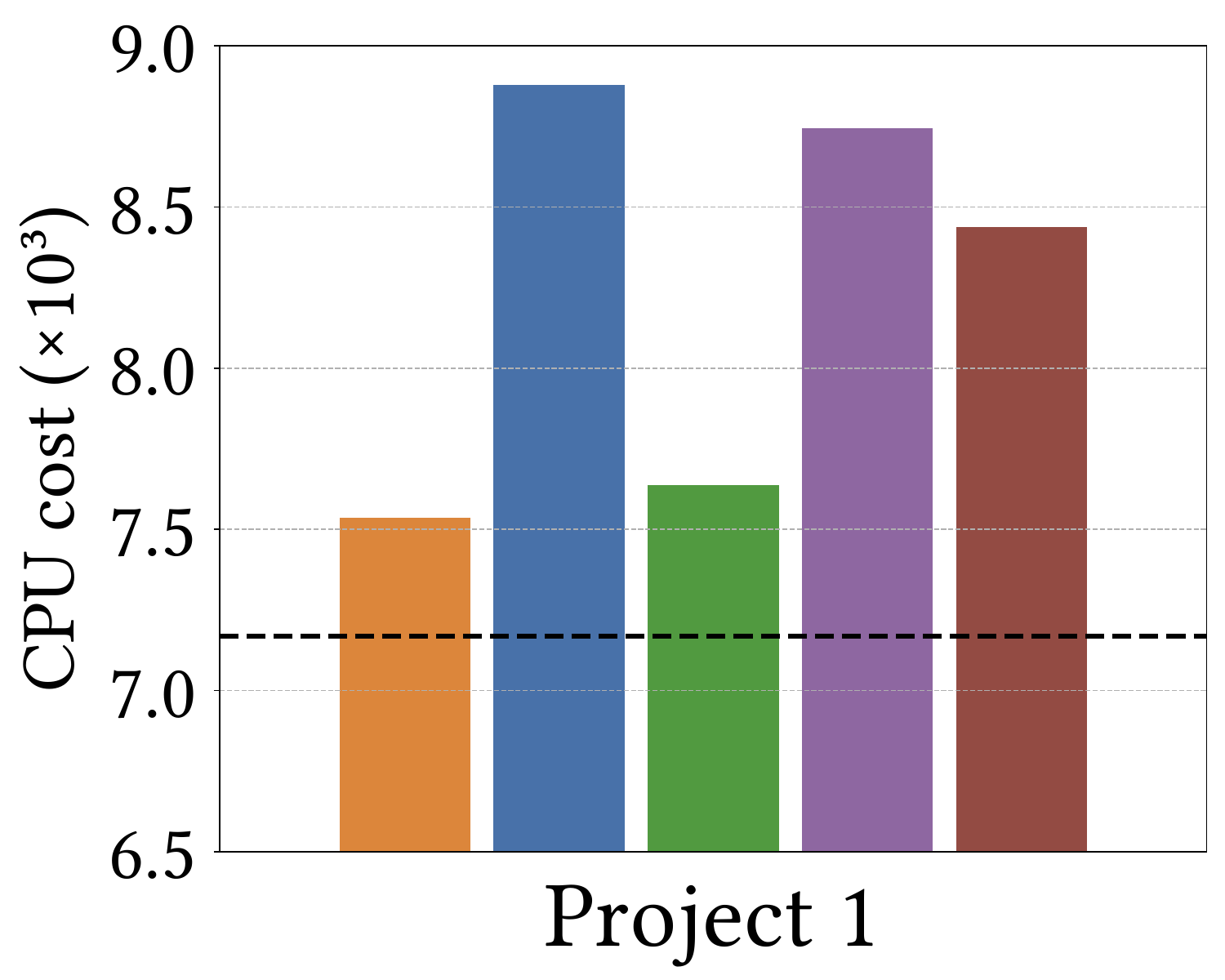}
    \includegraphics[width=0.19\linewidth]{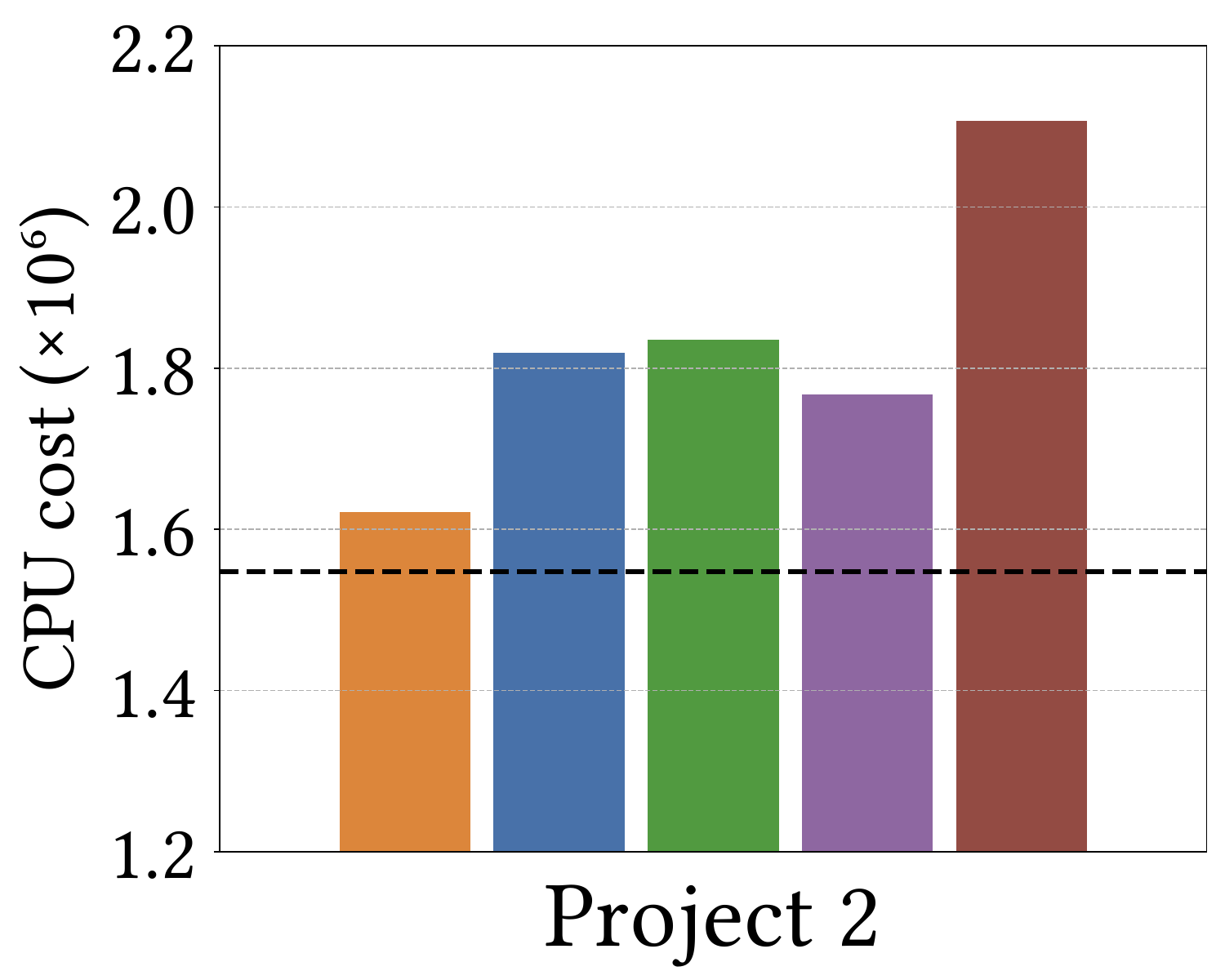}
    \includegraphics[width=0.19\linewidth]{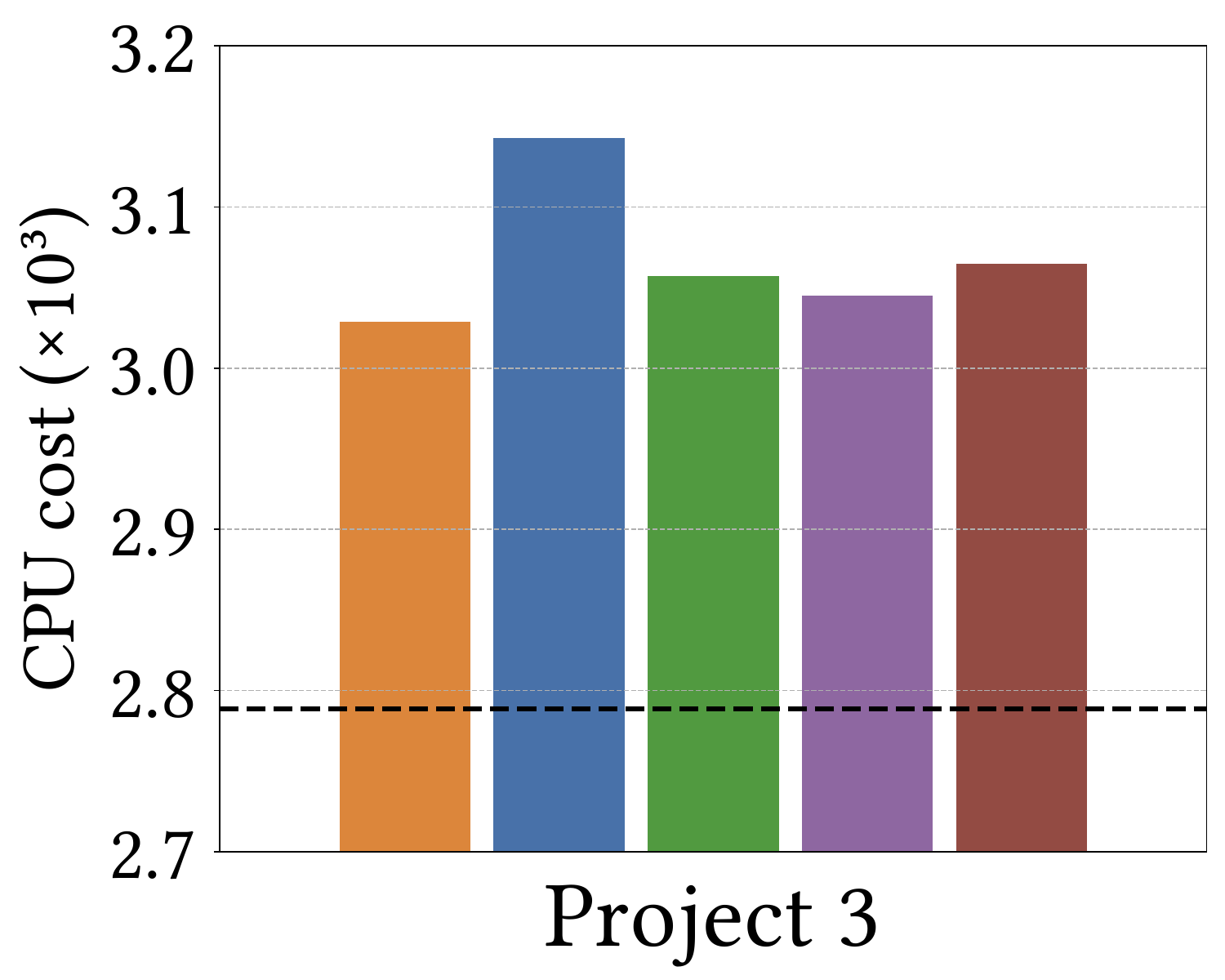}
    \includegraphics[width=0.19\linewidth]{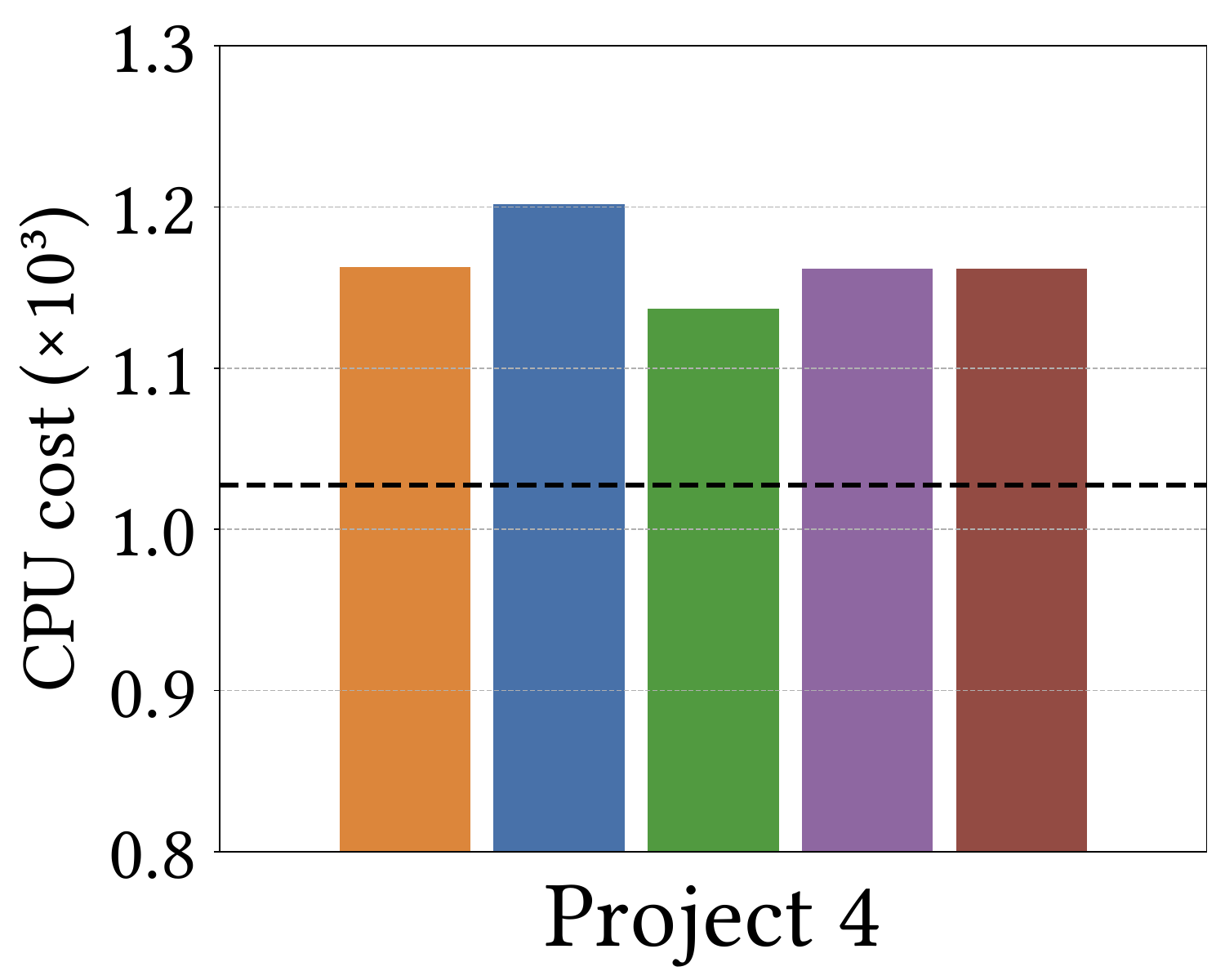}
    \includegraphics[width=0.19\linewidth]{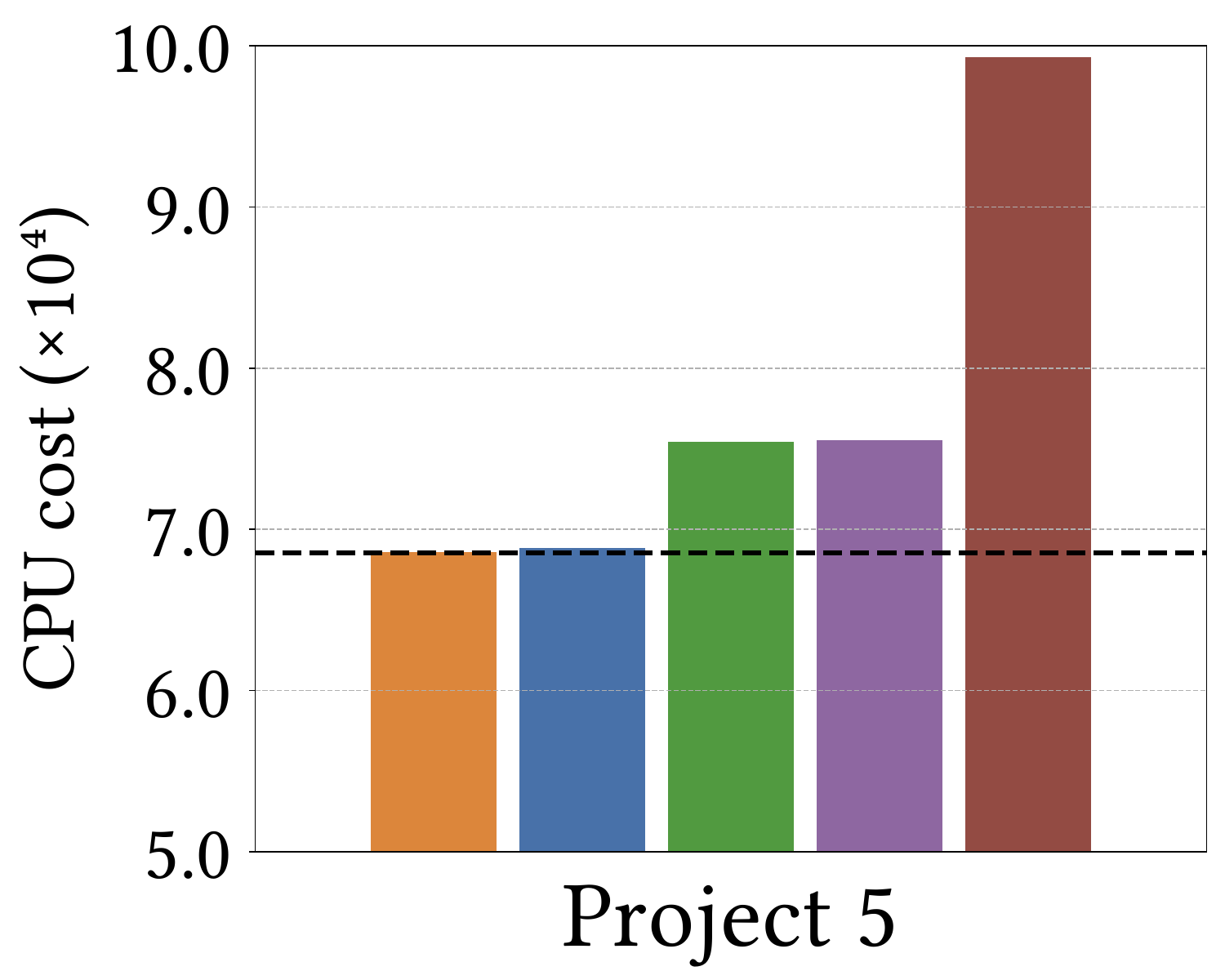}
\vspace{-1em}
\caption{Average CPU cost of learned query optimizers and \MC.}
\label{fig:e2e evaluation}
\vspace{0.5em}
    \includegraphics[width=0.19\linewidth]{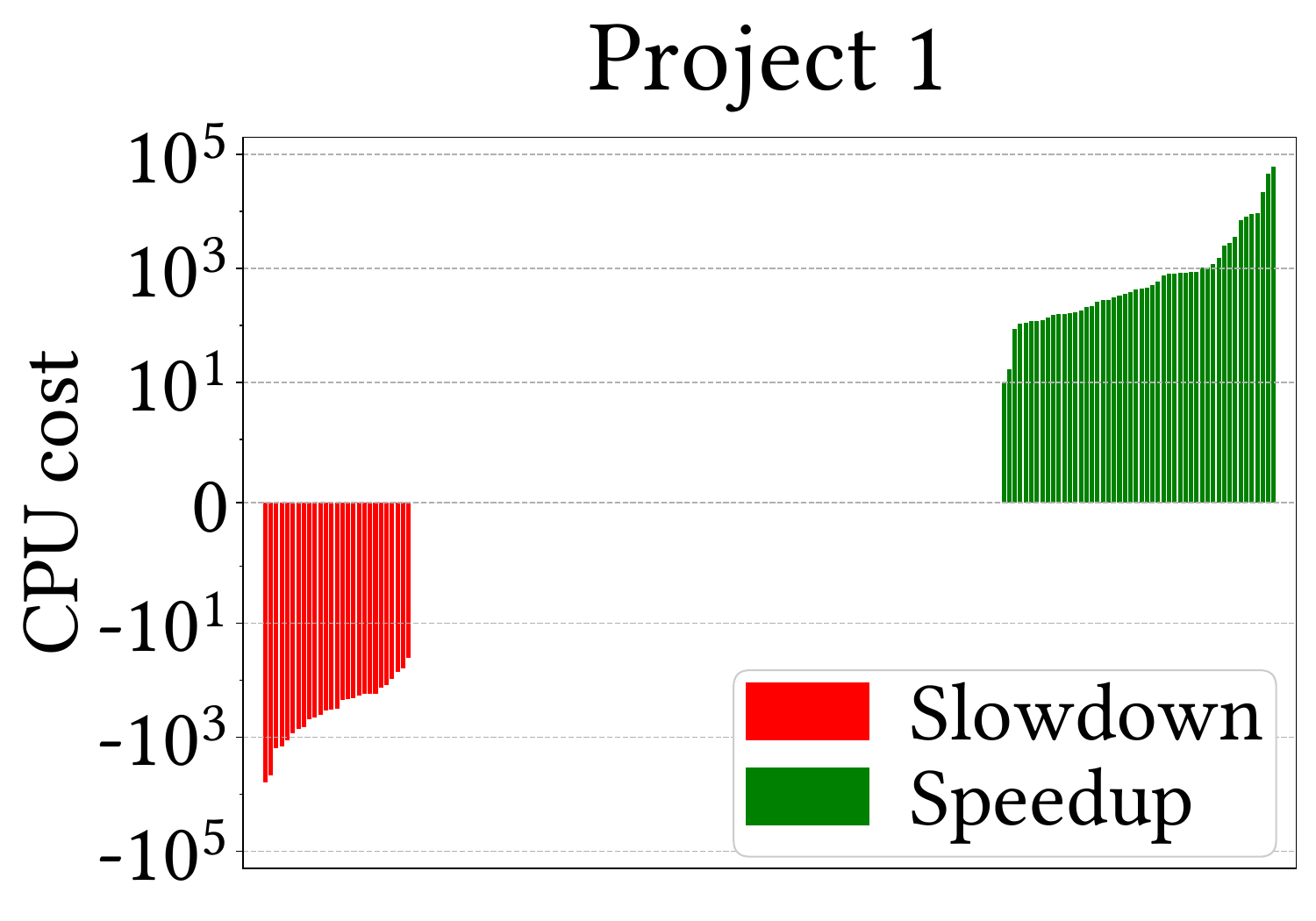}
    \includegraphics[width=0.19\linewidth]{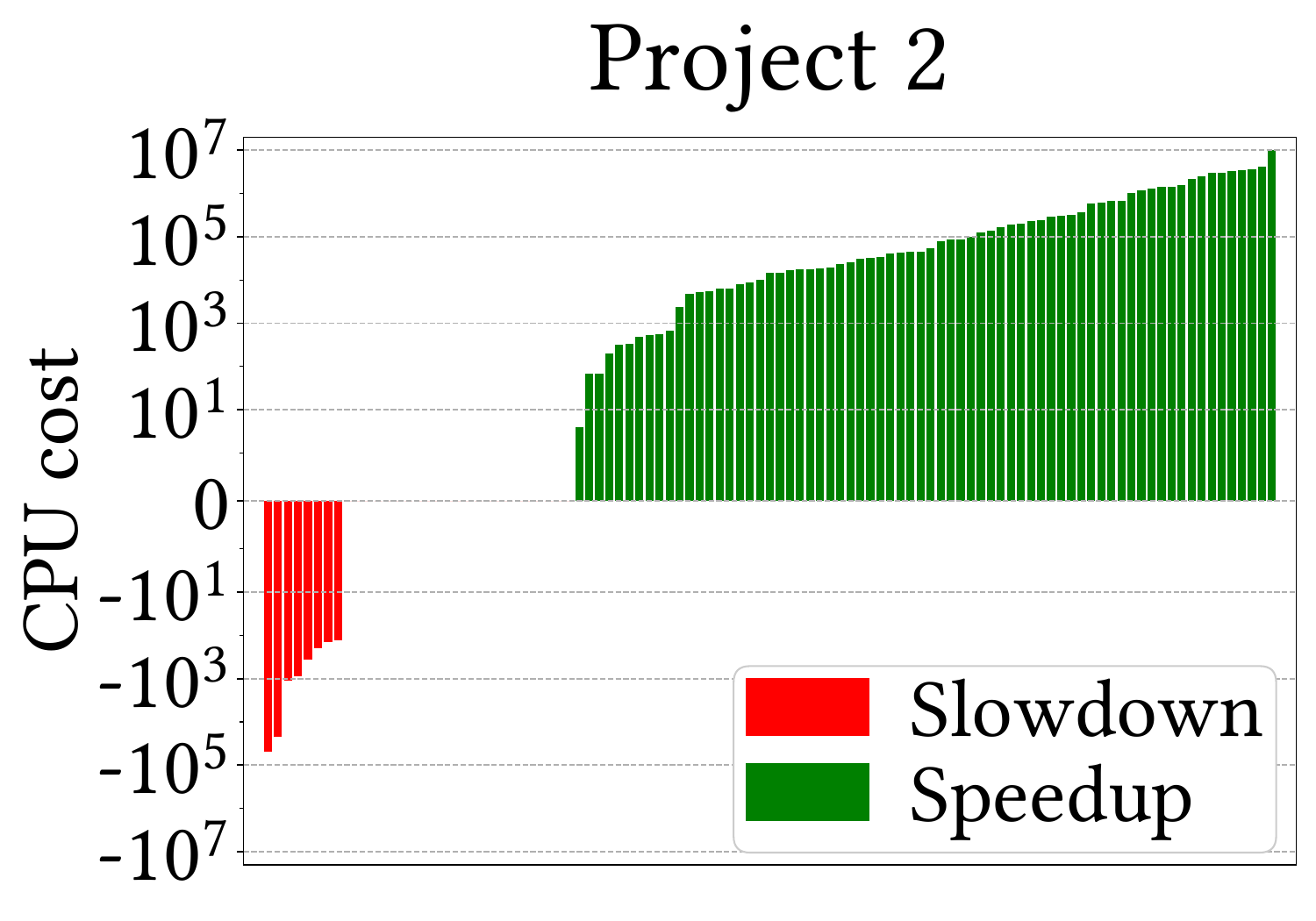}
    \includegraphics[width=0.19\linewidth]{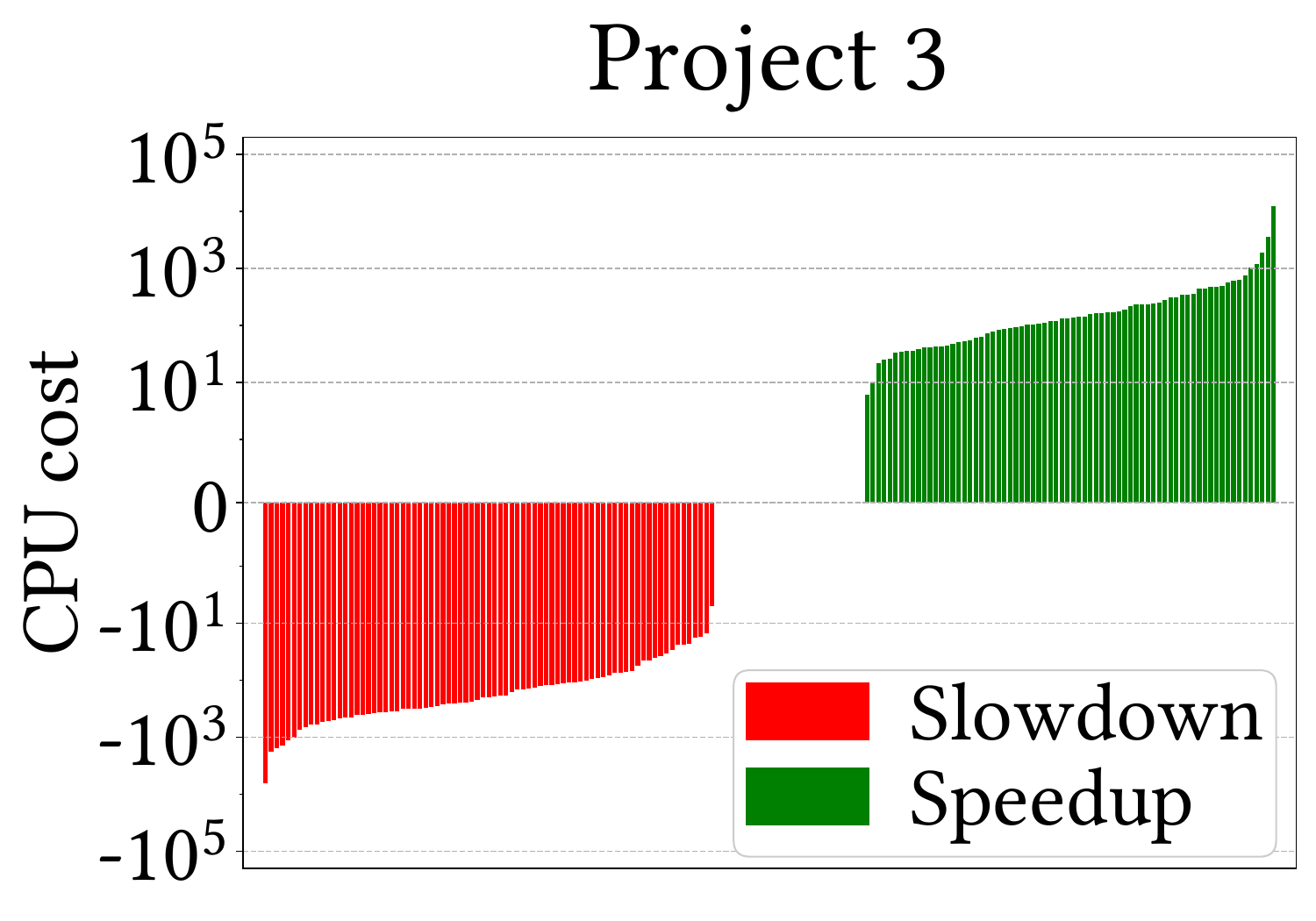}
    \includegraphics[width=0.19\linewidth]{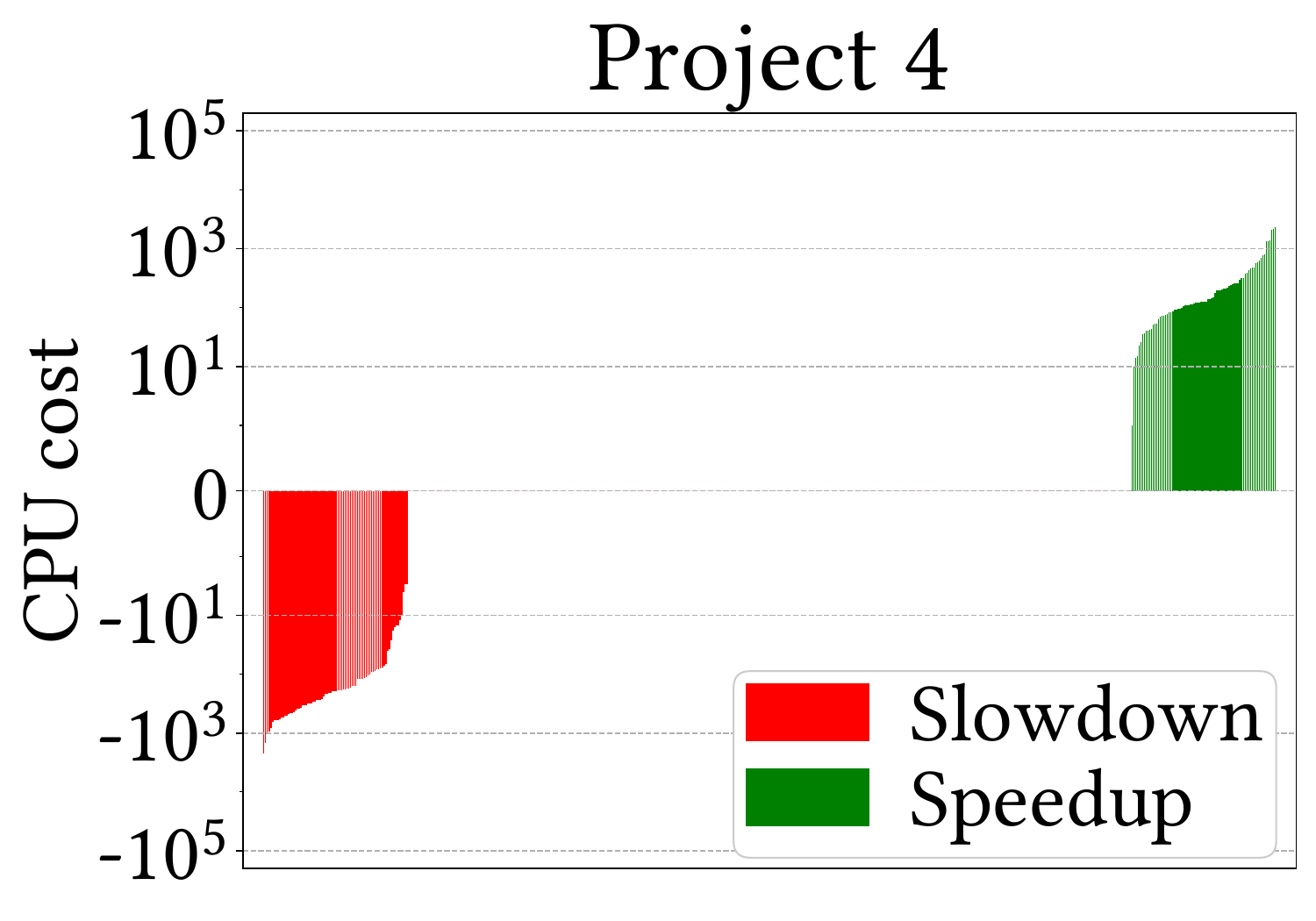}
    \includegraphics[width=0.19\linewidth]{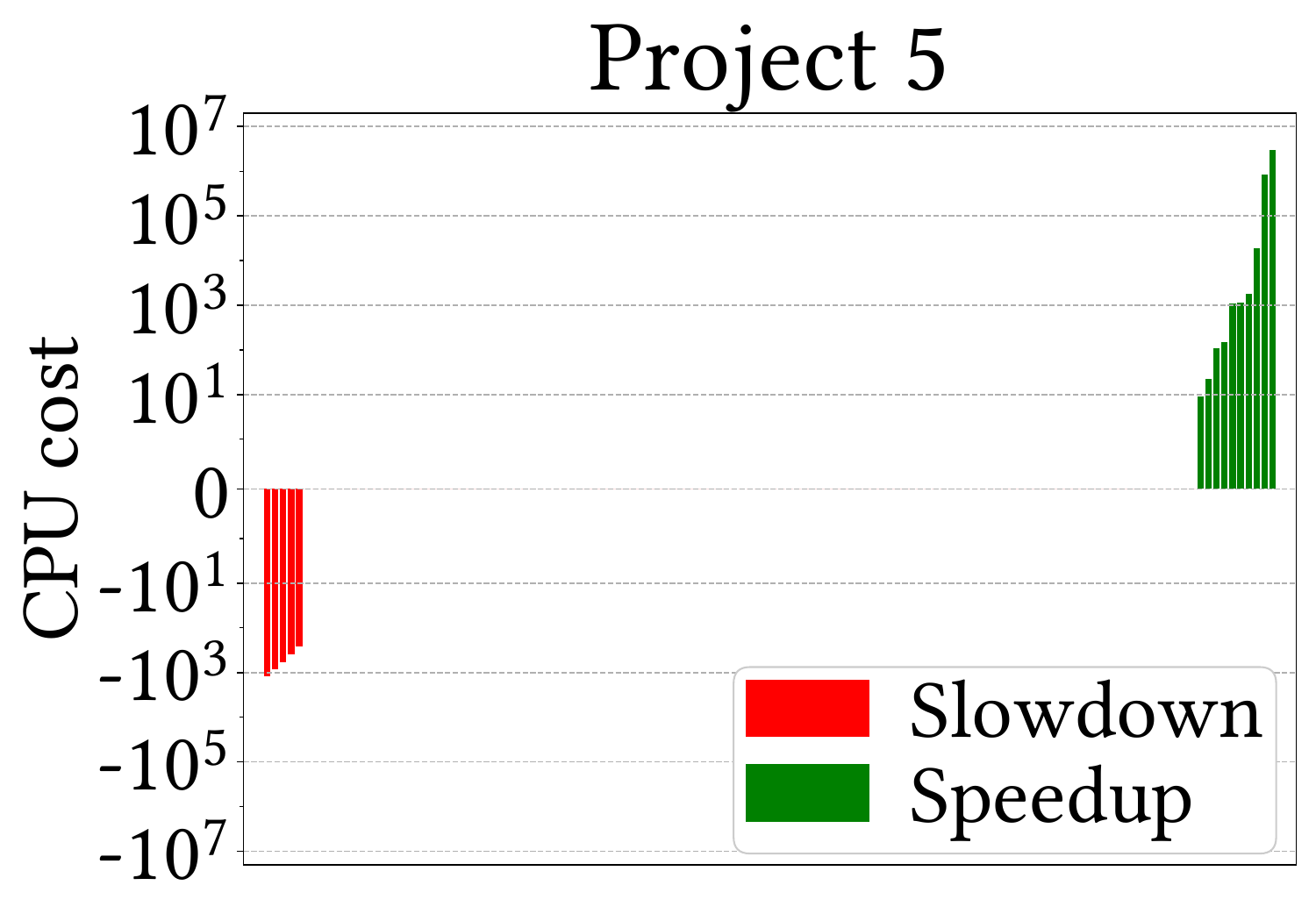}
\vspace{-1em}
\caption{Per-query execution cost of \sys in comparison with \MC.}
\label{fig:per-query}
\vspace{0.5em}
    \centering
    \includegraphics[width=0.45\textwidth]{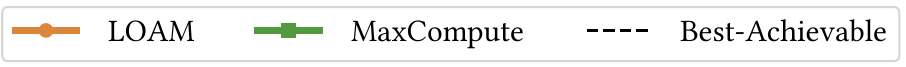} 
    
    \includegraphics[width=0.19\linewidth]{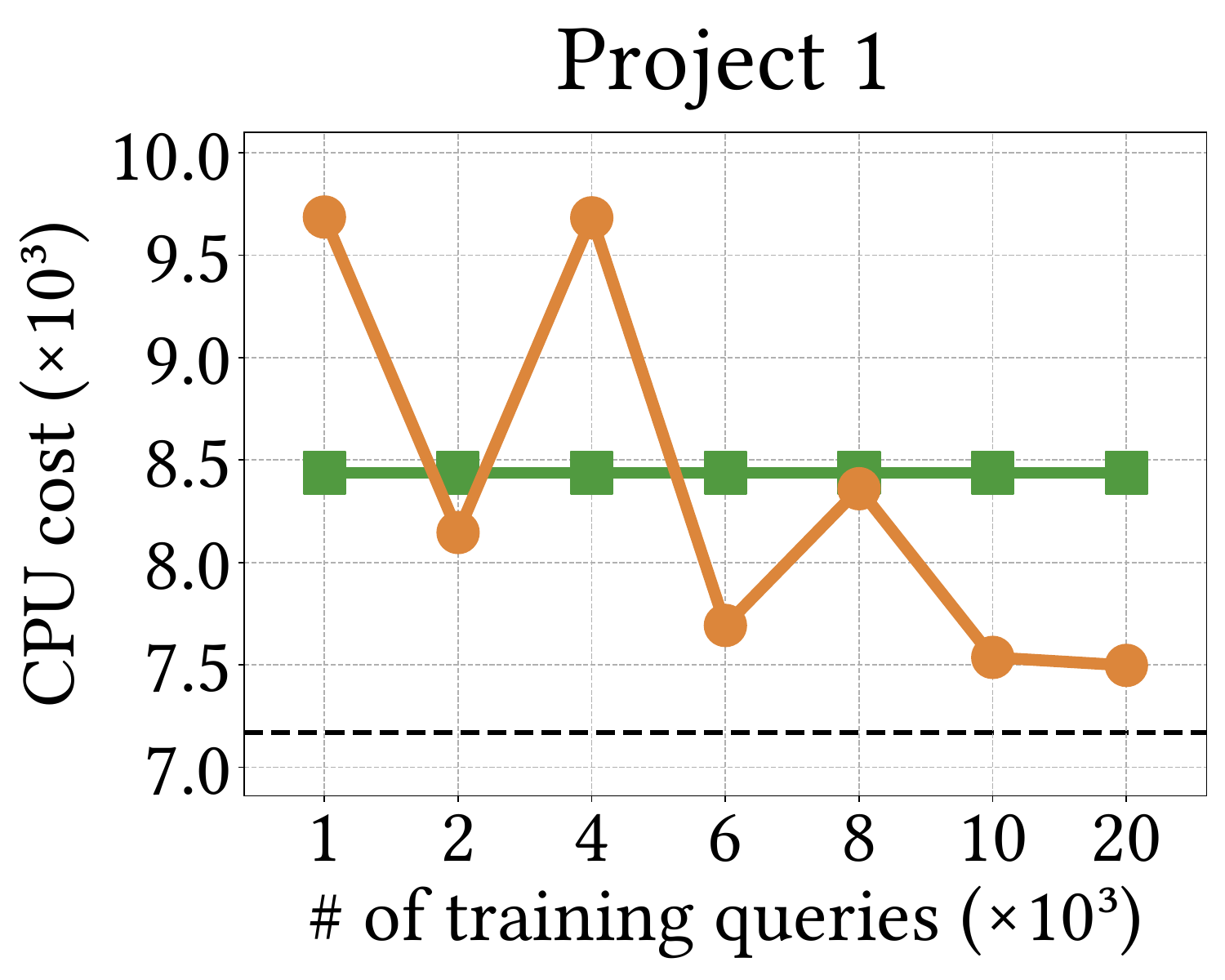}
    \includegraphics[width=0.19\linewidth]{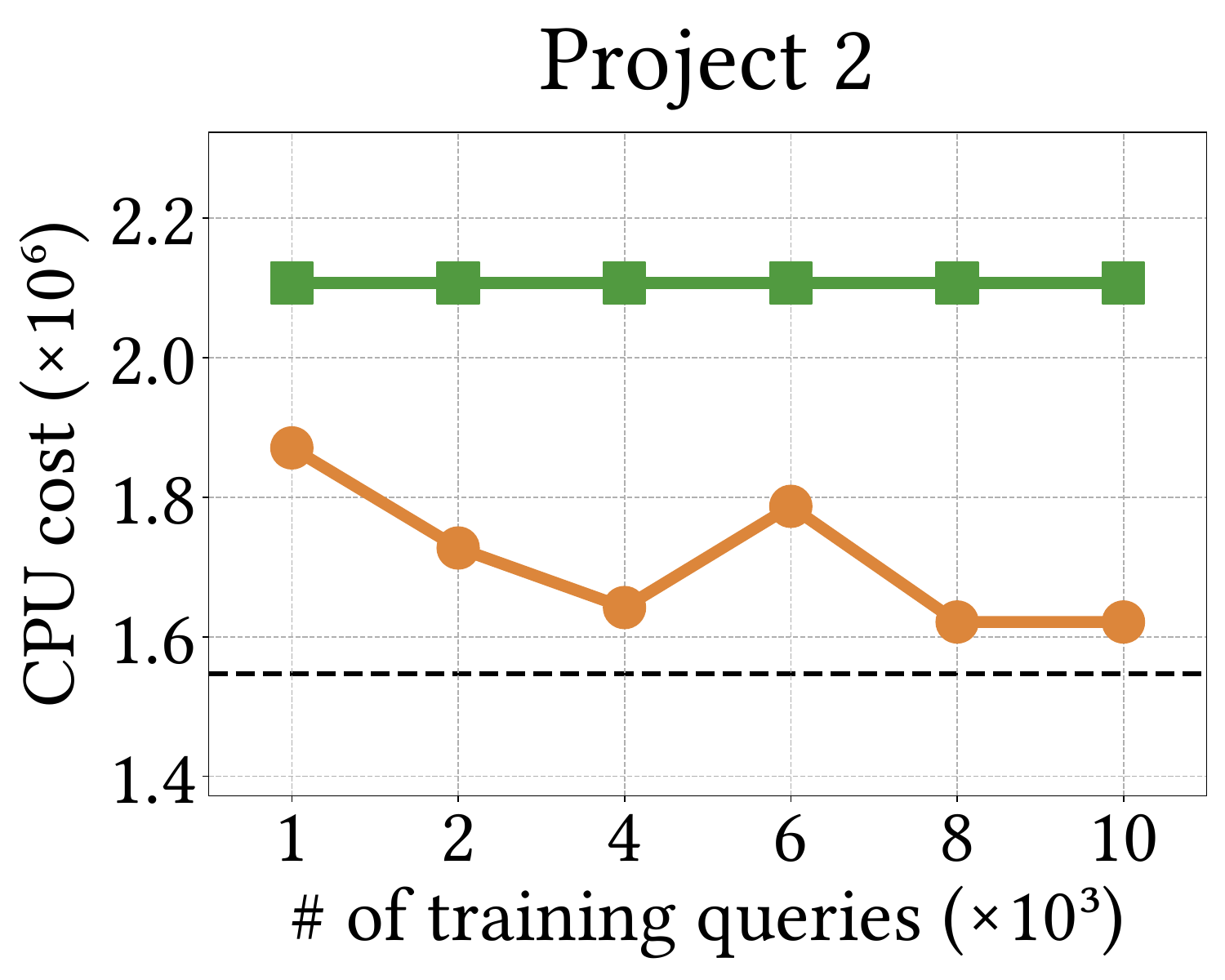}
    \includegraphics[width=0.19\linewidth]{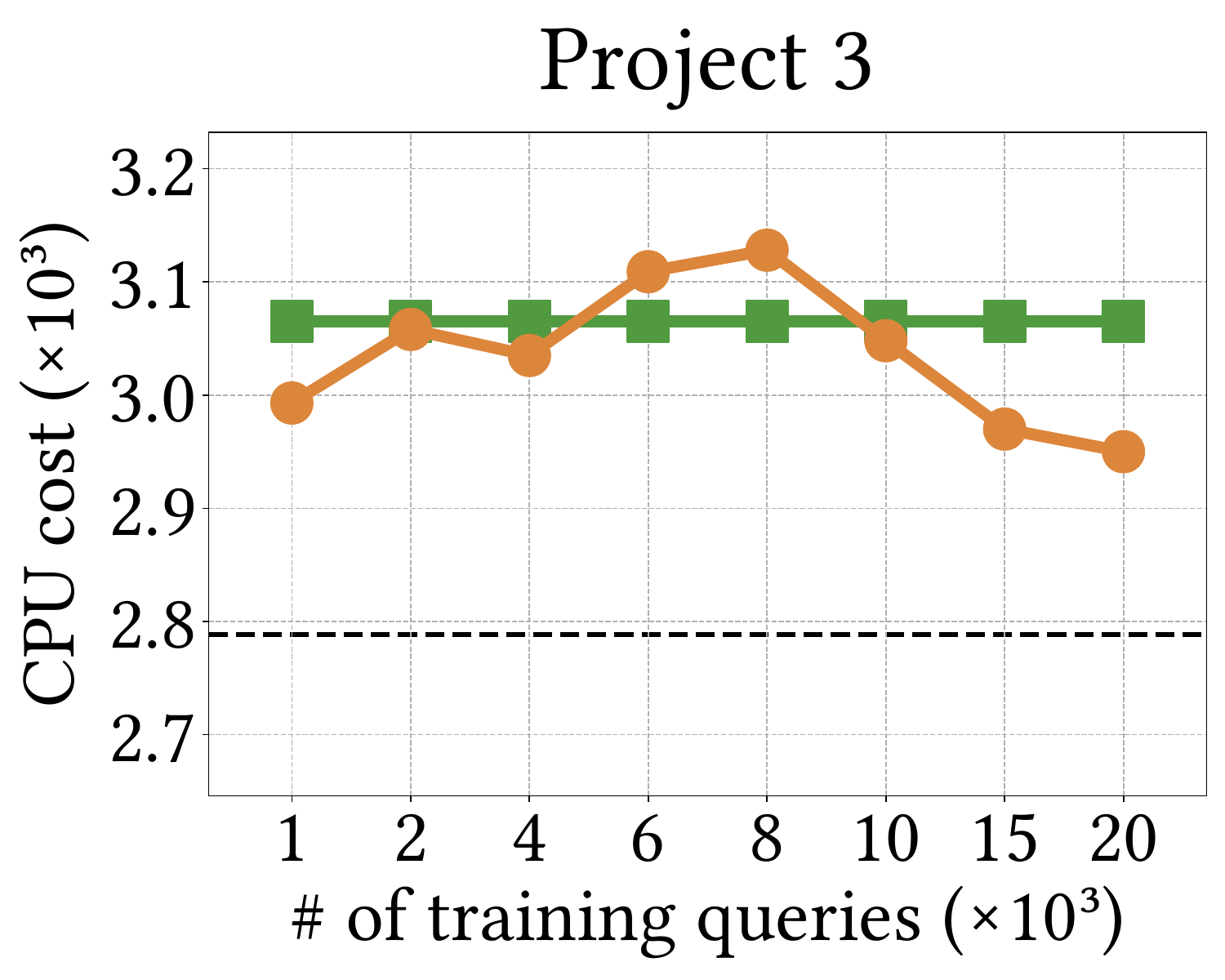}
    \includegraphics[width=0.19\linewidth]{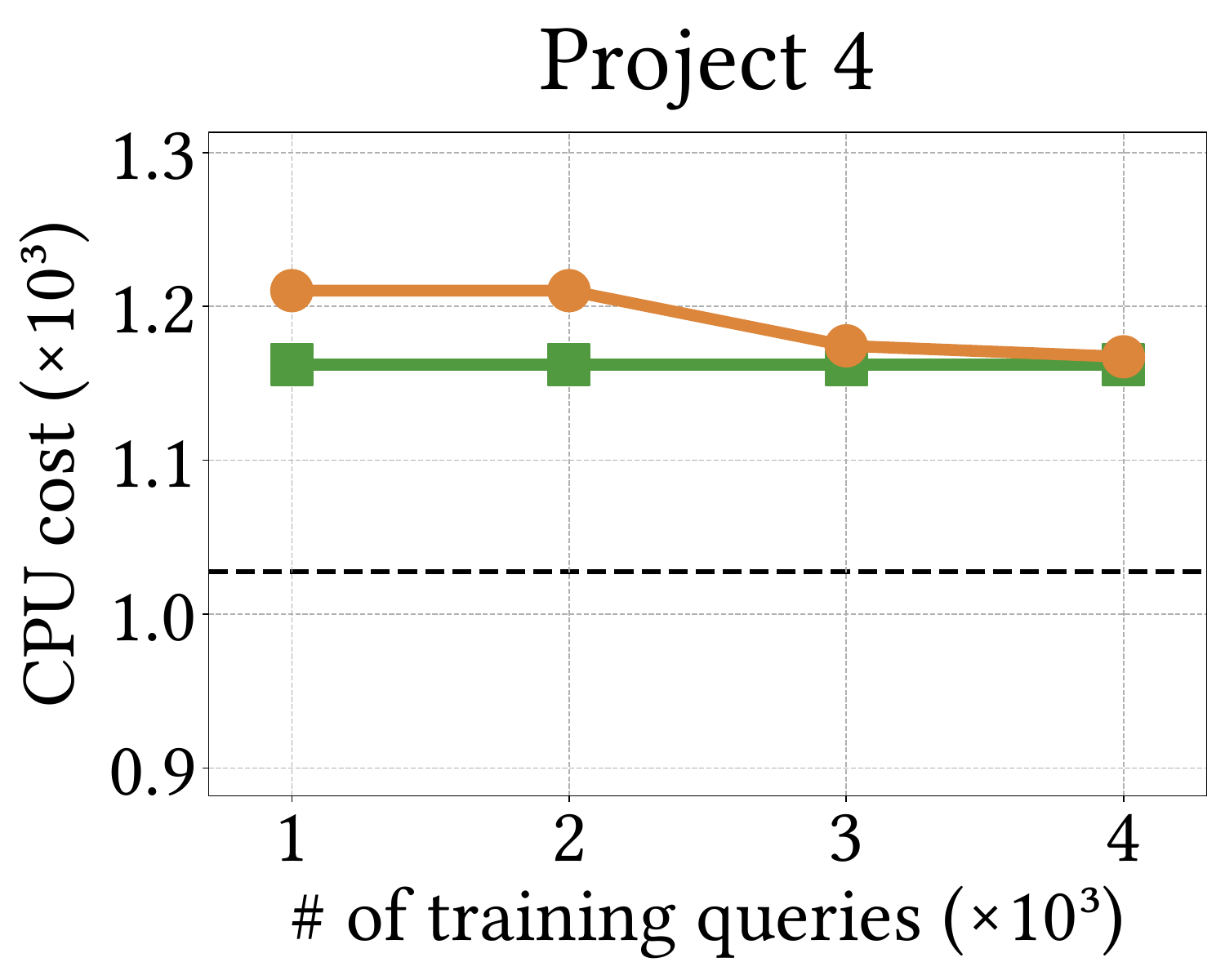}
    \includegraphics[width=0.19\linewidth]{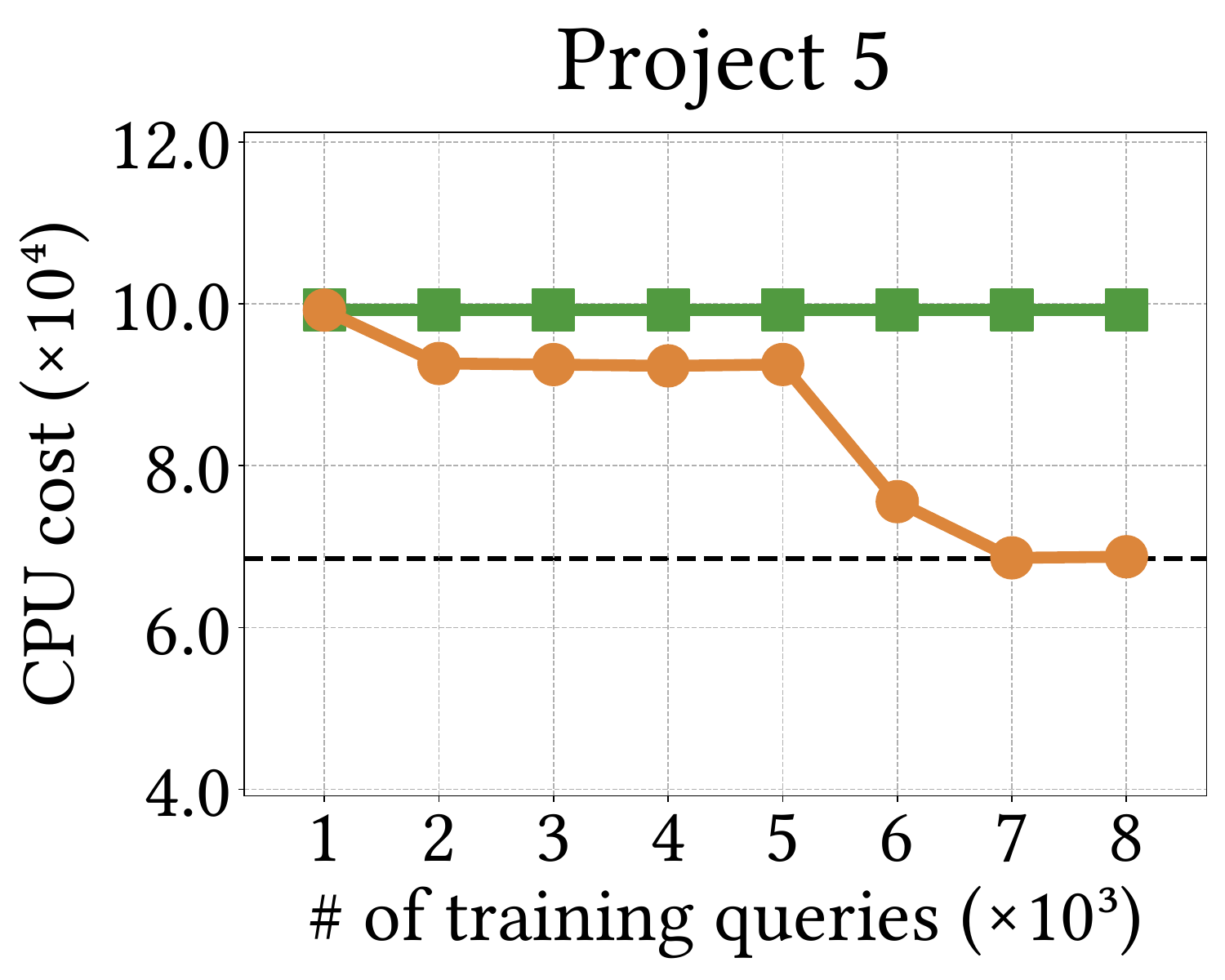}
\vspace{-1em}
\caption{Performance of \sys \wrt training data size.}
 \label{fig:training_size}
\end{minipage}
\begin{minipage}[c]{0.33\linewidth}
\centering
\vspace{1.5em}
\renewcommand{\arraystretch}{1.2}
\scalebox{0.6}{
\begin{tabular}{|c|ccccc|} 
    \hline
    \rowcolor{mygrey}
    & & \multicolumn{3}{c}{\textbf{\textsf{Training Time (s)}}} & \\
    \rowcolor{mygrey} \multirow{-2}*{\textbf{\textsf{Method}}} & {\citylife} & {\gd} & {\fin} & {\ads} & {\mkt} \\
    \hline
    \centering
    \textbf{\sys} & 2,950 &389 & 3,324 & 539 & 3,590 \\
    \textbf{\textsf{Transformer}} & 1,588 & 898 & 3,645 & 416 & 129 \\
    \textbf{\textsf{GCN}} & 206 & 103 & 364 & 115 & 93 \\
    \textbf{\textsf{XGBoost}} & 0.12 & 0.07 &0.14 & 0.05 & 0.11 \\
    \hline
    \end{tabular}
    }
    
    \vspace{0.5em}
    {\small\textsf{(a) Training overhead.}}
    \vspace{0.5em}
    
    \scalebox{0.6}{
    \begin{tabular}{|c|ccccc|} 
    \hline
     \rowcolor{mygrey}
     & & \multicolumn{3}{c}{\textbf{\textsf{Model Size (MB)}}} & \\
    \rowcolor{mygrey} \multirow{-2}*{\textbf{\textsf{Method}}} & {\citylife} & {\gd} & {\fin} & {\ads} & {\mkt} \\
    \hline
    \centering
    \textbf{\sys} & 15 &8.1 & 21 & 6.4 & 9.1 \\
    \textbf{\textsf{Transformer}} & 25.5 &25.4 &25.9 &25.2 &25.1 \\
    \textbf{\textsf{GCN}} & 153 & 140 & 155 & 124 & 118 \\
    \textbf{\textsf{XGBoost}} & 11 &5.1 &17.8 &3.4 &2.8 \\
    \hline
    \end{tabular}
    }
    
    \vspace{0.5em}
    {\small \textsf{(b) Model footprint.}}
    \vspace{0.5em}
    
    \scalebox{0.59}{
    \begin{tabular}{|c|ccccc|} 
    \hline
     \rowcolor{mygrey}
     & & \multicolumn{3}{c}{\textbf{\textsf{Inference Time (s)}}} & \\
    \rowcolor{mygrey} \multirow{-2}*{\textbf{\textsf{Method}}} & {\citylife} & {\gd} & {\fin} & {\ads} & {\mkt} \\
    \hline
    \centering
    \textbf{\sys} & 0.28 & 0.41 & 0.27 & 0.11 & 0.29 \\
    \textbf{\textsf{Transformer}} & 0.26&0.16 &0.31 &0.15 &0.14 \\
    \textbf{\textsf{GCN}} & 0.21 &0.24 &0.17 &0.06 &0.20 \\
    \textbf{\textsf{XGBoost}} & 0.28 & 0.08 &0.13 &0.08 &0.12 \\
    \hline
    \end{tabular}
    }
    
    \vspace{0.5em}
    {\small \textsf{(c) Average inference time.}}
    \vspace{-1em}
\caption{Extra cost of learned optimizers.}
\label{tbl:extra_cost}
\end{minipage}
\vspace{-0.8em}
\end{figure*}

Figure~\ref{fig:e2e evaluation} illustrates the end-to-end performance of \sys and other baselines on the $5$ selected projects. We make the following observations:  

1) \sys outperforms all baseline optimizers across almost all projects. In particular, it shows substantial improvements over \MC's native optimizer on \textsf{Projects~1, 2}, and \textsf{5}, achieving performance gains of $10\%$, $23\%$, and $30\%$, respectively. This near-universal superiority highlights the practical effectiveness of \sys in real-world production environments, where reliable performance is particularly important. By contrast, other learned optimizer baselines demonstrate limited improvements or even degraded performance compared to \MC. We attribute \sys's advantage to its adaptive training process, which enables it to generalize well to candidate plans, whereas other learned optimizers suffer considerably from distribution shifts that compromise cost estimation accuracy.


2) Across all projects, the actual performance improvements achieved by learned optimizers (including \sys) are strongly correlated with the theoretical bound on the improvement space, \ie, $D(M_d)$. Specifically, on \textsf{Projects~3} and \textsf{4}, which exhibit relatively smaller $D(M_d)$ (at $20\%$ and $23\%$ of the oracle model's cost, compared to $43\%$ and $40\%$ for \textsf{Projects~2} and \textsf{5}, respectively), the learned optimizers yield only performance comparable to the \MC's native query optimizer. This confirms that $D(M_d)$ is an effective metric for identifying unpromising projects in advance. Further empirical support for this observation is provided in Section~\ref{subsec:load-effect}. Consistent with Theorem~\ref{Thm:perfgap}, the best-achievable model (the dashed line in the figure) exhibits a cost deviance $D(M_b)$ close to $D(M_d)$ on \textsf{Projects~3} and \textsf{~4}, with gaps of approximately $7\%$, whereas the gaps for other projects range from $12\%$ to $33\%$.

3) \textsf{Project~1} has a $D(M_d)$ value of $25\%$, only marginally higher than those of \textsf{Projects~3} and \textsf{4}. Nevertheless, it achieves substantially better performance. This is because \textsf{Projects~3} and \textsf{~4} additionally suffer from insufficient training data, which hinders the training of an effective cost predictor. Specifically, \textsf{Project~3} contains over 7,000 columns, roughly twice as many as others, which accordingly necessitates more training queries for \sys to learn the data distributions. Meanwhile, \textsf{Project~4} has only about $4,000$ training queries, far fewer than the others.


Notably, these five projects exhibit the largest improvement space among all 30 sampled projects, yet some still do not yield clear performance gains due to theoretical and practical constraints. This further underscores the importance of project selection studied in Section~\ref{sec:project-selection} and carefully evaluating the performance of learned query optimizers before deploying in practical scenarios.

Despite \sys's superior performance in query optimization, it incurs minimal deployment and negligible operational overhead in \MC. As shown in Figure~\ref{tbl:extra_cost}, \sys requires less than an hour of training time and a memory footprint of about $20$ MB. During query optimization, it generates candidate plans in under $0.1$ seconds and adds model inference costs of $0.1$-$0.5$ seconds per query, amounting to only $0.23\%$--$0.74\%$ of the execution time for test queries drawn from production workloads.

\subsubsection{\textbf{Analysis of Performance Improvement/Regression}}
\label{subsec:per-query-evaluation}


To provide more detail, Figure~\ref{fig:per-query} compares the per-query CPU cost of \sys and \MC. Test queries are sorted by the cost difference (from slowdown to speedup) to visualize \sys's performance improvement or regression relative to \MC.

We observe that \sys brings much more performance gains than regressions on \textsf{Projects~1} and \textsf{2}. Specifically, only $26$ ($14\%$) and $8$ ($8\%$) queries experience slowdowns, respectively, while $50$ ($27\%$) and $70$ ($69\%$) queries see notable improvements. \textsf{Project~5} is largely flat, showing 10 speedups versus 5 slowdowns. Importantly, improvements also far exceed regressions in magnitude, with cost savings up to two orders of magnitude larger than the worst regressions across these projects.
This strong bias toward performance gains highlights \sys's effectiveness in identifying better execution plans for complex production queries.

By further statistics, we observe that among improved queries in these three projects, over half achieve relative performance improvements ranging from $17\%$ to $26\%$. This indicates that performance gains delivered by \sys are not driven by marginal optimizations on long-running queries, but instead represent significant and widespread improvements across a diverse range of queries.

Results on \textsf{Projects~3} and \textsf{~4} exhibit a different pattern: \sys incurs regressions that slightly outnumber or match improvements. Similar trends are observed across all learned optimizer baselines. This aligns with our findings in the end-to-end evaluation in Section~\ref{subsec:e2e-evaluation}, where learned optimizers struggle to achieve clear performance gains over \MC when the improvement space is limited and the training data is insufficient. These results reaffirm the importance of preemptively filtering out unpromising projects and ensuring the availability of sufficient training data.

\begin{figure*}
\begin{minipage}[c]{0.65\textwidth}
\centering
    \includegraphics[width=0.8\textwidth]{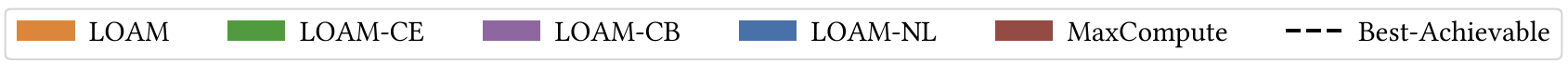} 
        
    \includegraphics[width=0.19\linewidth]{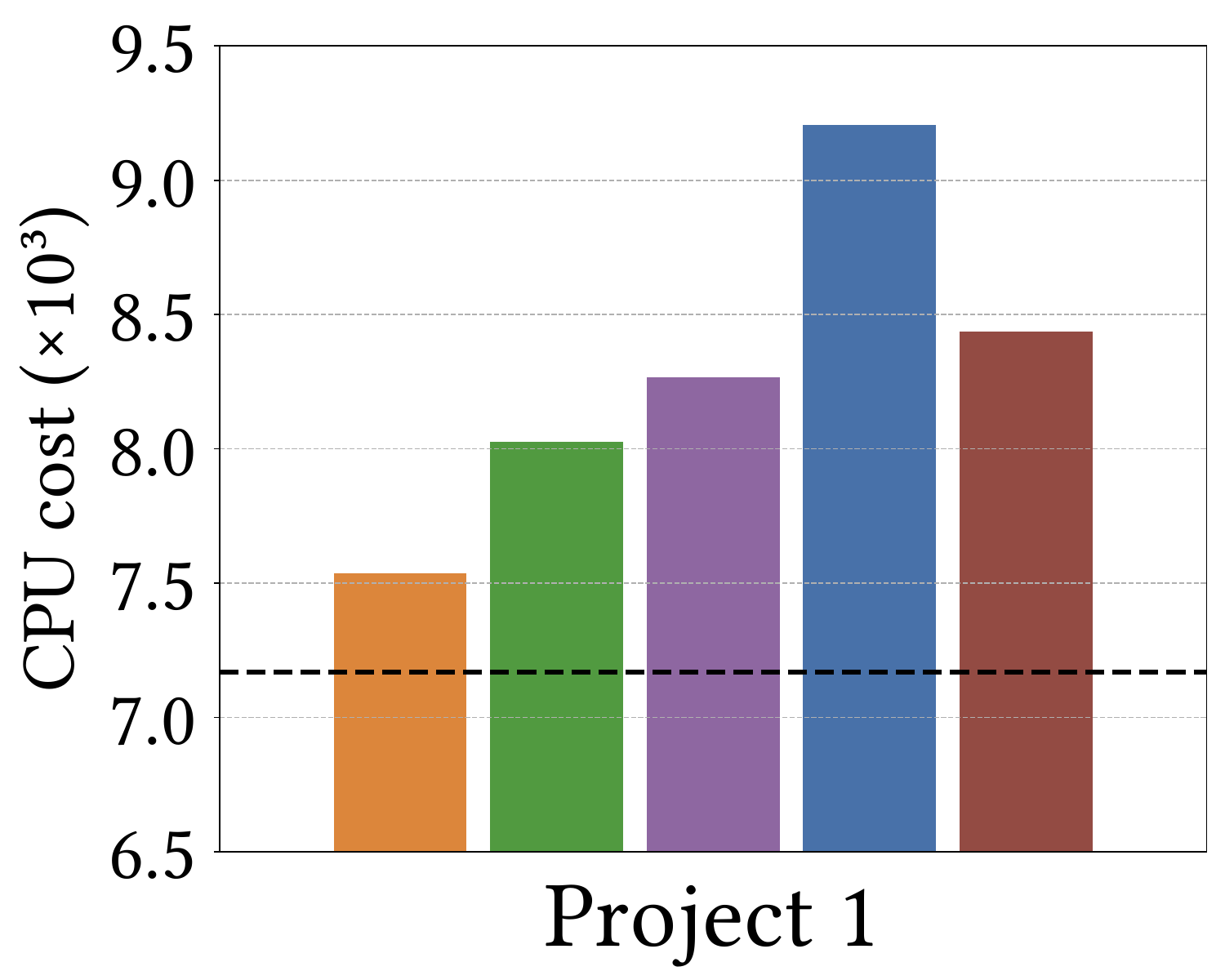}
    \includegraphics[width=0.19\linewidth]{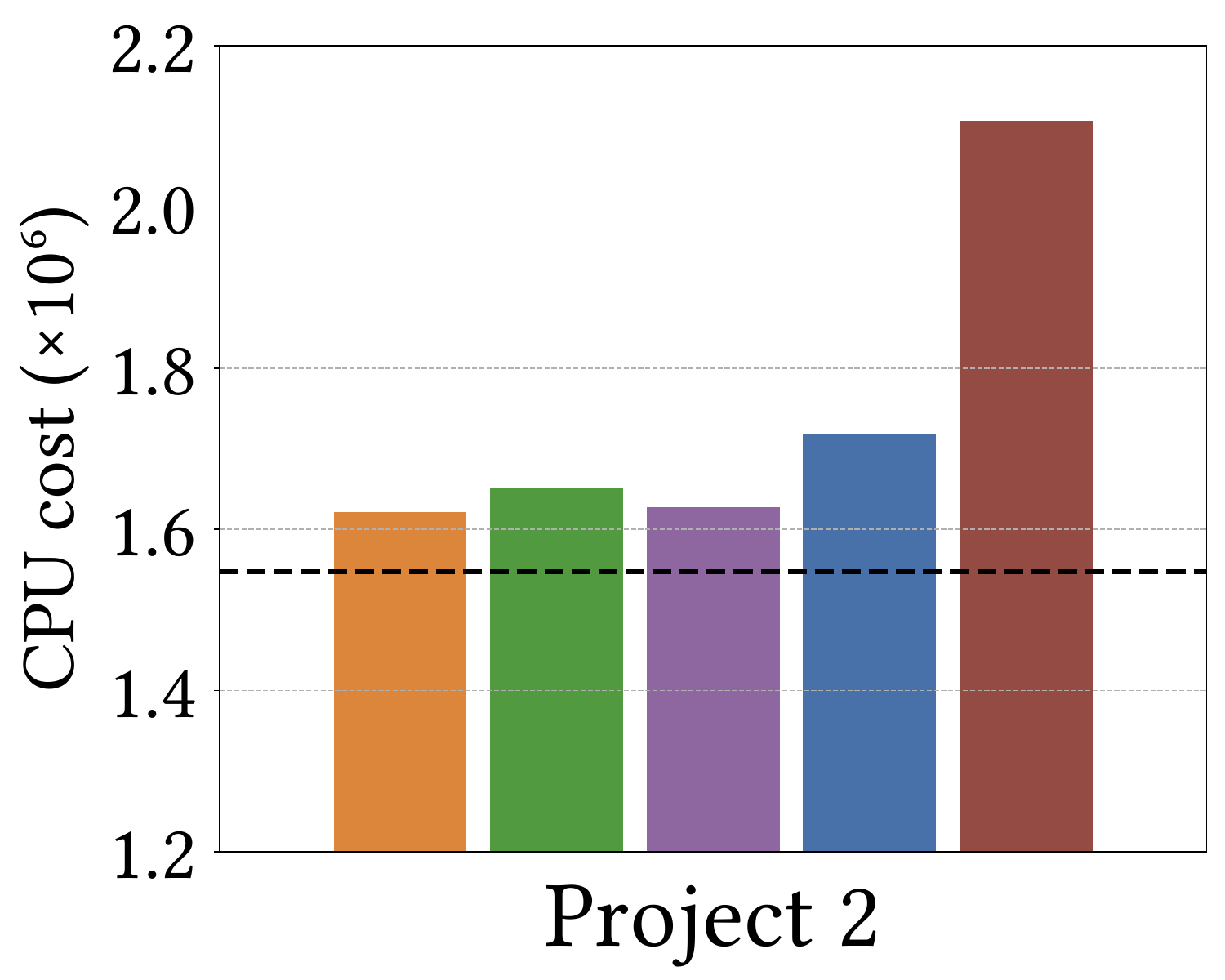}
    \includegraphics[width=0.19\linewidth]{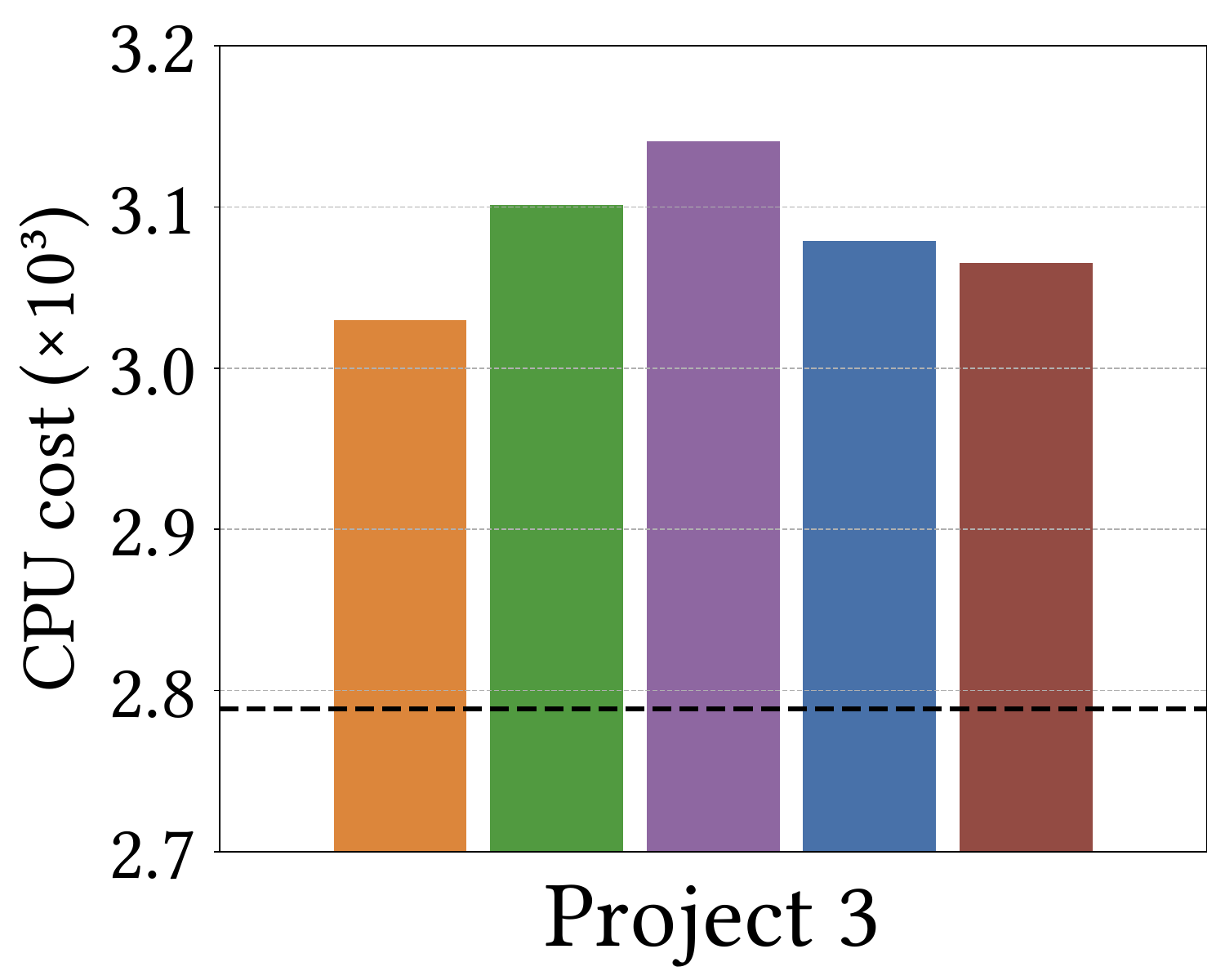}
    \includegraphics[width=0.19\linewidth]{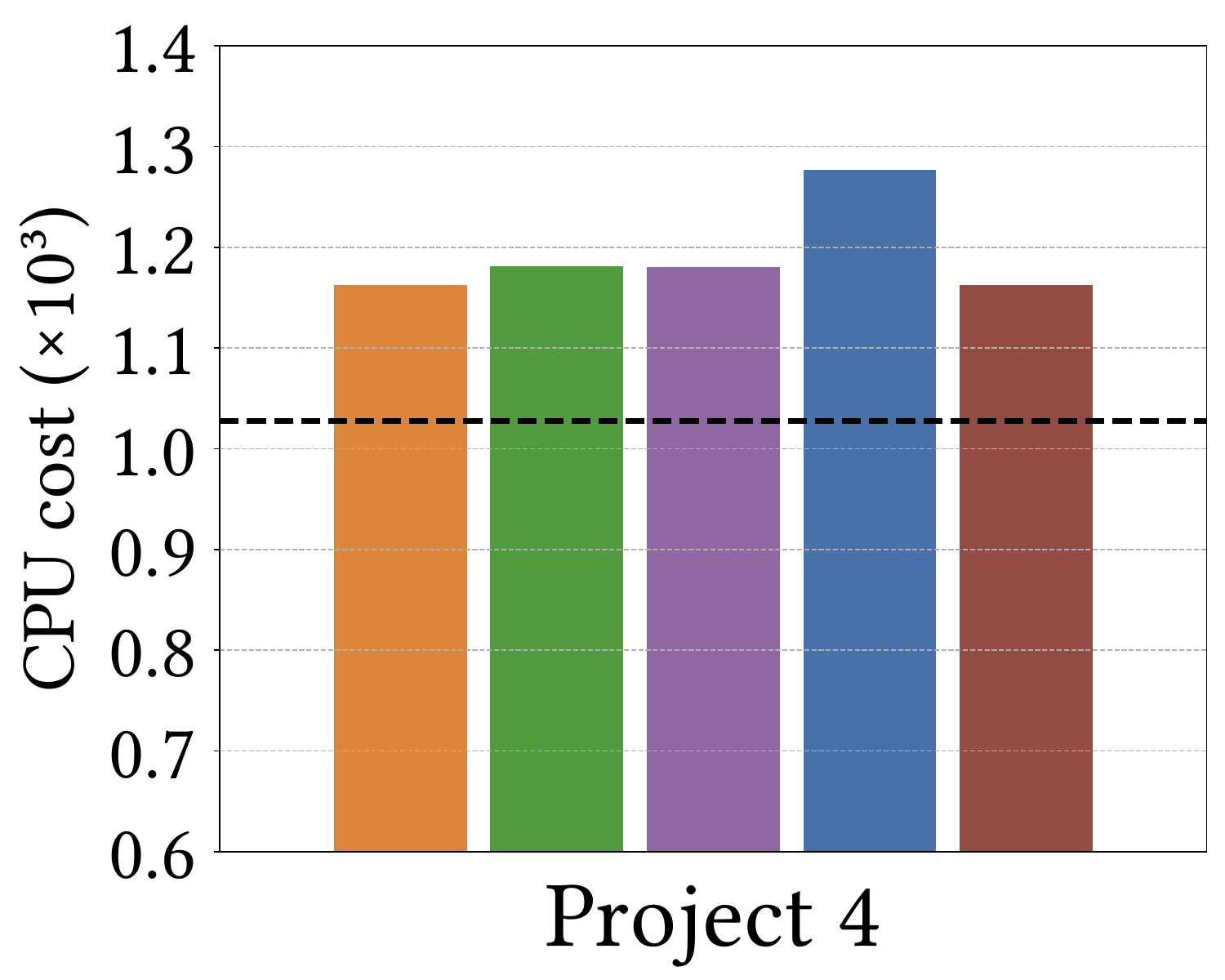}
    \includegraphics[width=0.19\linewidth]{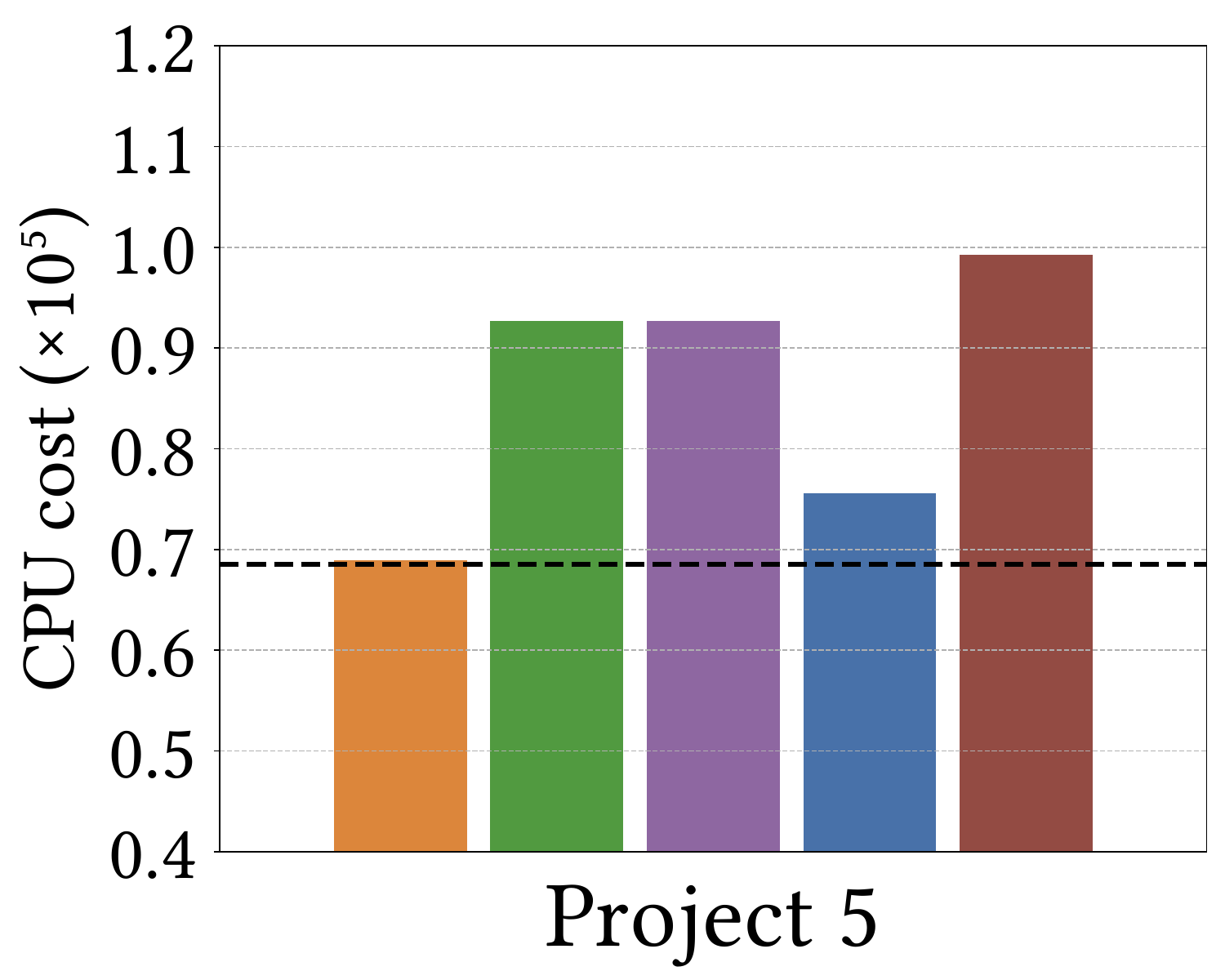}

    {\small \textsf{(a) E2E CPU cost.}}
    \vspace{0.3em}
    
    \includegraphics[width=0.19\linewidth]{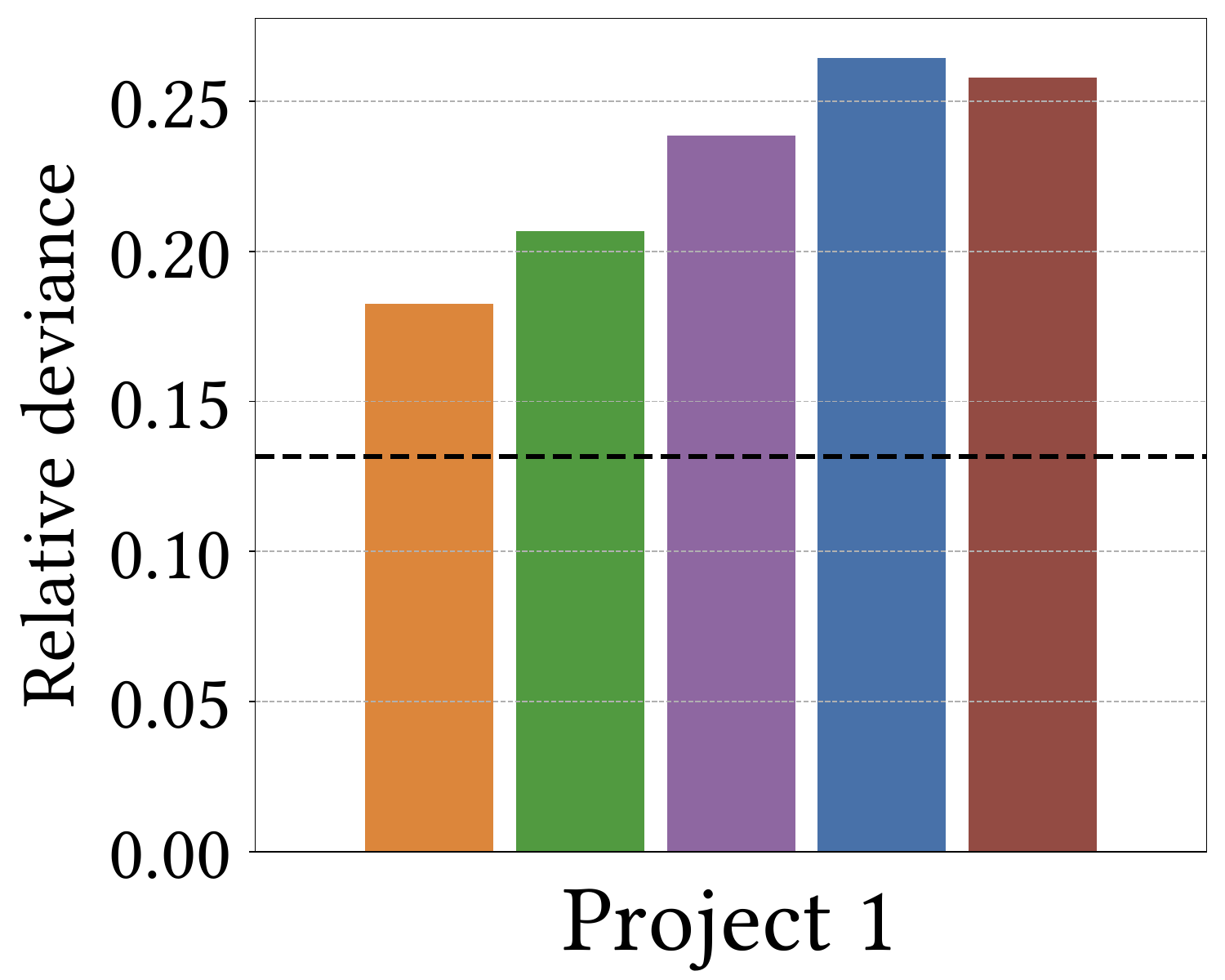}
    \includegraphics[width=0.19\linewidth]{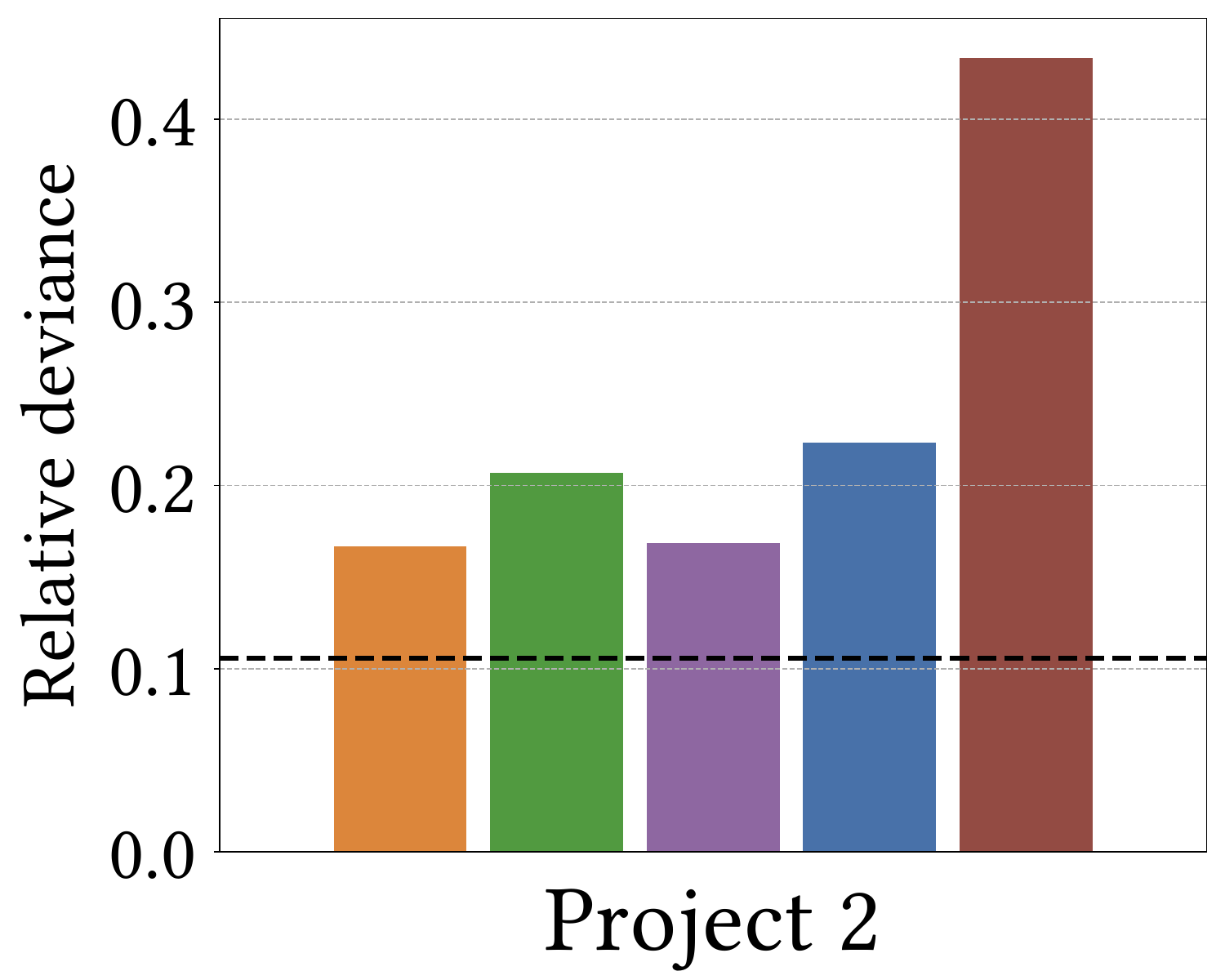}
    \includegraphics[width=0.19\linewidth]{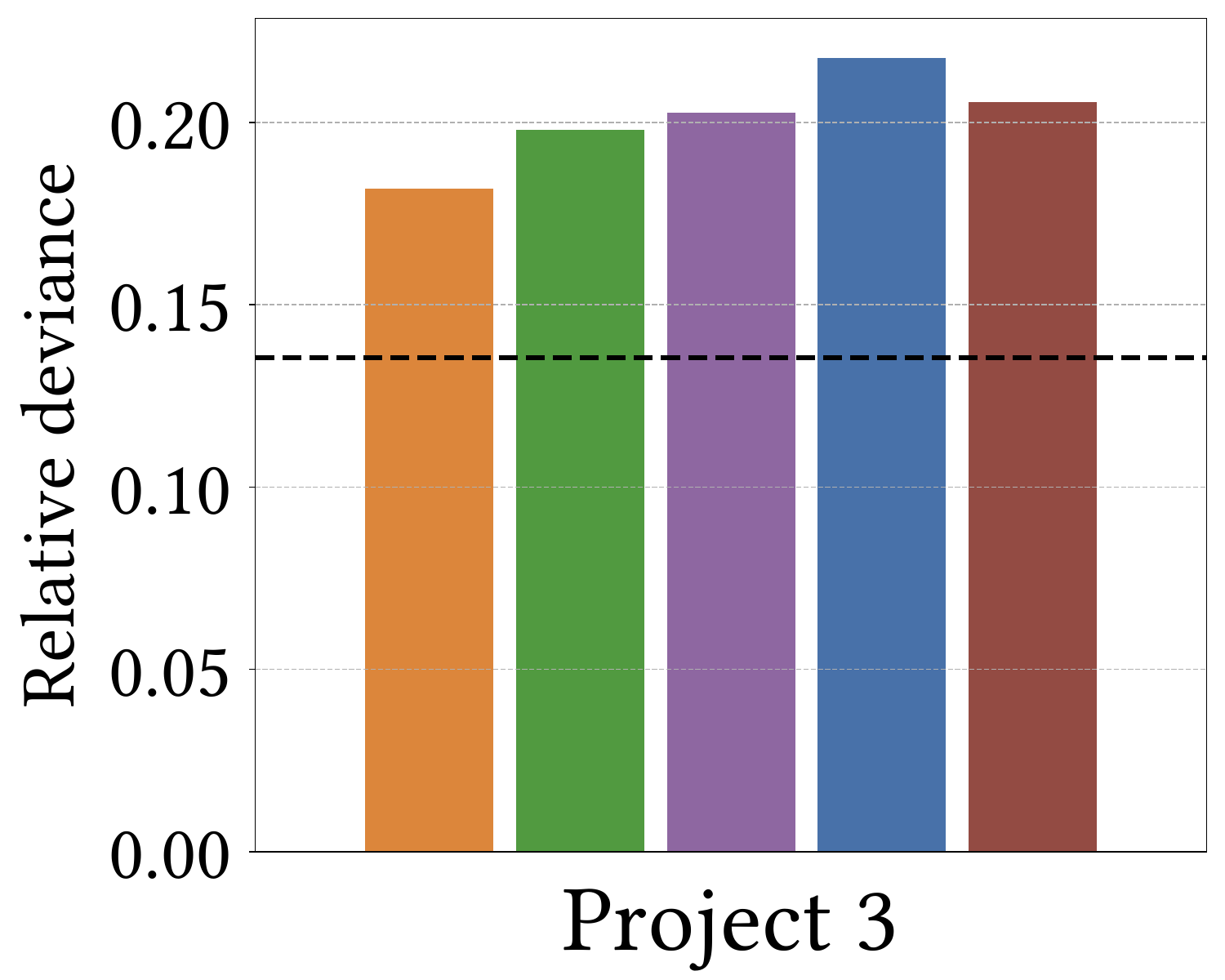}
    \includegraphics[width=0.19\linewidth]{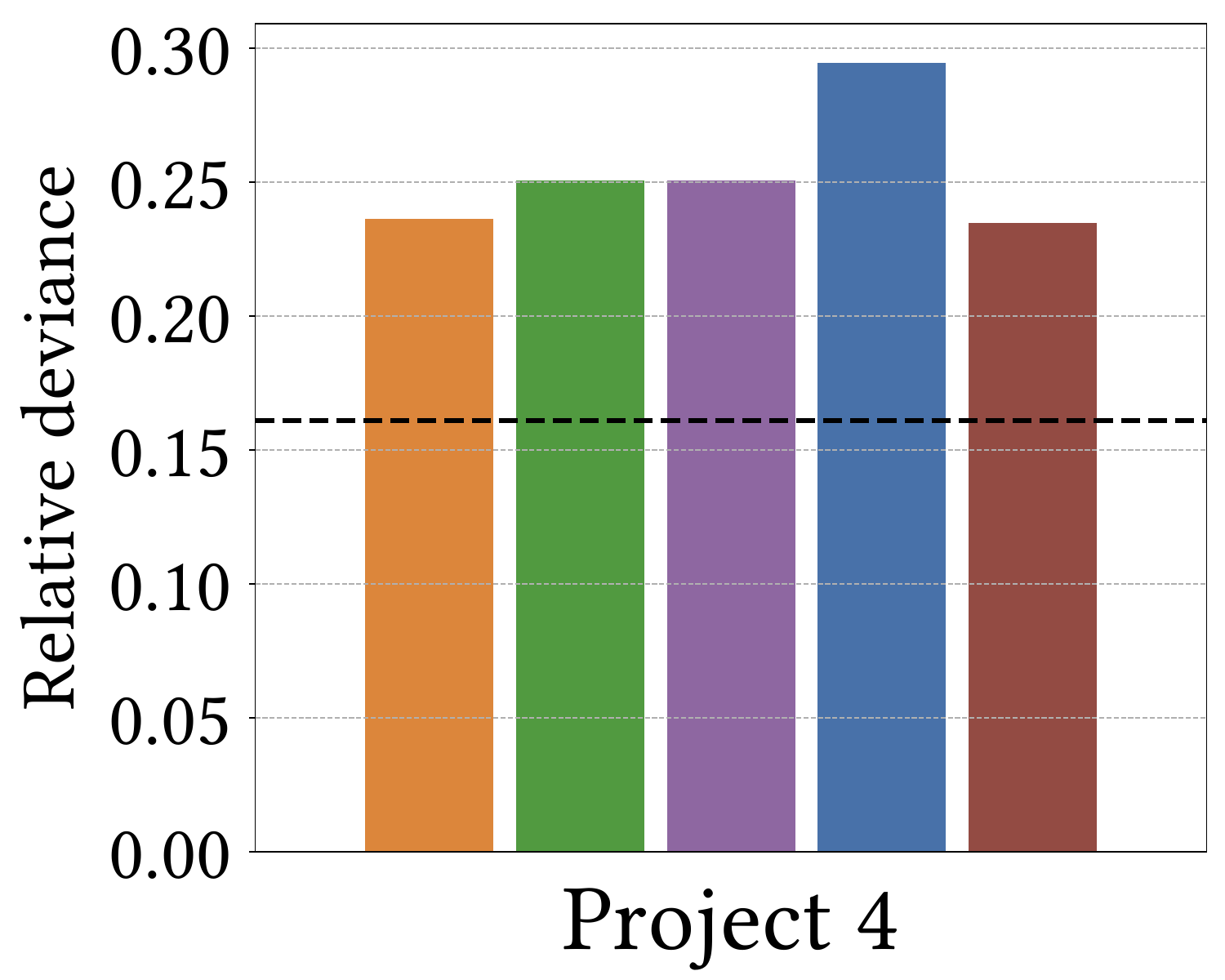}
    \includegraphics[width=0.19\linewidth]{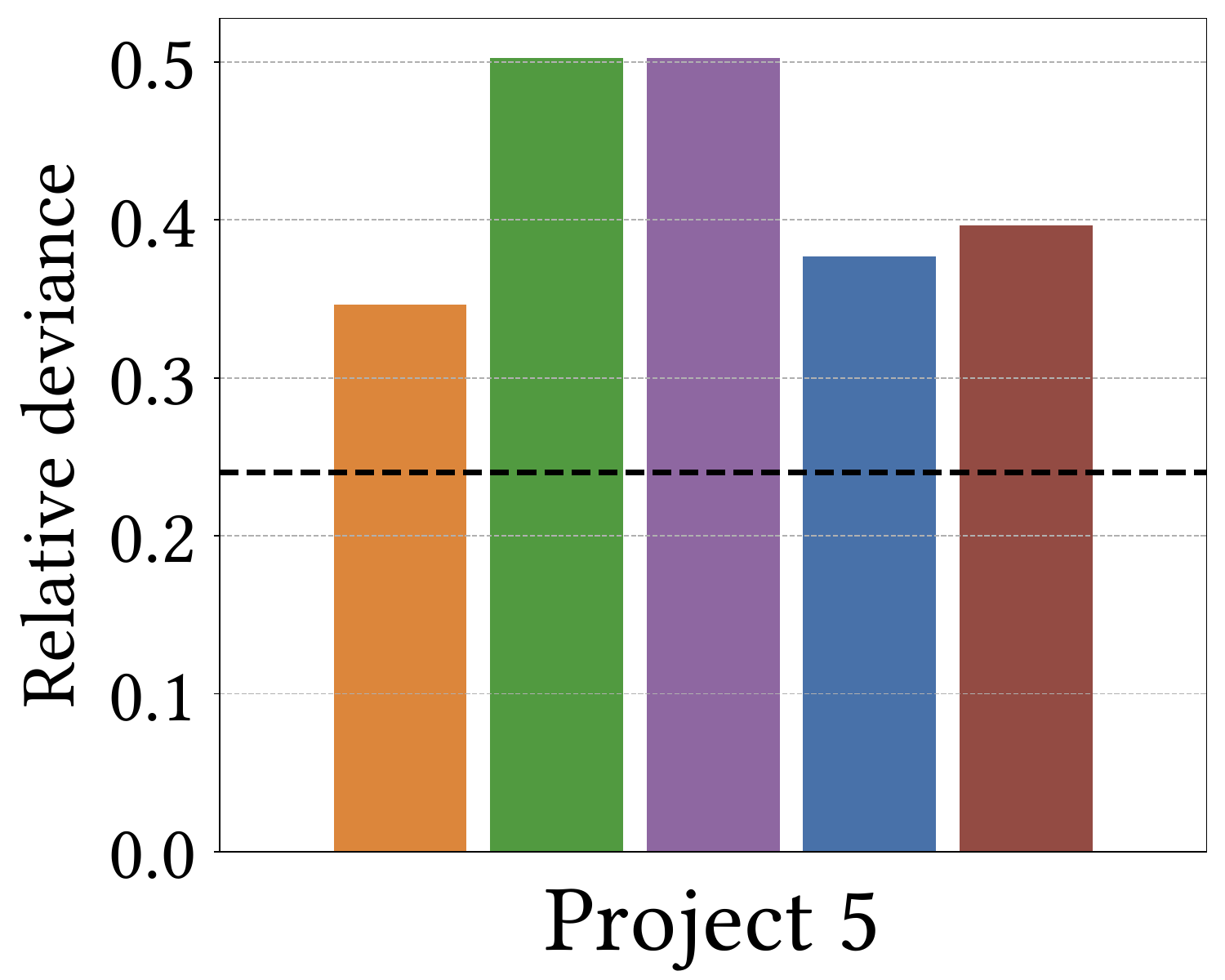}

    {\small \textsf{(b) Relative deviance.}}
    \vspace{0.3em}
    
\vspace{-1em}
\caption{Query optimization performance \wrt cost inference methods.}
\label{fig:error-wrt-load}
\end{minipage}
\begin{minipage}[c]{0.33\linewidth}
\renewcommand{\arraystretch}{1.2}
\vspace{0.7em}
    \scalebox{0.6}{
    \begin{tabular}{|c|ccccc|} 
    \hline
    \rowcolor{mygrey}
     & & \multicolumn{3}{c}{\textbf{\textsf{Average CPU cost}}} & \\
    \rowcolor{mygrey} \multirow{-2}*{\textbf{\textsf{Method}}} & {\citylife} & {\gd} & {\fin} & {\ads} & {\mkt} \\
    \hline
    \centering
    \textbf{\MC} & 8,438 & 2,107,260 & 3,065 & 1,162 & 99,284\\
    \textbf{\sysnative} & 9,136 & 1,807,823 & 3,016 & 1,174 & 92,488 \\
    \textbf{\sys} & 7,537 & 1,621,243 & 3,029 & 1,162 & 68,617 \\
    \hline
    \end{tabular}
    }
    \vspace{-1em}
    \caption{Effects of adaptive training.}
    \label{tbl:effect-adaptive-training}
    
    \vspace{1.5em}
    
    \includegraphics[width=0.48\textwidth]{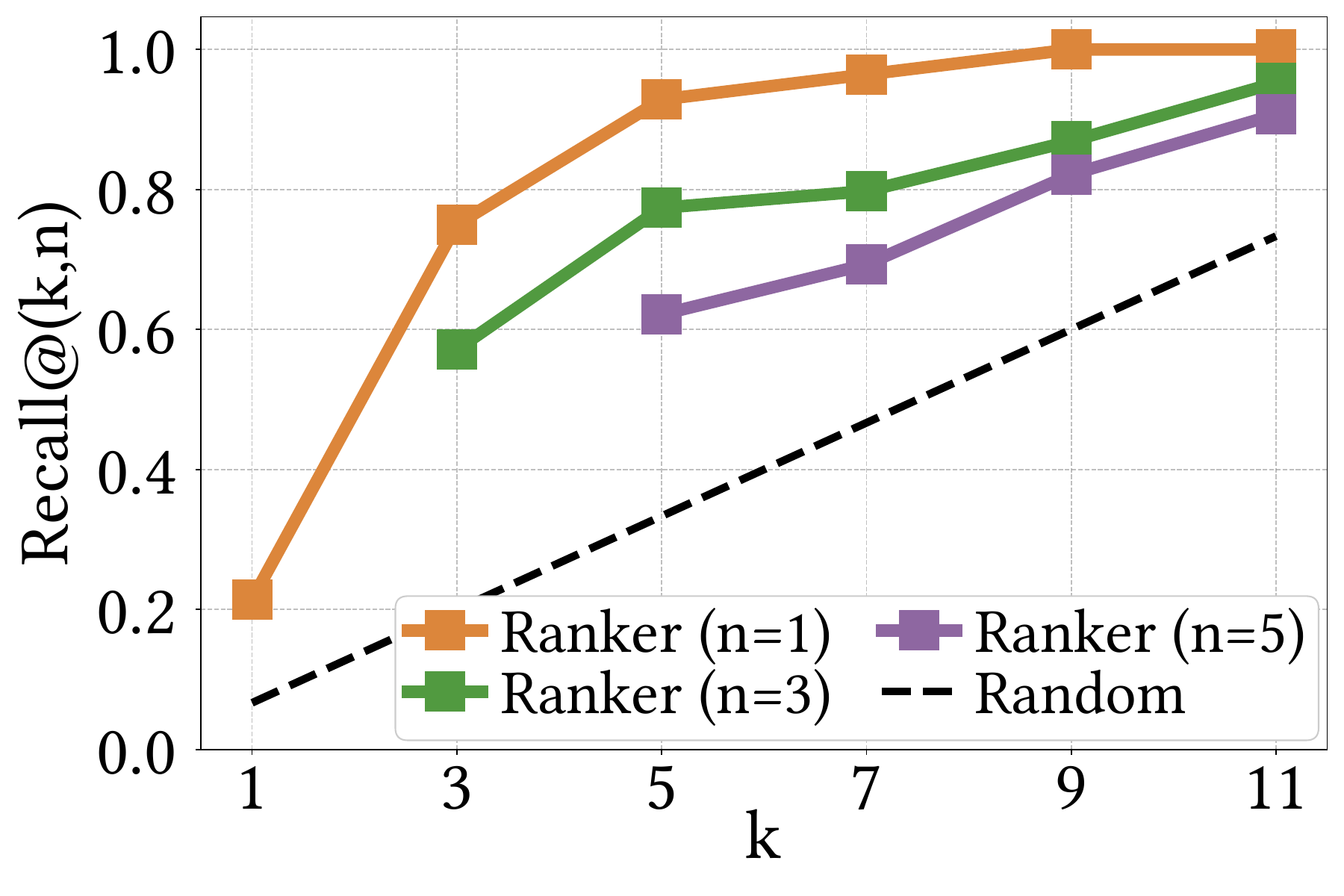}
    \hfill
    \includegraphics[width=0.48\textwidth]{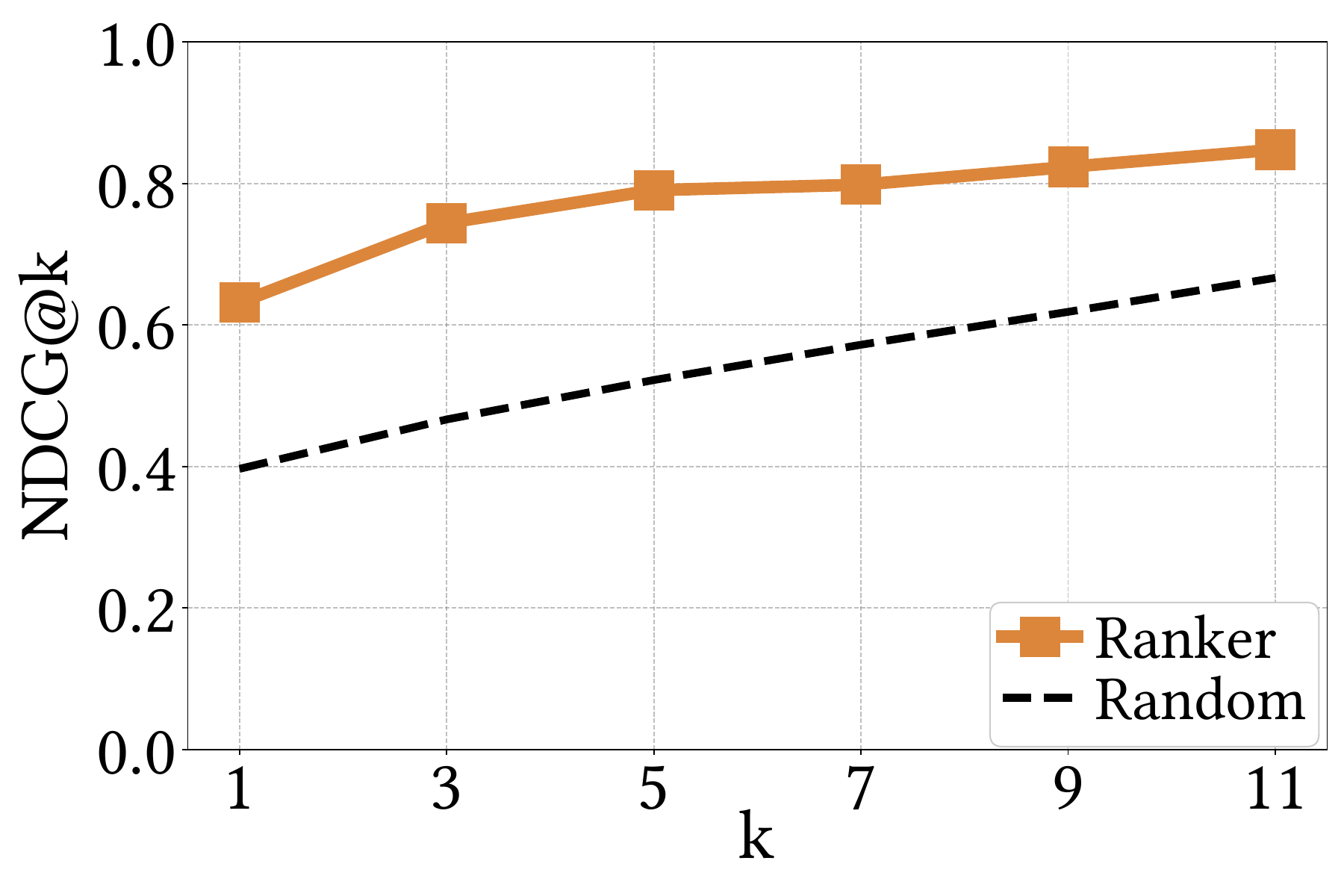}
    
    \hspace{0.5cm}{\small \textsf{(a) Recall \wrt $k$.}}
    \hfill
    {\small \textsf{(b) NDCG \wrt $k$.}}
    \vspace{-1em}
    \caption{Performance of \textsf{Ranker}.}
    \label{fig:rank_performance}
\end{minipage}
\vspace{-0.8em}
\end{figure*}

\subsubsection{\textbf{Effectiveness of Adaptive Training}}\label{subsec:grl-evaluation} 

We conduct an ablation study to evaluate the effectiveness of the adaptive training process. Figure~\ref{tbl:effect-adaptive-training} compares the performance of \sys with a variant without adaptive training, denoted as \sysnative. 
\sysnative is constructed by removing the domain classifier and the gradient reversal layer (GRL) from \sys and is trained solely to minimize the cost prediction loss $L_c(\cdot)$ in Eq.~\eqref{eq:loss}. Our experimental results reveal the following findings.

On \textsf{Projects~1}, \textsf{2}, and \textsf{5}, removing adaptive training causes pronounced performance degradation: \sysnative performs markedly worse than \sys and often comparable to, or even worse than, \MC. This exposes the severe impact of the distribution shift between default plans and knob-tuned candidate plans on learned optimizers’ performance. With adaptive training, \sys avoids overfitting to default plans and generalizes well to estimate costs for candidate plans, confirming the indispensable role of adaptive training. On \textsf{Projects~3} and \textsf{4}, \sysnative and \sys perform similarly, which is consistent with the findings in Sections~\ref{subsec:e2e-evaluation} and~\ref{subsec:per-query-evaluation}, where any learned optimizers struggle in cases with limited improvement space and insufficient training data.

Beyond the observed CPU cost savings, adaptive training enables \sys to be trained on historical query plan costs alone, greatly simplifying the data preparation phase. This stands in sharp contrast to prior work~\cite{balsa, neo, bao, lero}, which requires executing numerous candidate plans to collect their costs for training purposes.

\subsubsection{\textbf{Impact of Training Data Size}}
\label{subsec:training-size-effect}

In this experiment, we examine how the volume of training data affects \sys's performance. We incrementally vary the size of the training set from 1,000 to $\mathsf{MAX}$ queries, where $\mathsf{MAX}$ is the maximum number of historical queries available within 25 days, and evaluate the performance of \sys trained on each set. Figure~\ref{fig:training_size} reports the results, which reveal the following two key insights: 


1) On \textsf{Projects~1}, \textsf{~2}, and \textsf{~5}, the performance of \sys generally improves with more training data and eventually stabilizes. This behavior suggests that: i) insufficient training queries limit the model's benefits, and incorporating more training data offers clear advantages. This finding motivates the rules in Section~\ref{sec:project-selection} that filter out projects without enough historical queries for training \sys; and ii) these advantages are ultimately bounded. For example, on \textsf{Project~1}, continuously expanding the training set from 10,000 to 20,000 queries yields only marginal performance gains. 


We also notice an inherent performance gap between \sys and the best-achievable model (the dashed line in the figure), regardless of how we expand the training set. This arises because our project selection operates at the project level and excludes projects with low average improvement space, yet each retained project still contains queries with limited room for improvement. Enlarging the training set does little to improve \sys's capability in distinguishing the minor cost differences among candidate plans for such queries. This indicates that we could find a sweet spot for \sys that best balances training overhead against the performance gains.



2) 
\sys's performance varies significantly across projects, and each requires a distinct minimum number of training queries to match \MC. On \textsf{Project~1}, \sys surpasses \MC only after training on more than 6,000 queries, whereas on \textsf{Projects~2} and~\textsf{5}, \sys consistently outperforms \MC at all sampled training sizes. This comes from two factors. First, \textsf{Project~1} exhibits far less improvement space ($25\%$) than \textsf{Projects~2} and~\textsf{5} (over $40\%$). This makes it harder for \sys to distinguish and select superior plans when many candidates have similar estimated costs. Moreover, \textsf{Project~1} has about five times as many columns as \textsf{Project~2}, so \sys needs more training queries to learn the data distributions effectively. On \textsf{Projects~3} and \textsf{~4}, \sys's performance remains relatively stable and shows no notable difference from \MC, as discussed in Section~\ref{subsec:e2e-evaluation}.

\subsubsection{\textbf{Effectiveness of Our Practical Method for Cost Inference}}
\label{subsec:load-effect} 

To evaluate our practical solution to the invisible execution environments during plan cost inference, we compare \sys with three variants:
a) \textbf{\sysprefix\textsf{CE}} models cluster-wide environmental features as random variables, where each feature is computed as the average across all machines in the cluster. Using cluster-wide environmental data collected over the past 24 hours, a distribution is fitted, and its expected feature values are used for cost prediction; 
b) \textbf{\sysprefix\textsf{CB}} also leverages cluster-wide environments but uses the environment instance at the moment of query optimization for cost estimation; and
c) \textbf{\sysprefix\textsf{NL}} removes all environmental features from \sys and operates without using any information about machine load during both training and cost prediction. 

We evaluate these methods in terms of both the end-to-end CPU cost and their deviance from the oracle model (as detailed in Section~\ref{sec:per-bound}). 
For ease of comparison, the relative deviance, calculated as the deviance divided by the oracle model's expected cost, is reported. Figure~\ref{fig:error-wrt-load} presents the result. We observe that: 

1) \sysprefix\textsf{NL} consistently performs worse than \sys in both query execution cost and relative deviance, particularly on \textsf{Project~1}. This underscores the significant noise introduced by execution environments in query optimizations and the importance of incorporating environmental features
into plan cost modeling, as we highlighted in Challenge~1 and Section~\ref{sec:per-bound}.

2) \sys outperforms the variants \sysprefix\textsf{CE} and \sysprefix\textsf{CB} across all projects, especially on \textsf{Projects~1} and~\textsf{5}. 
This confirms the practical effectiveness of \sys's strategy that estimates plan costs under average-case environmental conditions. We attribute its advantage over using cluster-wide environmental features to two factors. First, at execution time, query tasks are typically scheduled to machines with more idle resources (\eg lower system load) for load balancing and efficient execution. Second, queries from different projects exhibit distinct characteristics and induce different load patterns on the machines. 
Consequently, cluster-wide averages often fail to reflect the machine-level load a given query experiences. In contrast, the expected machine-level environmental conditions observed from historical queries, as adopted in \sys, provide a closer approximation. 


3) 
The best-achievable model incurs a relative deviance of about $10\%$ across all test projects. This finding provides empirical support for Theorem~\ref{Thm:perfgap} and highlights the intrinsic performance gap between any query optimizer and the oracle model. Moreover, we observe that relative deviance not only offers richer practical insights but also mirrors the behavior of the widely used end-to-end execution cost metric, which is exactly the ultimate goal for query optimization. Taken together, relative deviance is promising to serve as a novel and effective metric for evaluating query optimizers in real-world systems such as \MC in future research.

\subsubsection{\textbf{Effectiveness of \textsf{Ranker}}}
\label{subsec:eval-ranker}

We evaluate the performance of \textsf{Ranker} using 28 projects drawn from the 30 sampled projects described in Section~\ref{subsec:qo-setup}, with 2 excluded for internal privacy reasons. For each project, we collect a workload consisting of historical queries from 5 consecutive days and compute their improvement space (\ie, $D(M_d)$) for training and testing \textsf{Ranker}. 
To obtain robust results, we conduct cross-validation across multiple splits of the 28 projects. In each split, 13 projects are used for training, and the remaining 15 projects are evaluated by \textsf{Ranker}, which estimates their improvement space and produces a ranking. 


As \textsf{Ranker} is designed to prioritize projects with large improvement space, we evaluate it using two widely adopted ranking metrics: i) $Recall@(k,n)$~\cite{recall}, which measures the fraction of $n$ ground-truth projects (those with the highest improvement space) included in the top-$k$ projects prioritized by \textsf{Ranker}; and ii) $NDCG@k$~\cite{NDCG}, which assesses the quality of the top-$k$ ranking by additionally considering the positional relevance of projects compared to the ground-truth ranking.

To benchmark \textsf{Ranker}'s performance, we compare it against the expected performance of a random ranking model, denoted as \textsf{Random} below, which generates a uniform random permutation of all test projects. The derivations of the expected $Recall@(k,n)$ and $NDCG@k$ for \textsf{Random} is detailed in Appendix~\ref{app:random-rank-model}~\cite{full-paper}. Figure~\ref{fig:rank_performance} presents the results averaged across all test splits. We observe that \textsf{Ranker} consistently and substantially outperforms \textsf{Random} across different $k$ values on both metrics. This superior performance is attributed to \textsf{Ranker}'s ability to effectively capture patterns in default plans that may reveal optimization opportunities. Notably, among the $5$ projects used for evaluating \sys's query optimization performance, \textsf{Ranker} successfully identifies approximately $3$ (\ie, \textsf{Projects~1, 2, 5} with substantial improvements) out of the top-$5$ projects it prioritized. This underscores its practical effectiveness in pinpointing projects that would benefit most from \sys's deployment. Beyond that, \textsf{Ranker} is expected to achieve further performance gains with more training projects (see Appendix~\ref{subsec:eval-ranker-size}~\cite{full-paper}).

\subsection{\textbf{Benefits in \MC}}
\label{subsec:mc-benefit} 

To understand the practical applicability of \sys and its overall benefits to \MC, we analyze the proportion of \MC's projects that would significantly benefit from \sys. Specifically, we require a reduction in end-to-end CPU cost of at least $10\%$ compared to the native \MC. In a mature production system like \MC, where the query optimizer is already highly optimized, it represents a meaningful headroom for continuous optimization. Given the massive scale of \MC's computational resource consumption, such relative reductions could translate into substantial absolute cost savings.

Recall that we apply \textsf{Filter} to \MC's projects to filter out those likely to face training challenges and obtain a subset of $40.5\%$ of projects that pass all filtering rules. In our earlier experiments in Section~\ref{subsec:e2e-evaluation}, we observed that \sys achieves a notable performance gain ($\ge 10\%$) on approximately $10\%$ of the randomly selected $30$ projects \ie, \textsf{Projects~1, 2}, and \textsf{5}. Thus, we can reasonably estimate that at least $4\%$ ($40.5\% \times 10\%$) of projects in \MC are expected to see a performance gain $\ge 10\%$ by deploying \sys, which is a non-trivial proportion considering \MC's currently vast landscape of more than 100,000 projects.


Notably, this estimation is restricted by current plan exploration strategies, which are highly conservative and may fail to produce candidate plans that substantially outperform default plans. For instance, for queries in \textsf{Project~1}, their optimal candidate plans achieve only an average of $15\%$ improvement over default plans. This limits the potential of query optimization and thereby curtails \sys's overall performance. However, the estimated value could be substantially improved by incorporating more diversified plan exploration strategies.

\vspace{-0.1em}
\section{Conclusions}
\label{sec:conclusion} 

This paper presents \sys, a one-stop framework for \MC that bridges the gap between theoretical advances in learned query optimization and practical deployment in modern distributed, multi-tenant data warehouses. \sys rethinks the architecture design and system operating paradigms for developing learned optimizers, and is generalizable and readily adaptable to other similar leading systems. Principally, \sys employs a statistics-free plan encoding to handle the potential absence of input statistics in \MC. It also integrates execution-time environments
into the modeling of plan costs. For online queries, it provides both a theoretical quantification of achievable performance under unseen environment variations and a practical smoothing strategy to mitigate environmental impacts. For system operating, \sys relies on an adaptive training process to preemptively generalize to candidate plans during the offline training phase, which eliminates risky online refinement while also incurring no extra cost for executing candidate plans. Furthermore, \sys includes a lightweight project selector to prioritize high-benefit deployments. From our experiments on production workloads in \MC, \sys is expected to see great performance gains over the native query optimizer on a non-trivial portion of \MC's projects, which could translate to substantial real-world resource savings.



\clearpage
\balance
\bibliographystyle{ACM-Reference-Format}
\bibliography{sample}

\newpage
\nobalance

\clearpage
\appendix

\noindent{\LARGE \textbf{APPENDIX}}
\section{\MC V.S. Similar Systems}~\label{app:mc-vs-similar-systems}

\noindent Leading industrial distributed multi-tenant data platforms, such as \textsf{Google BigQuery}~\cite{bigquery}, \textsf{Microsoft Azure Synapse Analytics}~\cite{azure_synapse_analytics}, \textsf{Snowflake Serverless}~\cite{snowflake}, and \textsf{Amazon Redshift Serverless}~\cite{redshift}, differ in their design philosophies and ecosystem-specific optimizations. Nonetheless, they adhere to a query execution workflow similar to \MC, as depicted in Figure~\ref{fig:workflow}, and face common challenges in deploying learned optimizers effectively. 

\stitle{Common Query Execution Workflow.} When queries are submitted to such systems, they are first converted into optimized physical plans via native cost-based optimizers built on expert-crafted cost models and optimization rules. For ease of scheduling and efficient execution, these plans are then expressed as a coarse-grained tree (or more generally, a DAG), where the vertices represent intra-machine pipelines that function as \emph{execution units}, akin to stages in \MC. These systems support distributed execution through distinct abstractions of computational resources and distribution mechanisms. For instance, \textsf{BigQuery} employs the \textsf{Dremel}~\cite{dremel} engine that allocates resources in units called ``slot'' (typically half of a CPU core and $1$ GB of RAM) and \textsf{Snowflake}'s engine relies on an elastic compute cluster model allowing for auto-adjusting resources during execution. Despite these differences, resources are always allocated dynamically to each unblocked execution unit, and greatly impact the execution performance. Upon completion of all executions, relevant information is logged for future reference and analysis.


\stitle{Common Challenges with \MC.} As a common practice in multi-tenant systems, users maintain their own database instances and submit workloads with varying characteristics. Given the numerous user workloads and our observation that learned optimizers struggle to equally benefit all workloads, how to select suitable workloads for employing learned query optimizers in an automatic and efficient manner must be initially addressed, akin to Challenge~4 (\emph{project selection}) in \MC. 

When developing learned query optimizers for given user workloads, the issue of missing helpful input statistics, corresponding to Challenge 2 (\emph{missing helpful statistics}) in \MC, remains a significant barrier. In such systems, the sheer scale of production data (\eg, ranging from terabytes to petabytes) and frequent data modifications make it overwhelmingly challenging to maintain statistical information (\eg, histograms of attributes) in real time. As alternatives, these systems rely on manual triggers to update statistics or on-demand calculations of statistics based on sampling during query optimization. However, manual triggers often fail to ensure that the statistics are timely and reflect the underlying data distributions. Meanwhile, the statistics computed using a controlled volume of sampled records (\eg, fewer than several thousand records~\cite{redshift}) suffer from high variance and are far from accurate. In addition to this challenge, the high variance in a query's execution costs is also inevitable in these systems due to dynamic resource scheduling and distributed computing, so Challenge 1 (\emph{environment variation}) also manifests.
 

Finally, in system operating, due to the resource-intensive nature of the OLAP queries they target and the stringent requirements for high stability and reliability in production-scale systems, the conventional refinement process for learned optimizers, which relies on either additional computational resources to execute candidate plans or the deployment of a partially trained optimizer to serve user queries while undergoing long-term online refinement, is totally unacceptable. Therefore, Challenge~3 (\emph{infeasibility of conventional refinement}) in \MC must also be addressed in such systems. 

The commonalities between these systems and \MC present the opportunity for cross-system transferability: the techniques in \sys developed for \MC, along with the related experiences and insights introduced in this paper, can serve as highly valuable references and are readily adaptable to other systems. This greatly amplifies the practical contribution of \sys and positions \sys as a foundational framework capable of driving advancements in learned query optimizations across a wide range of industrial distributed multi-tenant platforms.


\section{Supplements to Adaptive Cost Predictor}

This section introduces the missing details in our adaptive cost predictor design. 

\subsection{Hash Encoding For Table and Column Identifiers}~\label{app:hash-encode}

\noindent In \MC's production workloads, the number of tables and columns can be extremely large (see Table~\ref{tbl:dataset}), making it impractical to use one-hot (or multi-hot) encodings due to the potential for dimensional explosion. To address this, \sys builds upon the classic hash encoding technique, which efficiently maps an unbounded set of identifiers into a fixed-size feature space with controlled dimensionality. Below, we focus on how table identifiers are encoded; column identifiers are handled analogously.

Let $\mathcal{T}$ denote the space of table identifiers, and let $N$ represent the feature dimension. Formally, standard hash encoding applies a hash function $f: \mathcal{T} \rightarrow \mathbb{N}$ to map a given table identifier $T$ to an integer, which is then converted into an $N$-dimensional binary vector, where only the $(f(T) \% N)$-th position is set to $1$.

\begin{figure}[tp]
\centering
\includegraphics[width=\linewidth]{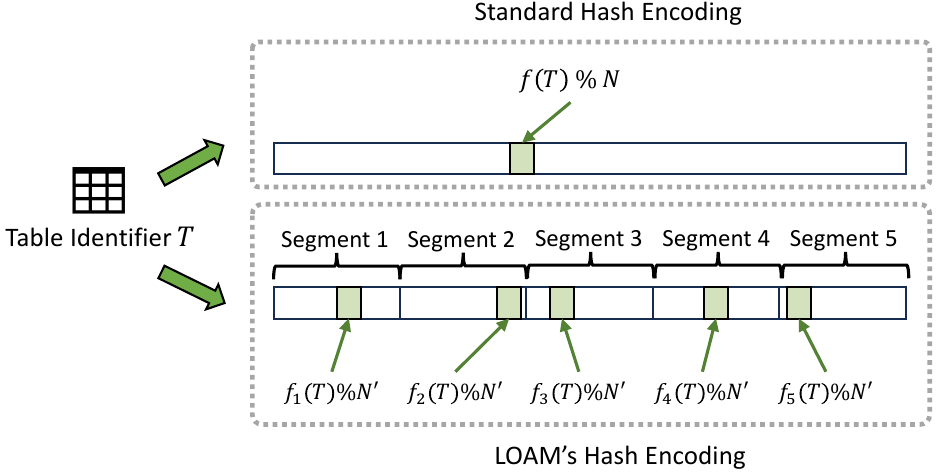}
\vspace{-0.5em}
\caption{Illustration of table identifier encoding in \sys.}
\label{fig:hash}
\vspace{-1.2em}
\end{figure}

However, this standard hash encoding suffers greatly from hash collisions, which limits its scalability. For example, using a native $50$-dimensional hash encoding can uniquely represent only around $50$ distinct tables before collisions become highly probable. To address this limitation, \sys introduces a substantially improved variant: as illustrated in Figure~\ref{fig:hash}, each table identifier is encoded into a $5 \times N'$-dimensional vector, where each $N'$-dimensional segment is generated by applying an independent hash function $f_i(T)$ to the table identifier $T$ and setting the $(f_i(T) \% N')$-th position to $1$. By incorporating multiple hashing functions across segments, this method dramatically reduces the probability of collisions without compromising the feature compactness. For instance, with $N'=10$, a total of $50$-dimensional encoding can reliably encode up to $10^5$ distinct table identifiers, far exceeding the capability of the native approach. Moreover, our method naturally extends to support encoding multiple identifiers simultaneously (\eg, all columns in a filter operator) by taking the union of their respective encodings.



\subsection{Metrics for Machine Loads}~\label{app:load-metric}

\noindent We elaborate on the four metrics used in \sys' plan encoding to model machine loads:

\begin{itemize}[left=0pt]
    \item \texttt{CPU\_IDLE}: The percentage of time the CPU remains idle and available for new tasks. 
    \item \texttt{IO\_WAIT}: The percentage of CPU time spent waiting for I/O operations to complete.
    \item \texttt{LOAD5}: The average system load (\ie, the number of processes using or waiting for resources like CPU and disk) over the past five minutes. 
    \item \texttt{MEM\_USAGE}: The percentage of memory currently in use.
\end{itemize}

\section{Proof for Theorem~\ref{Thm:perfgap}}~\label{app:proof-for-perf-bound}

\begin{proof}(Theorem~\ref{Thm:perfgap}) 
We first prove that for any model $M$, we have $\mathbb{E}[D_E(M)] \geq \mathbb{E}[D_E(M_b)]$. By definition, $M_b$ selects the plan $P_{M_b}$ that minimizes the expected cost $\mathbb{E}[C_E(P)]$. Therefore, we have $\mathbb{E}[C_E(P_M)] \ge \mathbb{E}[C_E(P_{M_b})]$. Subtracting $\mathbb{E}[C_E(P_{M_o})]$ from both sides gives $\mathbb{E}[C_E(P_M) - C_E(P_{M_o})] \ge \mathbb{E}[C_E(P_{M_b}) - C_E(P_{M_o})]$, which directly implies $\mathbb{E}[D_E(M)] \ge \mathbb{E}[D_E(M_b)]$.  

To prove $\mathbb{E}[R_{E}(M_b)] > \mathbb{E}[R_{E}(M_o)] = 0$, we note that $M_o$ always selects the plan $P_{M_o}$ that minimizes the cost under any environmental instance $e$ and attains a deviance value of $0$. It thereby holds that $C_e(P_{M_b}) \ge C_e(P_{M_o})$. Taking the expectation over all environments, we have $\mathbb{E}[R_E({M_b})] =  \mathbb{E}[C_E(P_{M_b}) - C_E(P_{M_o})] \ge 0 = \mathbb{E}[D_E(M_o)]$. Overall, the theorem holds. 
\end{proof}

\section{Supplements to Project Selector}

This section introduces the missing details in our project selector design.

\subsection{Rules in \textsf{Filter}}~\label{app:more-metrics}

\noindent In the current \textsf{Filter}, we restrict projects from both \emph{the volume of training queries} and \emph{the stability of the underlying tables} to mitigate potential training challenges. Specifically:

1) Experiments in Section~\ref{subsec:training-size-effect} have demonstrated that the volume of historical queries used to train the adaptive cost predictor strongly impacts the query optimization performance. To ensure a project has sufficient historical queries for robust model training, we define the following two metrics on the sampled workload $\mathcal{Q}$ consisting of historical queries collected over $d$ consecutive days: 

\begin{itemize}[left=0pt]
    \item $\mathsf{n\_query}(\mathcal{Q}) = \frac{|\mathcal{Q}|}{d}$. This metric calculates the average number of queries submitted per day. 
    \item $\mathsf{query\_inc\_ratio}(\mathcal{Q}) = \frac{1}{d-1}\sum_{i=2}^d\frac{|\mathcal{Q}_i|}{|\mathcal{Q}_{i-1}|}$, where $\mathcal{Q}_i \subseteq \mathcal{Q}$ represents the set of queries submitted on the $i$-th day ($1 \le i \le d$). This metric measures the growth rate of query counts over consecutive days.
\end{itemize}

A value of $\mathsf{query\_inc\_ratio}(\mathcal{Q})$ around or greater than $1$ indicates that the number of submitted queries is stable or growing steadily, making $\mathsf{n\_query}(\mathcal{Q})$ a reliable indicator of the general query volume. Based on these two metrics, we collaboratively apply two rules: (\textsf{R1}) $\mathsf{n\_query(\cdot)} \ge N_0$; and (\textsf{R2}) $\mathsf{query\_inc\_ratio(\cdot)} \ge r$. As observed in Section~\ref{subsec:training-size-effect}, a training set containing at least $10,000$ queries typically suffices to achieve good performance for query optimization. Therefore, in our implementation, we set $N_0 = 2, 000$ and $r$ to the minimum number such that $N_0 \times r^{30} \ge 10,000$.

2) Since the statistical-free plan encoding adopted by \sys relies on operator semantics and historical queries involving the same tables to coarsely infer data-level details, frequent creation and removal of tables within a project, \eg, temporarily created tables for data analysis purposes, can hinder the ability of \sys to effectively leverage insights from historical queries for predicting plan costs of future queries. For better predictability, we encourage workloads to involve more long-living tables, \ie, those with a lifespan exceeding $n$ days. To achieve this, we introduce the following metric to evaluate the proportion of queries involving such long-living tables: $\mathsf{stable\_table\_ratio}(\mathcal{Q}) = \frac{|\{Q \in \mathcal{Q}|\forall t \in T_Q, \mathsf{LifeSpan}(t) > n \}|}{|\mathcal{Q}|}$, where $T_Q$ denotes the set of tables involved in query $Q$, and $\mathsf{LifeSpan}(t)$ represents the lifespan of table $t$. Based on this metric, the rule (\textsf{R3}) $\mathsf{stable\_table\_ratio} \ge \theta$ is incorporated in \textsf{Filter}. In the current implementation, we set $n = 30$ and $\theta = 0.2$. 

Notably, these rules can be dynamically modified, added, or removed to better align with the evolving characteristics of \sys, user workloads, and system requirements. This flexibility guarantees that \textsf{Filter} can remain adaptable and effective over time.

\subsection{Plan Vectorization for \textsf{Ranker}}~\label{app:feat-for-ranker}

\noindent By our observation in Section~\ref{sec:project-selection} that \emph{the improvement space (\ie, $D(M_d)$) of a query $Q$ can often be reflected in the observable properties of $Q$'s default plan $P_d$, including its structure, input characteristics, and execution cost}, \textsf{Ranker} represents $Q$ by efficiently encoding these aspect as follows: 

1) \emph{The structure of $P_d$}. We adopt a $(1 + n)$-dimensional vector to represent the structural information of $P_d$, where the first position indicates the total number of operators in $P_d$, and the remaining $n$ positions encode the structural details. Specifically, we view each distinct combination of parent and child operators $\langle \mathsf{parent}, \mathsf{child} \rangle$ as a pattern and encode the count of occurrences for each pattern. For example, as illustrated in Figure~\ref{fig:ranker}, the second and third positions represent the counts for patterns $\langle \mathsf{MJ}, \mathsf{TS}\rangle$ and $\langle \mathsf{HJ}, \mathsf{MJ}\rangle$, with values $2$ and $1$, respectively. This approach simplifies the representation of the overall plan structure while still highlighting key structures that may reveal optimization space, such as nested joining operators. It's worth noting that this encoding is much more informative than simply counting operators. For instance, $\#\langle \mathsf{HA}, \mathsf{MJ} \rangle = 1$ might suggest potential optimization by reversing the order for the $\mathsf{MJ}$ and $\mathsf{HA}$ operations, which cannot be reflect from encodings like $\{\#\mathsf{HA}: 1, \#\mathsf{MJ}:1\}$.

\begin{figure}[t]
\centering
\includegraphics[width=1\linewidth]{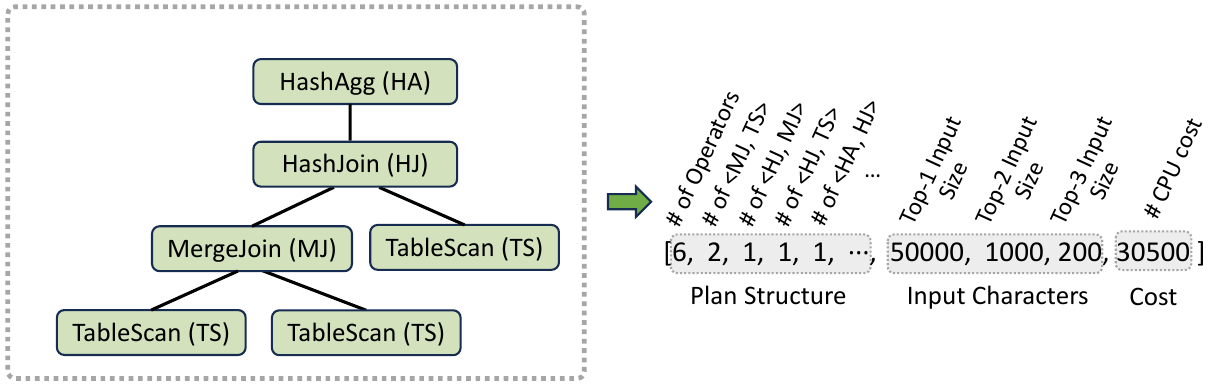}
\vspace{-0.5em}
\caption{Illustration of plan vectorization for \textsf{Ranker}.}
\label{fig:ranker}
\vspace{-1.2em}
\end{figure}

2) \emph{The sizes of input tables}. We represent the sizes (\ie, the number of records) of tables involved in $P_d$ using a vector, with each position corresponding to the size of a certain table. This enables \textsf{Ranker} to recognize cases where input size skew may create optimization opportunities, \eg, when a large table joins a small table, steering the native query optimizer to produce predicates from the smaller table to filter the larger one before performing the join would significantly improve the join performance. Since the number of involved tables varies across queries, our implementation considers only the top-$3$ largest table sizes using a $3$-dimensional vector. 

3) \emph{The cost of $P_d$}. We also employ a single value to represent the CPU cost of $P_d$, which may indicate suboptimal plan choices like poor joining orders and thereby also reveal optimization opportunities, as discussed in Section~\ref{sec:project-selection}.

To ensure consistent scaling, all feature values are normalized using min-max normalization. 

\section{Supplements to Experiments}
\balance

This section introduces the missing details for our experimental evaluation. 

\subsection{Deviance Estimation}~\label{app:proof}

\noindent To estimate the deviance for a given model $M$, we first derive its expectation theoretically based on the distributions of plan execution costs, and then describe how to estimate it practically through sampling and distribution fitting. 

\stitle{Theoretical Derivation.} Recall that in Section~\ref{sec:per-bound}, for a query $Q$ with candidate plans $\mathcal{P}_Q = \{P_1, P_2, \cdots, P_n\}$, the deviance of $M$ on $Q$ \wrt environment $E$ is defined as $D_E(M) = C_E(P_M) - C_E(P_{M_o})$. 

Suppose we are already given the distribution of the execution cost $C_E(P)$ for each plan $P \in \mathcal{P}$. For simplicity of notation, we denote the probability density function (PDF) and cumulative distribution function (CDF) of the random variable $C$ for execution costs as $f_C$ and $F_C$, respectively. To derive the expectation of $D_E(M)$, we start by reformulating $D_E(M)$ as below: 
\begin{align*}~\label{eq:re}
        D_E(M)=
        \begin{cases}
        C_E(P_M) - C_E^*& \text{if  } C_E^*\leq C_E(P_M) \\
        0 & \text{else} 
        \end{cases}, 
    \end{align*} 
where $C_E^*$ is the random variable representing the minimum execution cost among the plans in $\mathcal{P}_Q \setminus \{P_{M}\}$. This formulation is motivated as follows: If $P_M$ coincides with the plan selected by the oracle model $M_o$ under some environmental instances, no deviance occurs, \ie, $D_E(M) = 0$. Otherwise, $C^*_E$ corresponds to the execution cost of $P_{M_o}$, and thereby $D_E(M) = C_E(P_M) - C_E^*$ by definition.  

Given this, the PDF for $D_E(M)$ can be expressed as:
\begin{align*}
    f_{D_E(M)}(z)=
    \begin{cases}
    f_{C_E(P_M) - C_E^*}(z) & \text{if  } C_E^* \leq C_E(P_M) \\
    0 & \text{else} 
    \end{cases}, 
\end{align*}
where $f_{C_E(P_M) - C_E^*}(z)$ is the PDF of the difference between $C_E(P_M)$ and $C_E^*$, and can be computed via convolution:
\begin{align*}
    f_{C_E(P_M) - C_E^*}(z) = \int_{-\infty}^\infty f_{C_E(P_M)}(x)f_{C_E^*}(x - z)~dx.
\end{align*} 
Thus, the expected deviance $D_E(M)$ can be calculated as: 
\begin{equation}~\label{eq:re-exp}
    \mathbb{E}[D_E(M)] = \int_{0}^{\infty} z f_{D_E(M)}(z)~dz,
\end{equation} where the integral lower bound is $0$ because $D_E(M) \ge 0$. 

Note that to compute this integral, it is necessary to know the distribution of $C_E^*$. We derive its PDF in the following lemma. 

\begin{lemma}~\label{lem:c_min} 
    The PDF of $C_E^*$ is 
     \begin{align*}
        f_{C_E^*}(x) = \sum_{C\in\mathcal{C}} f_{C}(x) \prod_{\substack{C' \in\mathcal{C}, \\ C' \neq C}}\left[1 - F_{C'}(x) \right], 
    \end{align*}
    where $\mathcal{C} = \{C_E(P_i) \mid P_i \in \mathcal{P}_Q\setminus \{P_{M}\}\}$. 
\end{lemma}

\begin{proof}
    By definition, $C_E^*=\min \mathcal{C}$. Therefore, the CDF of $C_E^*$ can be derived as: 
    \begin{align*}
        F_{C_E^*}(x) &= P(C_E^* \leq x) 
        = 1 - P\big(\bigcap_{C \in \mathcal{C}} \{C > x\}\big) \\ 
        &=1 - \prod_{C \in \mathcal{C}} P(C > x)~~~~~\text{(}C \in \mathcal{C} \text{ is independent)}\\
        &=1 - \prod_{C \in \mathcal{C}}[1 - F_{C}(x)].
    \end{align*}
    Differentiating $F_{C_E^*}(\cdot)$, the PDF of $C_E^*$ is given by 
     \begin{align*}
        f_{C_E^*}(x) &= \frac{d}{dx}F_{C_E^*}(x)\\
        &= \frac{d}{dx}\left(1 - \prod_{C \in \mathcal{C}}[1-F_{C}(x)]\right)\\
        &=\sum_{C \in \mathcal{C}}f_{C}(x)\prod_{\binom{C' \in \mathcal{C}}{C' \neq C}}[1-F_{C'}(x)].
    \end{align*}
    Thus, the lemma holds. 
\end{proof}

\stitle{Practical Estimation.} 
While $\mathbb{E}[D_E(M)]$ can be theoretically computed via Eq.~\eqref{eq:re-exp} given the execution cost distributions $C_E(P)$ for all plans $P \in \mathcal{P}_Q$, deriving the exact distributions $C_E(P)$ is infeasible in real-world systems due to numerous uncontrollable runtime factors (\eg, resource contention). Therefore, we adopt a systematic empirical approach to estimate these distributions in two stages:

\begin{figure*}[!t]
\centering
\begin{minipage}{0.47\linewidth}
    \subfigure[Histogram of execution costs.]{
    \begin{minipage}[c]{0.45\linewidth}
        \centering
        \includegraphics[width=1\textwidth]{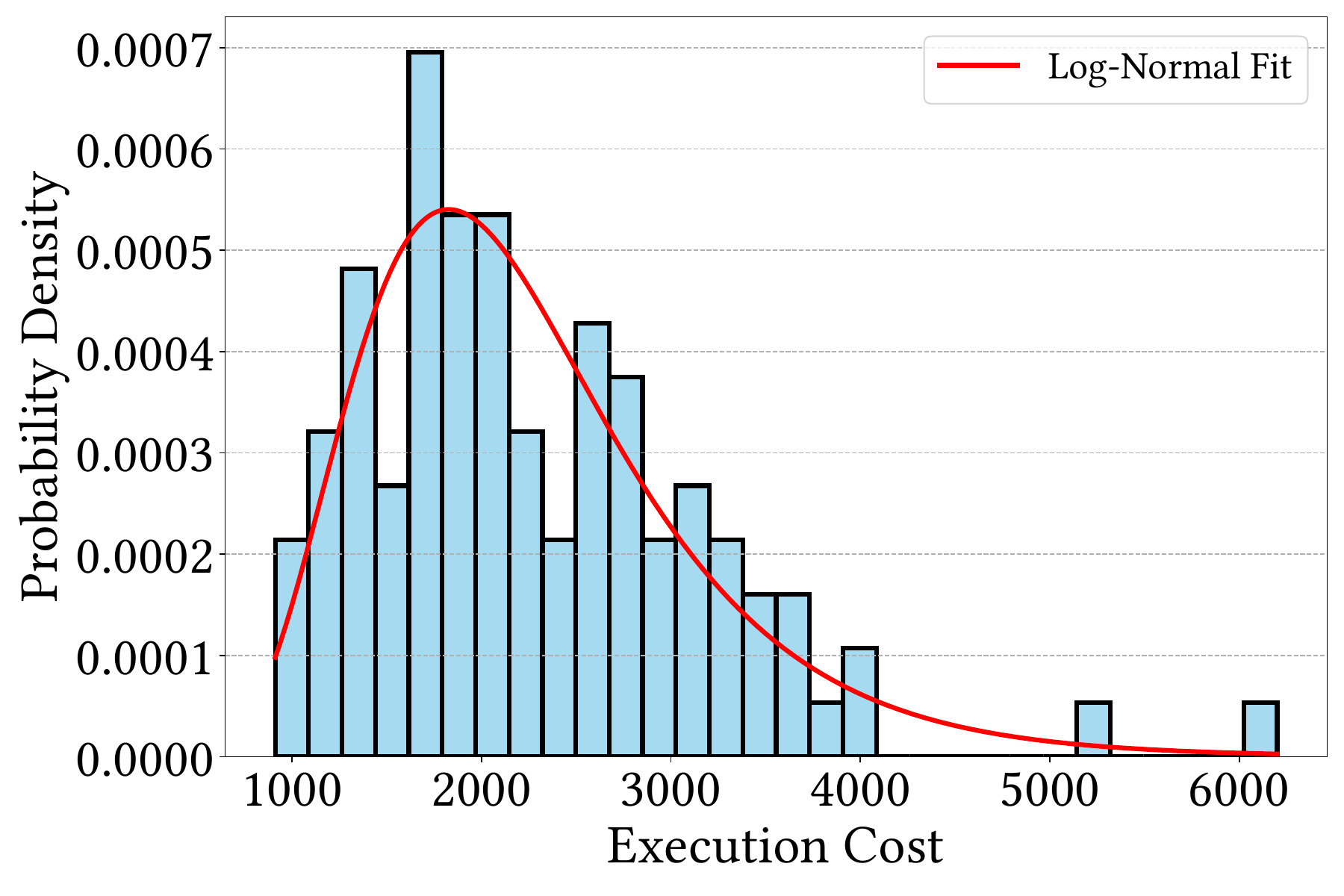}
    \end{minipage}
}\hspace{1em}
\subfigure[Q-Q plot of execution costs.]{
    \begin{minipage}[c]{0.45\linewidth}
        \centering
        \includegraphics[width=1\textwidth]{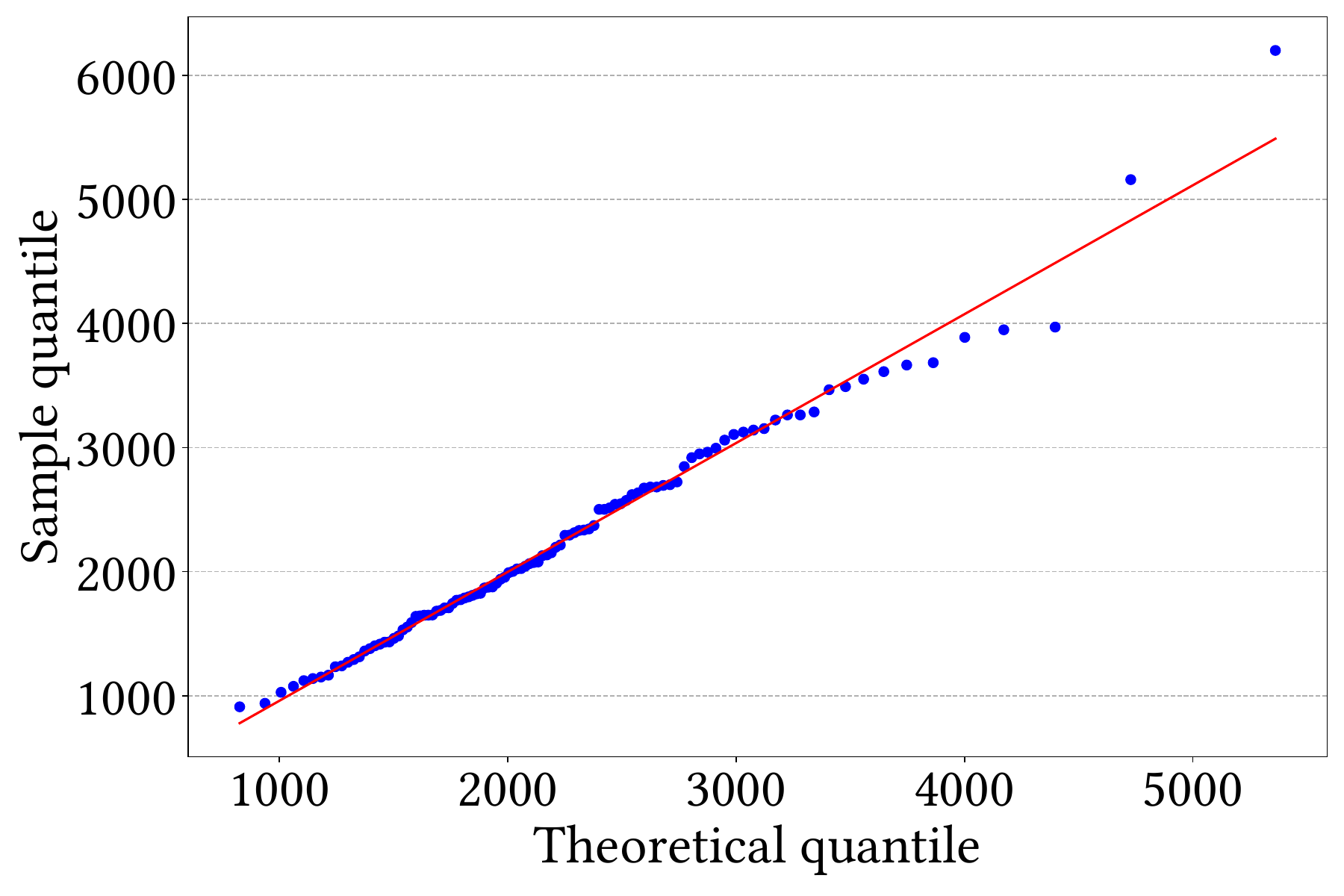}
    \end{minipage}
}
\vspace{-1.2em}
\caption{Cost distribution of an example query plan.}
\label{fig:cost_variation}
\end{minipage}
\hfill
\begin{minipage}{0.47\linewidth}
\vspace{1.05em}
\subfigure[\textsf{Recall} \wrt $k$ ($n = k$).]{
    \begin{minipage}[c]{0.45\linewidth}
        \centering
        \includegraphics[width=1\textwidth]{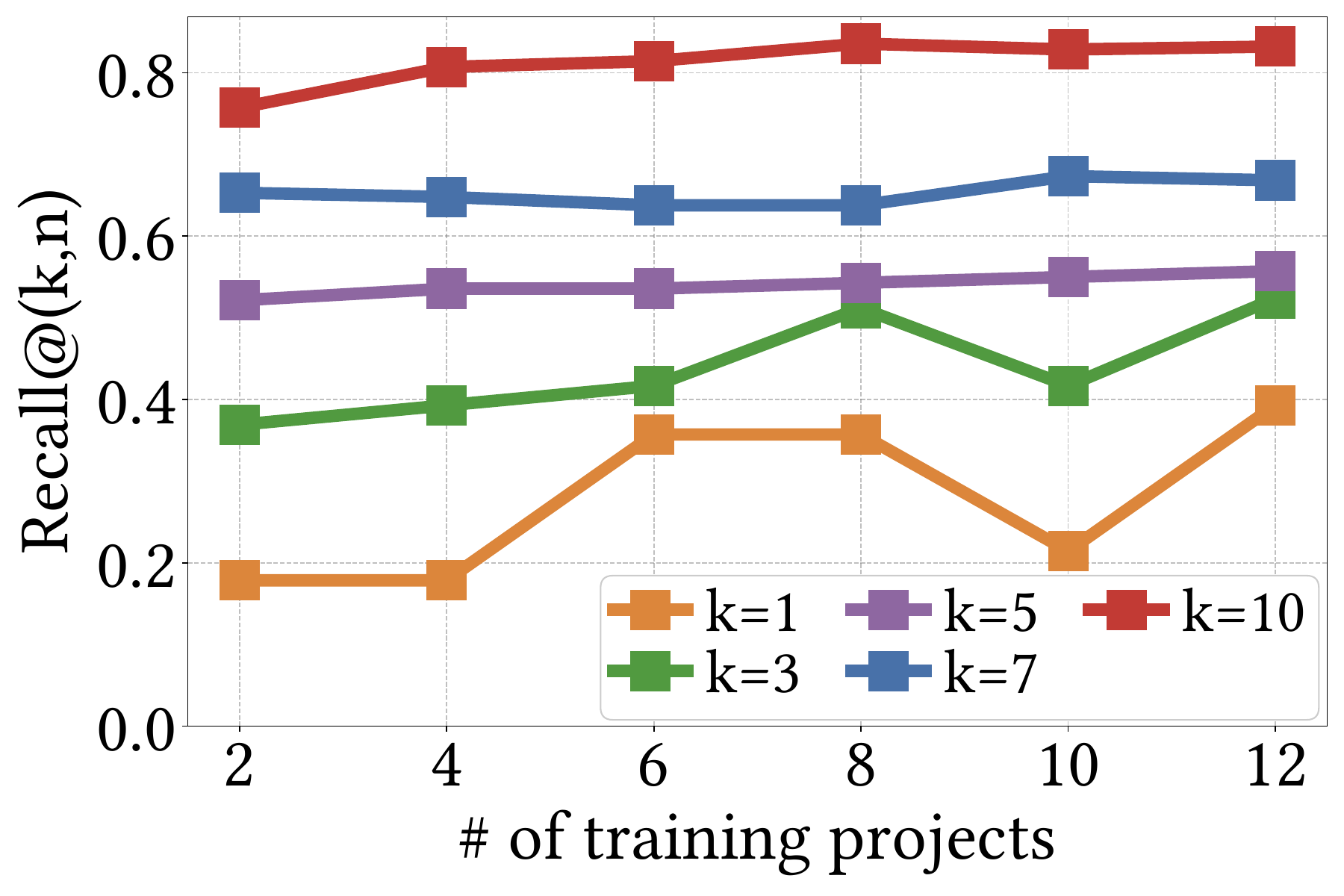}
    \end{minipage}
}\hspace{1em}
\subfigure[\textsf{NDCG} \wrt $k$.]{
    \begin{minipage}[c]{0.45\linewidth}
        \centering
        \includegraphics[width=1\textwidth]{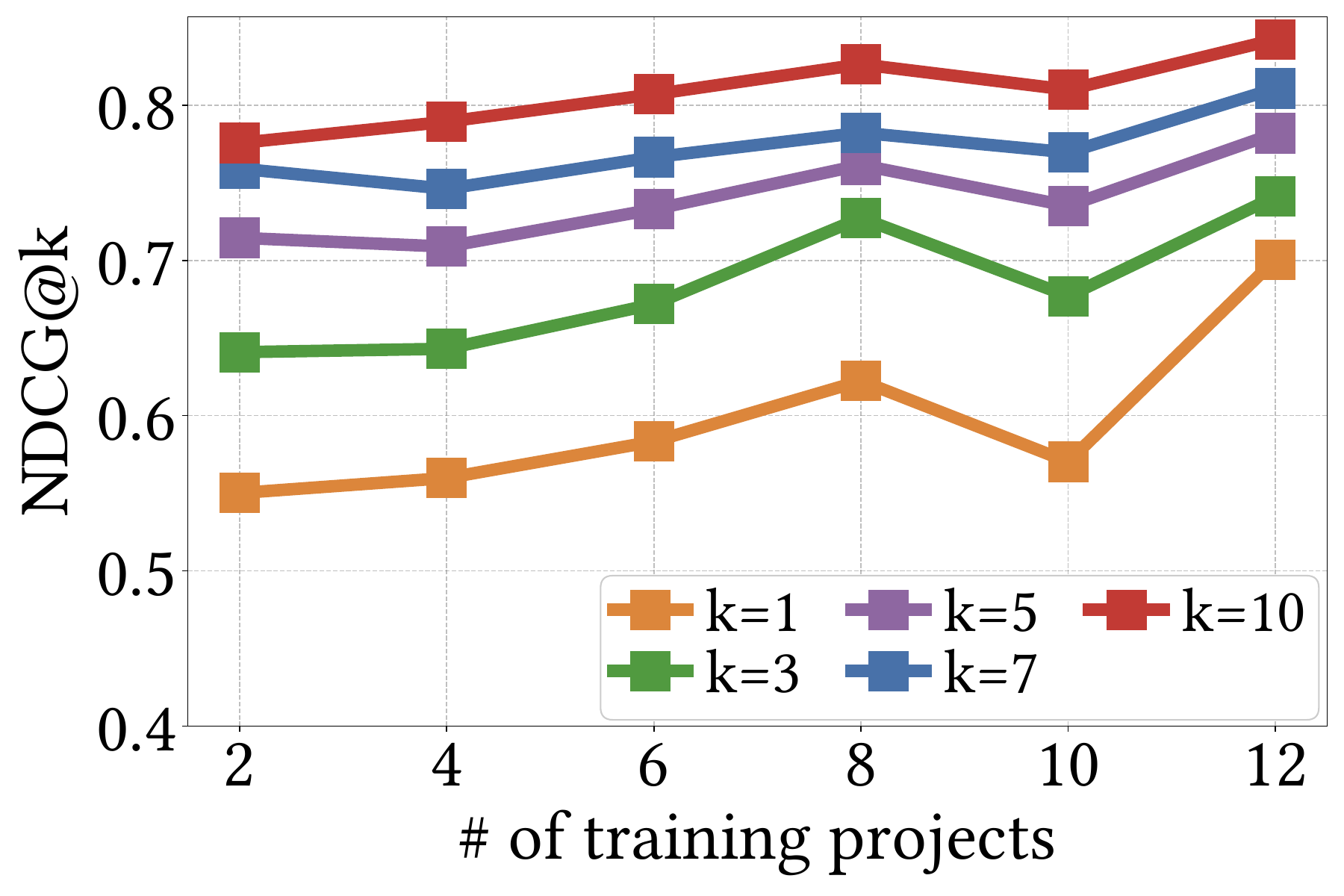}
    \end{minipage}
}
\vspace{-1.2em}
\caption{Performance of \textsf{Ranker} \wrt the number of training projects.}
\label{fig:training_size_rank}
\end{minipage}
\vspace{-1.5em}
\end{figure*}

1) \emph{Distribution modeling.} 
To empirically characterize the distribution of execution costs for query plans, we collect historical execution costs for a set of recurring queries periodically executed in \MC. Statistical analysis of these samples reveals that the execution costs exhibit a notable log-normal pattern. Figure~\ref{fig:cost_variation}(a) shows the cost histogram for a representative query plan alongside its fitted log-normal curve, while Figure~\ref{fig:cost_variation}(b) presents the corresponding Quantile-Quantile (Q-Q) plot. Both plots demonstrate strong agreement between the empirical cost distributions and the log-normal model. To formally validate this observation, we conduct a Kolmogorov–Smirnov test across all collected query plans. The average p-value of approximately $0.6$ further supports that the execution costs do not exhibit statistically significant deviations from a theoretical log-normal distribution. Following these observations and results, we model each plan's execution cost as a random variable following a log-normal distribution. It is important to note that the theoretical framework presented earlier remains applicable to any arbitrary cost distribution determined by different workload characteristics and system environments.

2) \emph{Parameter estimation.} Although the execution costs of query plans generally follow a log-normal distribution, their parameters, \eg, the mean and variation, can vary substantially across plans due to the different plan structures and influences from system environments. To estimate the specific cost distribution for each plan $P \in \mathcal{P}_Q$, we execute $P$ multiple times and use the maximum likelihood estimation approach described in \cite{DBLP:param_estimate} to fit the parameters based on the obtained execution costs.

\subsection{Ranking Metrics for a Random Model}~\label{app:random-rank-model}

\noindent In section~\ref{subsec:eval-ranker}, we compare \textsf{Ranker} with a random model \textsf{Random} that generates a uniformly random ranking (permutation) of all test projects, denoted as $r$. In the ranking $r$, each project has an equal probability of $1 / N$ of appearing at any position, where $N$ is the total number of test projects. We next derive the expectations of the metrics $Recall@(k,n)$ and $NDCG@k$ for \textsf{Random}. 

\sstitle{$\mathbb{E}[Recall@(k,n)]$}: $Recall@(k, n)$ is defined as the proportion of $n$ ground-truth projects included in the top-$k$ projects prioritized by \textsf{Random}. Since each project is equally likely to appear at any position of $r$, the probability that a specific project appears in the top-$k$ positions is given by $k/N$. Therefore, out of the $n$ ground-truth projects, the expected number of projects in the top-$k$ positions is $n\cdot k/N$. We thus have $\mathbb{E}[Recall@(k,n)] = \big(n \times \frac{k}{N}\big) / n = k / N$. 

\sstitle{$\mathbb{E}[NDCG@k]$}: $NDCG@k$ evaluates the quality of the top-$k$ ranking of $r$, denoted as $r_k$, by considering the positional relevance of projects compared to the true ranking. It is calculated as:
$$NDCG@k=\frac{DCG@k}{IDCG@k}, $$
where $$DCG@k=\sum_{i=1}^k\frac{2^{rel_i}-1}{\log_2(i+1)}$$ measures the discounted cumulative gain of $r_k$, and $rel_i$ represents the relative score (\ie, the improvement space $D(M_d)$) of the project at position $i$; $IDCG@k$ is the ideal value for $DCG@k$ of the true ranking. To compute the expectation $\mathbb{E}[NDCG@k]$, we note that $IDCG@k$ is independent of $r$ and remains constant for a fixed set of test projects, we thereby have $\mathbb{E}(NDCG@k) = \frac{\mathbb{E}[DCG@k]}{IDCG@k}$, where $\mathbb{E}[DCG@k] = \sum_{i=1}^k\frac{\mathbb{E}[2^{rel_i} - 1]}{\log_2(i+1)} = \sum_{i=1}^k\frac{\frac{1}{N}\sum_{i=1}^{N}2^{rel_i}-1}{\log_2(i+1)}$.

\subsection{Ranking Performance \wrt Training Size}~\label{subsec:eval-ranker-size}


\noindent We examine how the number of training projects affects \textsf{Ranker}'s performance. Following the experimental setup in~\ref{subsec:eval-ranker}, we implement cross-validation by partitioning the $28$ available projects into multiple test configurations. In each configuration, we fix $15$ projects for testing and vary the number of training projects from $2$ to $12$. The averaged results across all configurations are shown in Figure \ref{fig:training_size_rank}. We have the following observations: 

1) Even with minimal training data (\eg, as few as two projects), \sys achieves a robust and significant improvement over the random baseline \textsf{Random} (see Figure~\ref{fig:rank_performance} in Section~\ref{subsec:eval-ranker}). This demonstrates that \textsf{Ranker} can effectively learn to capture fundamental patterns in default plans that offer opportunities for query optimization, even in data-constrained scenarios. However, incorporating more training projects leads to substantial performance gains, \eg, $NDCG@1$ improves from $0.55$ to $0.7$ when increasing the number of training projects from $2$ to $12$. This affirms the critical role of using more training projects in enhancing \textsf{Ranker}'s performance.

2) Furthermore, despite minor fluctuations, the performance of \sys consistently improves \wrt both metrics as it is trained with more projects. This highlights the promising potential for \textsf{Ranker} to continue improving over time as additional training projects become available.

\end{document}